\documentclass[]{llncs}
\usepackage{amsmath}
\usepackage{amssymb}
\usepackage{proof}
\usepackage{stmaryrd}
\usepackage{mathpartir}
\usepackage{ cmll }
\usepackage{ tipa }
\usepackage{diagrams}
\usepackage{graphicx}
\usepackage{xcolor}
\usepackage[all,cmtip]{xy}
\usepackage{thm-restate}
\usepackage{hyperref}
\usepackage{scalerel}
\newcommand{\obseq}{\logsim}

\newcommand{\mypara}[1]{{\bf #1. }}

\newcommand{\tbcong}{\cong}

\newcommand{\outp}[2]{#1 \langle #2 \rangle}
\newcommand{\inl}{\m{inl}}
\newcommand{\inr}{\m{inr}}

\newcommand{\Ga}{\Gamma}
\newcommand{\D}{\Delta}
\newcommand{\cut}{\mathsf{cut}}
\newcommand{\cpy}{\mathsf{copy}}

\newcommand{\lb}{\llbracket}
\newcommand{\rb}{\rrbracket}

\newcommand{\llp}{\llparenthesis}
\newcommand{\rrp}{\rrparenthesis}
\newcommand{\pif}[1]{\stretchleftright{\llp}{#1}{\rrp}}

\newcommand{\m}[1]{\mathsf{#1}}
\newcommand{\ov}[1]{\overline{#1}}
\newcommand{\tensor}{\otimes}
\newcommand{\lolli}{\multimap}
\newcommand{\bang}{!}
\newcommand{\one}{\mathbf{1}}
\newcommand{\nub}{\mathbf{\nu}}
\newcommand{\tra}[1]{\xrightarrow{#1}}
\newcommand{\wtra}[1]{\stackrel{#1}{\Longrightarrow}}

\newcommand{\lft}[1]{{{#1}\mathsf{L}}}
\newcommand{\rgt}[1]{{{#1}\mathsf{R}}}

\newcommand{\barb}[2]{\mathsf{barb}(#1,#2)}
\newcommand{\wbarb}[2]{\mathsf{wbarb}(#1,#2)}

\newcommand{\candit}{\!\Leftrightarrow\!}

\newcommand{\name}[1]{\mbox{{\scriptsize {{(#1)}}}}}

\newcommand{\linkr}[2]{[#1 \leftrightarrow #2]}

\newcommand{\para}{\mid}
\newcommand{\zero}{{\bf 0}}

\newcommand{\logsim}{\approx_{\mathtt{L}}}
\newcommand{\out}[1]{\langle #1\rangle}

\newcommand{\logeq}[4]{#1 \logsim #2 :: #3 [#4]}

\newcommand{\fn}[1]{\mbox{\it fn}(#1)}

\newcommand{\bn}[1]{\mbox{\it bn}(#1)}

\def\sub#1#2{\{\raisebox{.5ex}{\small$#2$}\! / \mbox{\small$#1$}\}}
\newcommand{\subst}[2]{\sub{#2}{#1}}

\newcommand{\redd}{\tra{~~~}}

\newcommand{\B}[1]{#1}

\newcommand{\TT}{\mathsf{T}}
\newcommand{\FF}{\mathsf{F}}

\newcommand{\blue}[1]{{\color{blue} #1}}

\newcommand{\llet}[2]{\m{let}\, #1 \,\m{in}\, #2}
\newcommand{\munit}{\langle\rangle}
\newcommand{\pack}[2]{\m{pack}\,#1\,\m{with}\,#2}
\newcommand{\mpair}[2]{\langle #1 \tensor #2 \rangle}

\newcommand{\hopi}{{Sess$\pi\lambda^+$}}
\newcommand{\valpi}{{Sess$\pi\lambda$}}

\makeatletter
\let\c@proposition\c@theorem
\let\c@corollary\c@theorem
\let\c@lemma\c@theorem
\let\c@definition\c@theorem
\let\c@example\c@theorem
\let\c@remark\c@theorem
\makeatother
\numberwithin{proposition}{section}
\numberwithin{corollary}{section}
\numberwithin{lemma}{section}
\numberwithin{theorem}{section}
\numberwithin{definition}{section}
\numberwithin{example}{section}
\numberwithin{remark}{section}

\begin{document}

\title{On Polymorphic Sessions and Functions }
\subtitle{A Tale of Two (Fully Abstract) Encodings}

\author{Bernardo Toninho \and Nobuko Yoshida}

\institute{
  Department of Computing,
  Imperial College London,
  United Kingdom
}

\maketitle

\begin{abstract}

%
%
%
  This work exploits the logical foundation of
  session types to determine what kind of type discipline for the
    $\pi$-calculus can exactly capture, and is captured by,
    $\lambda$-calculus behaviours. Leveraging the proof theoretic
    content of the soundness and completeness of sequent calculus and
    natural deduction presentations of linear logic, we develop the first
  \emph{mutually inverse} and \emph{fully abstract}
    processes-as-functions and functions-as-processes encodings
    between a polymorphic session $\pi$-calculus and a linear
    formulation of System F.
%
%
  We are then able to derive results of the session calculus from the theory
  of the $\lambda$-calculus: (1) we obtain a characterisation of inductive
  and coinductive session types via their algebraic representations in
  System F; and (2) we extend our results to account for \emph{value} and
  \emph{process} passing, 
  entailing strong normalisation. 
\end{abstract}

\section{Introduction}
\label{sec:intro}

Dating back to Milner's seminal work 
\cite{DBLP:journals/mscs/Milner92}, encodings of $\lambda$-calculus
into $\pi$-calculus are \B{seen} as essential benchmarks to examine
expressiveness of various extensions of the $\pi$-calculus.  Milner's
original motivation was to demonstrate the power of link mobility by
decomposing higher-order computations into pure name passing.  Another
goal was to analyse functional behaviours in a broad computational
universe of concurrency and non-determinism.  While
\emph{operationally} correct encodings of many higher-order
constructs 
exist, \B{it is challenging to obtain encodings that are
  precise wrt behavioural equivalence: the semantic distance between
  the $\lambda$-calculus and the $\pi$-calculus typically requires
  either restricting process behaviours
  \cite{DBLP:conf/birthday/Sangiorgi00} (e.g.~via typed equivalences
  \cite{DBLP:journals/acta/BergerHY05}) or enriching the
  $\lambda$-calculus with constants that allow for a suitable
  characterisation of the term equivalence induced by the behavioural
  equivalence on processes \cite{DBLP:conf/mfps/Sangiorgi93}.}



Encodings in $\pi$-calculi also gave rise to new typing disciplines:
Session types
\cite{DBLP:conf/concur/Honda93,honda.vasconcelos.kubo:language-primitives},
a typing system that is able to ensure deadlock-freedom for
communication protocols between two or more parties
\cite{DBLP:conf/popl/HondaYC08}, were originally motivated ``from
process encodings of various data structures in an asynchronous
version of the $\pi$-calculus'' \cite{Honda12}. Recently, a
propositions-as-types correspondence between linear logic and session
types
\cite{DBLP:conf/concur/CairesP10,DBLP:journals/mscs/CairesPT16,DBLP:journals/jfp/Wadler14}
has produced several new developments and logically-motivated
techniques
\cite{DBLP:conf/fossacs/ToninhoCP12,DBLP:journals/jfp/Wadler14,DBLP:conf/esop/CairesPPT13,DBLP:conf/icfp/LindleyM16}
to augment both the theory and practice of session-based
message-passing concurrency. Notably, parametric session polymorphism
\cite{DBLP:conf/esop/CairesPPT13} (in the sense of
Reynolds~\cite{DBLP:conf/ifip/Reynolds83}) has been proposed and a
corresponding abstraction theorem has been shown.

Our work expands upon the proof theoretic
consequences of this pro\-po\-si\-tions\--as\--types correspondence to
address the problem of  
how to {\em exactly} match the behaviours induced by \B{session}
$\pi$-calculus encodings of the $\lambda$-calculus with those of the
$\lambda$-calculus.  We develop {\em mutually inverse} and {\em fully
  abstract} encodings (up to typed observational congruences) between
a polymorphic session-typed $\pi$-calculus and the \B{polymorphic}
$\lambda$-calculus. The encodings arise from the proof theoretic
content of the equivalence between sequent calculus (i.e. the session
calculus) and natural deduction (i.e. the $\lambda$-calculus) for
\emph{second-order} intuitionistic linear logic, greatly generalising
\cite{DBLP:conf/fossacs/ToninhoCP12}.
%
\B{While fully abstract encodings between $\lambda$-calculi and
  $\pi$-calculi have been proposed
  (e.g.~\cite{DBLP:journals/acta/BergerHY05,DBLP:conf/mfps/Sangiorgi93}),
  our work is the first to consider a two-way, \emph{both} mutually
  inverse \emph{and} fully abstract embedding between the two calculi
  by crucially exploiting the linear logic-based session
  discipline. This also sheds some definitive light on the nature of
  concurrency in the (logical) session calculi, which exhibit ``don't
  care'' forms of non-determinism (e.g. processes may race on
  stateless replicated servers) rather than ``don't know'' non-determinism
  (which requires less harmonious logical features
  \cite{DBLP:journals/pacmpl/BalzerP17}). }


In the spirit of Gentzen \cite{GentzenND35},
we use our encodings as a tool
to study non-trivial properties of the session
calculus, deriving them from results in the $\lambda$-calculus: We
show the existence of inductive and coinductive sessions in the
polymorphic session calculus by considering the representation of
initial $F$-algebras and final $F$-coalgebras
\cite{DBLP:conf/lics/Mendler87} in the polymorphic $\lambda$-calculus
\cite{DBLP:journals/tcs/BainbridgeFSS90,DBLP:journals/mscs/Hasegawa94}
(in a linear setting \cite{DBLP:journals/lmcs/BirkedalMP06}). By
appealing to full abstraction, we are able to derive
processes that satisfy the necessary algebraic properties
and thus form adequate \B{\emph{uniform}} representations of inductive and coinductive
session types.
\B{The derived algebraic properties enable us to reason about standard
  data structure examples, providing a logical justification to typed
  variations of the representations in
  \cite{DBLP:journals/iandc/MilnerPW92}. }

We systematically extend our results to a session calculus with
$\lambda$-term and process passing (the latter being the core calculus
of \cite{DBLP:conf/esop/ToninhoCP13}, inspired by Benton's LNL
\cite{DBLP:conf/csl/Benton94}).  By showing that our encodings
naturally adapt to this setting, we prove that it is possible to
encode higher-order process passing in the first-order
{session} calculus fully
abstractly, providing a typed and proof-theoretically justified
re-envisioning of Sangiorgi's encodings of higher-order $\pi$-calculus
\cite{sangiorgipi}. In addition, the encoding
instantly provides a strong normalisation property of the higher-order
session calculus.

Contributions and the outline  
of our paper are as follows:\\[-5mm]
\begin{description}
\item[\S~\ref{sec:ftopi}] 
develops a functions-as-processes encoding of a linear
  formulation of System F, Linear-F, using a logically
  motivated polymorphic session $\pi$-calculus, Poly$\pi$,  
and shows that the encoding is operationally sound and complete.
\item[\S~\ref{sec:pitof}]  
develops a processes-as-functions encoding of
  Poly$\pi$ into Linear-F, arising from the completeness of the
  sequent calculus wrt natural deduction, also o\-pe\-ra\-tio\-nally sound and complete.
\item[\S~\ref{sec:fullabs}] 
studies the relationship between the two encodings, 
  establishing they are \emph{mutually inverse} and \emph{fully abstract}
  wrt typed congruence, the first two-way embedding satisfying
  \emph{both} properties.
\item[\S~\ref{sec:apps}]  
develops a \emph{faithful} representation
  of inductive and coinductive session types in Poly$\pi$ via the
  encoding of initial and final (co)algebras in the polymorphic
  $\lambda$-calculus. 
We demonstrate a use of {these} algebraic properties via examples. 
\item[\S~\ref{sec:hovals},\ref{sec:hopi}] 
study term-passing and process-passing session calculi,
  extending our encodings to provide embeddings into the first-order
  session calculus. 
We show full abstraction and mutual
inversion results, and derive  
strong normalisation of the higher-order session
calculus from the encoding. \\[-5mm]
\end{description}
In order to introduce our encodings, 
we first overview 
Poly$\pi$,  its 
typing system and behavioural equivalence ({\bf \S~\ref{sec:sessionpi}}).   
{We discuss related work and conclude 
with future work ({\bf \S~\ref{sec:related}}).}
Detailed proofs can be found in \cite{longversion}. 

\section{Polymorphic Session $\pi$-Calculus}\label{sec:sessionpi}
This section summarises the polymorphic session $\pi$-calculus
\cite{DBLP:conf/esop/CairesPPT13}, dubbed Poly$\pi$,  
arising as a process assignment to second-order linear logic
\cite{DBLP:journals/tcs/Girard87}, its typing system
and behavioural equivalences.  
%
\subsection{Processes and Typing}\label{sec:procstypes}
\mypara{Syntax} 
Given an infinite set $\Lambda$ of names $x,y,z,u,v$, the grammar of
processes $P,Q,R$ and session types $A,B,C$ is defined by:
{\small\[
\begin{array}{l}
\begin{array}{lclllllllllllllllll}
P,Q,R & ::= & x\langle y \rangle.P &\mid& x(y).P &\mid &P\mid Q & \mid& (\nub
              y)P &\mid & [x\leftrightarrow y] & \mid & \zero\\[1mm]
& \mid & x\langle A \rangle.P & \mid & x(Y).P & \mid & x.\m{inl};P &
 \mid & x.\m{inr};P & \mid & x.\m{case}(P,Q) & \mid & \bang x(y).P\\[1mm] 
\end{array}\\[1mm]
\begin{array}{lcl}
A, B & ::= & \one \mid A \lolli B \mid A \tensor B \mid A \with B \mid A
\oplus B \mid\,\, \bang A \mid \forall X . A \mid \exists X . A \mid X
\end{array}
\end{array}
\]}
\noindent 
$x\langle y\rangle .P$ denotes the output of channel $y$ on $x$ with
continuation process $P$; $x(y).P$ denotes an input along $x$, 
bound to $y$ in $P$; $P\mid Q$ denotes parallel composition; $(\nub y)P$
denotes the restriction of name $y$ to the scope of $P$; $\zero$
denotes the inactive process; $[x\leftrightarrow y]$ denotes the
linking of the two channels $x$ and $y$ (implemented as renaming); $x\langle A \rangle.P$ and
$x(Y).P$ denote the sending and receiving of a \emph{type} $A$ along
$x$ bound to $Y$ in $P$ of the receiver process;
$x.\m{inl};P$ and $x.\m{inr};P$ denote the emission of a selection
between the $\m{l}$eft or $\m{r}$ight branch of a receiver
$x.\m{case}(P,Q)$ process; $\bang x(y).P$ denotes an input-guarded
replication, that spawns replicas upon receiving an input along $x$.
We often abbreviate $(\nub y)x\langle y \rangle.P$ to $\ov{x}\langle y
\rangle.P$ and omit trailing $\zero$ processes.
By convention, we range over linear channels with $x,y,z$ and shared
channels with $u,v,w$. 

The syntax of session types is that of (intuitionistic) linear logic
propositions which are assigned to channels according to their usages
in processes: $\one$ denotes the type of a channel along which no
further behaviour occurs; $A \lolli B$ denotes a session that waits to
receive a channel of type $A$ and will then proceed as a session of
type $B$; dually, $A\tensor B$ denotes a session that sends a channel
of type $A$ and continues as $B$; 
$A\with B$ denotes a session that
offers a choice between proceeding as behaviours $A$ or $B$;
$A \oplus B$ denotes a session that internally chooses to continue as
either $A$ or $B$, signalling appropriately to the communicating
partner; $\bang A$ denotes a session offering an unbounded (but
finite) number of behaviours of type $A$; $\forall X.A$ denotes a
polymorphic session that receives a type $B$ and behaves uniformly as
$A\{B/X\}$; dually, $\exists X.A$ denotes an existentially typed
session, which emits a type $B$ and behaves as $A\{B/X\}$.

\mypara{Operational Semantics}
The operational semantics of our calculus is presented as a standard
labelled transition system (Fig.~\ref{fig:LTS}) 
in the style of the 
\emph{early} system for the $\pi$-calculus 
\cite{sangiorgipi}. 

In the remainder of this work we write $\equiv$ for a standard
$\pi$-calculus structural congruence extended with the clause
$[x\leftrightarrow y] \equiv [y\leftrightarrow x]$.  \B{In order to
  streamline the presentation of observational equivalence
  \cite{DBLP:conf/esop/PerezCPT12,DBLP:conf/esop/CairesPPT13}}, we
write $\equiv_\bang$ for structural congruence extended with the
so-called sharpened replication axioms \cite{sangiorgipi}, which
capture basic equivalences of replicated processes \B{(and are present in
the proof dynamics of the exponential of linear logic)}.
%
%
A transition
$P \tra{~\alpha~} Q$ denotes that 
$P$ may evolve to 
$Q$
by performing the action represented by label $\alpha$.
An action $\alpha$ ($\overline{\alpha}$) requires a matching $\overline{\alpha}$ ($\alpha$) in the environment to enable progress.
Labels  
include: the silent internal action $\tau$, output and bound
output actions ($\overline{x\out{y}}$ and $\overline{(\nu z)x\out z}$);
input action $x(y)$; 
the binary choice actions ($x.\inl$, 
$\overline{x.\inl}$, $x.\inr$, and 
$\overline{x.\inr}$); and 
output
and input actions of types ($\overline{x\out{A}}$ and $x(A)$).

  The labelled transition relation is
  defined by the rules in Fig.~\ref{fig:LTS}, subject to the side
  conditions: in rule $(\mathsf{res})$, we require
  $y\not\in\fn{\alpha}$; in rule $(\mathsf{par})$, we require
  $\bn{\alpha} \cap \fn{R} = \emptyset$; in rule $(\mathsf{close})$,
  we require $y\not\in\fn{Q}$. We omit the symmetric versions of
  $(\mathsf{par})$, $(\mathsf{com})$, $(\m{lout})$, $(\m{lin})$,
  $(\mathsf{close})$ and closure under $\alpha$-conversion.
%
%
We write $\rho_1
\rho_2$ for the composition of relations $\rho_1, \rho_2$.
We write $\tra{}$ to stand for $\tra{\tau}\equiv$.
Weak transitions are defined as usual: 
we write $\wtra{}$ for the reflexive, transitive closure of
$\tra{\tau}$ and $\tra{}^+$ for the transitive closure of $\tra{\tau}$.
Given $\alpha \neq \tau$, notation $\wtra{\alpha}$ stands for $\wtra{~}\tra{\alpha}\wtra{~}$ and 
$\wtra{\tau}$ stands for $\wtra{}$. 

\begin{figure}[t]
{\small
 \[
\begin{array}{ccccc}
\inferrule[\name{$\mathsf{out}$}]
{}{\outp{x}{y}.P \tra{\overline{x\out{y}}} P}
\hspace{0.3cm}
\inferrule[\name{$\mathsf{in}$}]
{}{x(y).P \tra{x(z)} P \subst{z}{y}}
\hspace{0.3cm}
\inferrule[\name{$\mathsf{outT}$}]
{}{\outp{x}{A}.P \tra{\overline{x \out A}} P}
\hspace{0.3cm}
\inferrule[\name{$\mathsf{inT}$}]
{}{x(Y).P \tra{x(B)} P \subst{B}{Y}}
\\[1mm]
\begin{array}{ll}
\begin{array}{c}
\inferrule[\name{$\mathsf{lout}$}]
  {}{x.\inl;P \tra{ \overline{x.\inl} } P}
\hspace{0.3cm}  
  \inferrule[\name{$\mathsf{id}$}]
{}{(\nu x)(\linkr{x}{y} \para P) \tra{\tau}  P\subst{y}{x}} 
\\[1mm]
\inferrule[\name{$\mathsf{lin}$}]
{}{x.\m{case} (P,Q) \tra{x.\inl} P}
\hspace{0.3cm}
\inferrule[\name{$\mathsf{rep}$}]
{}{\bang x(y).P \tra{x(z)} P \subst{z}{y}\para \bang x(y).P}
\\[1mm]
\end{array}\quad 
\begin{array}{l}
\inferrule[\name{$\mathsf{open}$}]
{P \tra{\overline{x\out y}} Q} {(\nub y)P \tra{\overline{(\nub
      y)x\out y}} Q}\\
\end{array}\\
\end{array}\\
\inferrule[\name{$\mathsf{close}$}]
{P \tra{\overline{(\nub y)x\out y}} P' \,\, Q \tra{x(y)} Q'}
{P \para Q \tra{\tau} (\nub y)(P' \para Q')}
\hspace{0.3cm}
\inferrule[\name{$\mathsf{par}$}]
{P \tra{\alpha} Q}{P\para R \tra{\alpha} Q \para R}
\hspace{0.3cm} 
\inferrule[\name{$\mathsf{com}$}]
{P\tra{\overline{\alpha}} P' \,\, Q \tra{\alpha} Q'} {P \para
  Q \tra{\tau} P' \para Q'}
\hspace{0.3cm}
\inferrule[\name{$\mathsf{res}$}]
{P \tra{\alpha} Q} {(\nub y)P \tra{\alpha} (\nub y)Q}

\end{array}
\vspace{-4ex}
\]
}
\caption{\label{fig:LTS}Labelled Transition System.}
\vspace{-4ex}
\end{figure}

\mypara{Typing System} 
The typing rules of Poly$\pi$ are given in Fig.~\ref{fig:typingpi},
following \cite{DBLP:conf/esop/CairesPPT13}. 
The rules define the judgment $\Omega ; \Ga ; \D \vdash P ::
z{:}A$, denoting that process $P$ offers a session of type $A$ along
channel $z$, using the \emph{linear} sessions in $\Delta$,
(potentially) using the unrestricted or \emph{shared} sessions in
$\Ga$, with polymorphic type variables maintained in $\Omega$. We 
use a well-formedness judgment $\Omega \vdash A\,\m{type}$
which states that $A$ is well-formed wrt the type variable environment
$\Omega$ (i.e. $\mathit{fv}(A) \subseteq \Omega$). 
We often write $T$
for the right-hand side typing $z{:}A$, $\cdot$ for the empty
context and $\D,\D'$ for the
union of contexts $\D$ and $\D'$, only defined when $\D$ and $\D'$
are disjoint. 
We write $\cdot \vdash P :: T$ for $\cdot ; \cdot ; \cdot \vdash P ::
T$. 

\begin{figure}[t]
\small
\[
\begin{array}{c}
\inferrule*[left=$(\rgt\lolli)$]
{\Omega ; \Ga ; \D , x{:}A \vdash P :: z{:}B}
{\Omega ; \Ga ; \D \vdash z(x).P :: z{:}A\lolli B}
\quad
\inferrule*[left=$(\rgt\tensor)$]
{\Omega ; \Ga ; \D_1 \vdash P :: y{:}A \quad 
 \Omega ; \Ga ; \D_2 \vdash Q :: z{:}B}
{\Omega ; \Ga ; \D_1 , \D_2 \vdash (\nub x)z\langle y \rangle.(P \mid
  Q) :: z{:}A\tensor B}\\[1em]
\inferrule*[left=$(\rgt\forall)$]
{\Omega , X ; \Ga ; \D \vdash P :: z{:}A}
{\Omega ; \Ga ; \D \vdash z(X).P :: z{:}\forall X.A}
\quad
\inferrule*[left=$(\lft\forall)$]
{\Omega \vdash B\,\m{type}\quad \Omega ; \Ga ; \D , x{:}A\{B/X\} \vdash P :: z{:}C}
{\Omega ; \Ga ; \D , x{:}\forall X . A \vdash x\langle B\rangle.P ::
  z{:}C}
\\[1em]
\inferrule*[left=$(\rgt\exists)$]
{\Omega \vdash B\,\m{type}\quad \Omega ; \Ga ; \D \vdash P :: z{:}A\{B/X\}}
{\Omega ; \Ga ; \D \vdash z\langle B\rangle.P ::
  z{:}\exists X.A}
\quad
\inferrule*[left=$(\lft\exists)$]
{\Omega , X ; \Ga ; \D , x{:}A \vdash P :: z{:}C}
{\Omega ; \Ga ; \D , x{:}\exists X.A \vdash x(X).P :: z{:}C}
  \\[1em]
  \inferrule*[left=$(\m{id})$]
{\, }
  { \Omega ; \Ga ; x{:}A \vdash [x\leftrightarrow z] :: z{:}A }
  \quad
\inferrule*[left=$(\cut)$]
{\Omega ; \Ga ; \D_1 \vdash P :: x{:}A \quad 
 \Omega ; \Ga ; \D_2 , x{:}A \vdash Q :: z{:}C}
{\Omega ; \Ga ; \D_1 , \D_2 \vdash (\nub x)(P \mid Q) :: z{:}C}
\end{array}
\]
\vspace{-4mm}
\caption{Typing Rules (Abridged -- See
  \cite{longversion} for all rules). \label{fig:typingpi}}
\vspace{-4mm}
\end{figure}

As in 
\cite{DBLP:conf/concur/CairesP10,DBLP:journals/mscs/CairesPT16,DBLP:conf/esop/PerezCPT12,DBLP:journals/jfp/Wadler14},
the typing discipline enforces that channel outputs always have as object a
\emph{fresh} name, in the style of the internal mobility
$\pi$-calculus \cite{DBLP:journals/tcs/Sangiorgi96a}. 
%
We clarify a few of the key rules: 
Rule $\rgt\forall$ defines the
meaning of (impredicative) universal quantification over session
types, stating that a session of type $\forall X.A$ inputs a type and
then behaves uniformly as $A$; dually, to use such a session (rule
$\lft\forall$), a process must output a type $B$ which then warrants
the use of the session as type $A\{B/X\}$. Rule $\rgt\lolli$ captures
session input, where a session of type $A\lolli B$ expects to receive
a session of type $A$ which will then be used to produce a session of
type $B$.
Dually, session output (rule $\rgt\tensor$) is achieved by
producing a
fresh session of type $A$ (that uses a disjoint set of sessions to those of
the continuation) and outputting the fresh session along $z$, which
is then a session of type $B$.
Linear 
composition is captured by rule $\cut$ which 
enables a
process that offers a session $x{:}A$ (using linear sessions in
$\D_1$) to be composed with a process that \emph{uses} that session
(amongst others in $\D_2$) to offer $z{:}C$.
%
As shown in \cite{DBLP:conf/esop/CairesPPT13}, typing
entails Subject Reduction, Global Progress, and Termination.



\mypara{Observational Equivalences}
\label{sec:logeq}
We briefly summarise the typed congruence and logical equivalence with
polymorphism, giving rise to a suitable notion of relational
parametricity in the sense of Reynolds
\cite{DBLP:conf/ifip/Reynolds83}, defined as a contextual logical
relation on typed processes \cite{DBLP:conf/esop/CairesPPT13}. The
logical relation is reminiscent of a typed bisimulation. However,
extra care is needed to ensure well-foundedness
due to impredicative type instantiation. As a consequence, the logical
relation allows us to reason about process equivalences where type
variables are not instantiated with \emph{the same}, but
rather \emph{related} types.

\mypara{Typed Barbed Congruence ($\tbcong$)}
We use the typed contextual congruence from 
\cite{DBLP:conf/esop/CairesPPT13}, which preserves  
\emph{observable} actions, called barbs. 
Formally, \emph{barbed congruence}, noted $\tbcong$, is 
the largest 
equivalence on well-typed processes
that is $\tau$-closed, barb preserving, and contextually closed under 
typed contexts; see 
\cite{DBLP:conf/esop/CairesPPT13} and \cite{longversion} for the full definition. 


\mypara{Logical Equivalence ($\logsim$)}
The definition of logical equivalence is no more than a typed
contextual bisimulation with the following intuitive reading: given
two open processes $P$ and $Q$ (i.e. processes with non-empty
left-hand side typings), we define their equivalence by inductively
closing out the context, composing with equivalent processes offering
appropriately typed sessions. When processes are closed, we have a
single distinguished session channel along which we can perform
observations, and proceed inductively on the structure of the
offered session type. We can then show that such an equivalence
satisfies the necessary fundamental properties
(Theorem~\ref{thm:logeqprops}).

The logical relation is defined using the candidates technique of
Girard \cite{girardtypes}. In this setting, an
\emph{equivalence candidate} is a relation on typed
processes satisfying basic closure conditions: an equivalence
candidate must be compatible with barbed congruence and closed under
forward and converse reduction.

\begin{definition}[Equivalence Candidate]\label{d:equivcand}\rm
An \emph{equivalence candidate} $\mathcal{R}$ at  $z{:}A$ and $z{:}B$,
noted $\mathcal{R} :: z{:} A \candit B$, is a binary relation on
processes such that,
for every $(P, Q) \in \mathcal{R} :: z{:} A\candit B$ both 
 $\cdot \vdash P :: z{:}A$ and $\cdot \vdash
Q :: z{:}B$ hold, together with the following (we often write 
$(P, Q) \in \mathcal{R} :: z{:}A\candit B$
as $P \,\mathcal{R}\, Q :: z{:}A\candit B$):
\begin{enumerate}
\item 
If
$(P, Q) \in \mathcal{R} :: z{:}A\candit B$, 
  $\cdot \vdash P \cong P':: z{:}A$, and 
  $\cdot \vdash Q \cong Q':: z{:}B$ then 
  $(P', Q') \in \mathcal{R} :: z{:}A\candit B$.
\item 
If $(P, Q) \in \mathcal{R} :: z{:}A\candit B$ then, for all $P_0$ such that 
$\cdot \vdash P_0 :: z{:}A$ and $P_0 \wtra{} P$, we have 
$(P_0, Q) \in  \,\mathcal{R} :: z{:}A\candit B$. 
Symmetrically for $Q$.
\end{enumerate} 
\end{definition}
To define the logical relation we rely on some auxiliary notation,
pertaining to the treatment of type variables arising due to
impredicative polymorphism.
We write $\omega : \Omega$ to denote a mapping $\omega$ that
assigns a closed type to the type variables in $\Omega$. We write
$\omega(X)$ for the type mapped by $\omega$ to variable $X$.
Given two mappings $\omega : \Omega$ and $\omega'
: \Omega$, we define an equivalence candidate {assignment} $\eta$
between $\omega$ and $\omega'$ as a mapping of equivalence candidate
$\eta(X) :: {-}{:}\omega(X) \candit \omega'(X)$ to the type variables
in $\Omega$, where the particular choice of a distinguished
right-hand side channel is \emph{delayed} (i.e. to be instantiated later on).
We write $\eta(X)(z)$ for the instantiation of
the (delayed) candidate with the name $z$. We write $\eta : \omega \candit
\omega'$ to denote that $\eta$ is a candidate assignment
between $\omega$ and $\omega'$; and $\hat{\omega}(P)$ to denote
the application of mapping $\omega$ to $P$.

We define a sequent-indexed family of process relations, that is, a
set of pairs of processes $(P,Q)$, written $\Gamma ; \Delta \vdash
\logeq{P}{Q}{T}{\eta : \omega \candit \omega'}$, satisfying
some conditions, typed under $\Omega ;
\Gamma ; \Delta \vdash T$, with $\omega : \Omega$, $\omega' : \Omega$
and $\eta : \omega \candit \omega'$. 
Logical equivalence is defined inductively on the size of the typing
contexts and then on the structure of the right-hand side type.
We show only select cases (see \cite{longversion} for the full definition).

\begin{definition}[Logical Equivalence]\rm
\label{def:logeqbase}
\label{d:logeqind}
{\textcolor{darkgray}{\sffamily\bfseries\mathversion{bold}(Base Case)}}
Given a type $A$ and mappings $\omega, \omega', \eta$, we 
define \emph{logical equivalence}, noted  
$P \logsim Q :: z{:}A[\eta : \omega \candit \omega']$, 
as the smallest symmetric binary relation containing all
pairs of processes $(P,Q)$ such that 
(i) $\cdot \vdash \hat{\omega}(P) :: z{:}\hat{\omega}(A)$; 
(ii) $\cdot \vdash \hat{\omega}'(Q) :: z{:}\hat{\omega}'(A)$; 
and (iii) satisfies the conditions given below: 
\begin{itemize}
\small
\item 
$\logeq{P}{Q}{z{:}X}{\eta : \omega\candit\omega'} \text{ iff } (P,Q)
\in \eta(X)(z)$
\item 
$ \logeq{P}{Q}{z{:}A\lolli B}{\eta : \omega\candit\omega'}$  iff $\forall P', y.~(P \tra{z(y)} P') \Rightarrow
 \exists Q'. Q \wtra{z(y)} Q'$ s.t. 
$\forall R_1,R_2.\logeq{R_1}{R_2}{y{:}A}{\eta : \omega
   \candit \omega'} \logeq{(\nu y)(P' \,|\, R_1)}{(\nu
   y)(Q'  \,|\, R_2)}{z{:}B}{\eta : \omega\candit\omega'}$
\item 
$\logeq{P}{Q}{z{:}A\tensor B}{\eta : \omega\candit\omega'}$ 
 iff $\forall P', y.~~(P \tra{\ov{(\nu y)z\out y}} P') \Rightarrow
 \exists Q'. Q \wtra{\ov{(\nu y)z\out y}} Q'$  s.t. 
$\exists P_1,P_2,Q_1,Q_2. \ P' \equiv_\bang P_1 \mid P_2 \wedge Q'
\equiv_\bang Q_1 \mid Q_2\wedge 
\logeq{P_1}{Q_1}{y{:}A}{\eta : \omega\candit\omega'} \wedge 
 \logeq{P_2}{Q_2}{z{:}B}{\eta : \omega\candit\omega'}$
\item 
$\logeq{P}{Q}{z{:}\forall X. A}{\eta : \omega\candit\omega'}$ 
 iff 
$\forall B_1 , B_2 ,P',\mathcal{R} ::{-}{:}B_1\candit B_2.~~( P
\tra{z(B_1)} P' )$  implies 
$\exists Q' . Q \wtra{z(B_2)} Q', ~P'  \logsim  Q'::z{:}A[\eta[X\mapsto \mathcal{R}] : \omega[X \mapsto
B_1] \candit \omega'[X\mapsto B_2]]$
\end{itemize}
{\textcolor{darkgray}{\sffamily\bfseries\mathversion{bold}(Inductive Case)}}
Let $\Gamma, \Delta$ be non empty. 
Given $\Omega ; \Gamma ; \Delta \vdash P :: T$ and
$\Omega ; \Gamma ; \Delta \vdash Q :: T$, 
 the binary relation on processes 
  $\Gamma ; \Delta \vdash
\logeq{P}{Q}{T}{\eta : \omega\candit \omega'}$
(with $\omega , \omega' : \Omega$ and $\eta : \omega \candit\omega'$) 
is inductively defined as:
{\small
$$
\begin{array}{lcl}
\Gamma ; \Delta , y:A \vdash
\logeq{P}{Q}{T}{\eta : \omega\candit\omega'} &\mbox{ iff }&
  \forall R_1 , R_2 .\mbox{ s.t. } \logeq{R_1}{R_2}{y{:}A}{\eta :
    \omega\candit\omega' } ,\\
&&
\hspace{-3cm}\Gamma ; \Delta \vdash \logeq{(\nu y)(\hat{\omega}(P) \mid \hat{\omega} (R_1))}{(\nu
    y)(\hat{\omega}'(Q) \mid \hat{\omega}'(R_2))}{T}{\eta : \omega\candit\omega' }\\[1mm]
\Gamma , u:A ; \Delta \vdash
\logeq{P}{Q}{T}{\eta : \omega\candit\omega'} &\mbox{ iff }&
  \forall R_1 , R_2 .\mbox{ s.t. } \logeq{R_1}{R_2}{y{:}A}{\eta :
    \omega\candit\omega' } ,\\
&& \hspace{-4.6cm}\Gamma ; \Delta \vdash \logeq{(\nub u)(\hat{\omega} (P) \mid \bang u(y). \hat{\omega}(R_1))}{(\nub
    u)(\hat{\omega}'(Q) \mid \bang u(y). \hat{\omega}'(R_2))}{T}{\eta : \omega\candit\omega' }
\end{array}
$$}
\end{definition}

For the sake of readability we often omit the $\eta : \omega \candit
\omega'$ portion of $\logsim$, which is henceforth implicitly
universally quantified. Thus, we write $\Omega ; \Ga ; \D \vdash P
\logsim Q :: z{:}A$ (or $P \logsim Q$) iff the two given processes are logically
equivalent for all consistent instantiations of its type variables.

It is instructive to inspect the clause for type input
($\forall X.A$):  
the two processes must be able to match inputs of any pair of
\emph{related} types (i.e. types related by a candidate), such that
the continuations are related at the open type $A$ with the
appropriate type variable instantiations, following Girard
\cite{girardtypes}. 
The power of this style of logical relation arises from a combination
of the extensional flavour of the equivalence and the
  fact that polymorphic equivalences do not require the same 
type to be instantiated in both processes, but rather that the
types are \emph{related} (via a suitable equivalence candidate
relation).


\begin{theorem}[Properties of Logical Equivalence \cite{DBLP:conf/esop/CairesPPT13}]
\label{thm:logeqprops}
~
\begin{description}
\item[Parametricity:]  If $\Omega ; \Gamma; \Delta \vdash P :: z{:}A$ then, for all $\omega, \omega' : \Omega$ and $\eta : \omega\candit\omega'$,
we have $\Gamma ; \Delta \vdash
\logeq{\hat{\omega}(P)}{\hat{\omega'}(P)}{z{:}A}{\eta :
  \omega\candit\omega'}$.
\item[Soundness:] If $\Omega ; \Gamma ; \Delta \vdash P \logsim Q :: z{:}A$ 
then $\mathcal{C}[P] \cong \mathcal{C}[Q] :: z{:}A$, for
any closing $\mathcal{C}[-]$.
\item[Completeness:] If 
$\Omega ; \Gamma; \Delta \vdash P \cong Q :: z{:}A$
then 
$\Omega ; \Gamma ; \Delta \vdash P \logsim Q :: z{:}A$.
\end{description}
\end{theorem}



\section{To Linear-F and Back}\label{sec:sysf}

We now develop our mutually inverse and fully abstract encodings
between Poly$\pi$ and a linear polymorphic $\lambda$-calculus
\cite{DBLP:conf/aplas/ZhaoZZ10} that we dub Linear-F. 
We first introduce the syntax and typing of the linear
$\lambda$-calculus and then proceed to detail our encodings and their
properties (we omit typing ascriptions from the existential
polymorphism constructs for readability).

\begin{definition}[Linear-F]\label{def:linf_syntax}\rm
The syntax of terms $M,N$ and types $A, B$ of Linear-F is
given below. 
\[
\small
\begin{array}{lcl}
M, N & ::= & \lambda x{:}A.M \mid  M\,N \mid \mpair{M}{N}
       \mid \llet{x\tensor y = M}{N}  \mid \,\bang M \mid
       \llet{\bang u = M}{N} \mid \Lambda X.M \\[1mm]
 & \mid &  M[A] \mid \pack{A}{M} \mid \llet{(X,y) =
          M}{N}  \mid \llet{\one = M}{N} \mid \munit \mid \TT \mid \FF\\[0.5em]
A,B & ::= & A \lolli B \mid A \tensor B \mid \, \bang A \mid \forall X.A
            \mid \exists X.A \mid X \mid \one \mid \mathbf{2} 
\end{array}
\]
\end{definition}
\noindent The syntax of types is that of the multiplicative and
exponential fragments of second-order intuitionistic linear logic:
$\lambda x{:}A.M$ denotes linear $\lambda$-abstractions; $M\,N$
denotes the application; $\mpair{M}{N}$ denotes the multiplicative
pairing of $M$ and $N$, as reflected in its elimination form
$\llet{x\tensor y = M}{N}$ which simultaneously deconstructs the pair
$M$, binding its first and second projection to $x$ and $y$ in $N$,
respectively; $\bang M$
denotes a term $M$ that does not use any linear variables and so may
be used an arbitrary number of times;
$\llet{\bang u = M}{N}$ binds the underlying exponential term of $M$
as $u$ in $N$; $\Lambda X.M$ is the type abstraction former; $M[A]$
stands for type application; $\pack{A}{M}$ is the existential type
introduction form, where $M$ is a term where the existentially typed
variable is instantiated with $A$; $\llet{(X,y) = M}{N}$ unpacks an
existential package $M$, binding the representation type to $X$ and
the underlying term to $y$ in $N$; the multiplicative unit $\one$ has
as introduction form the nullary pair $\munit$ and is eliminated by
the construct $\llet{\one = M}{N}$, where $M$ is a term of type
$\one$.  Booleans (type $\mathbf{2}$ with values $\TT$ and $\FF$) are
the basic observable.

The typing judgment in Linear-F is given as 
$\Omega ; \Ga ; \D \vdash M : A$, following the DILL formulation of
linear logic \cite{barberdill96}, stating
that term $M$ has type $A$ in a linear context $\D$ (i.e.
bindings for linear variables $x{:}B$),
intuitionistic context $\Ga$ (i.e. binding for intuitionistic variables
 $u{:}B$) and type variable context $\Omega$. 
The typing rules are standard \cite{DBLP:conf/esop/CairesPPT13}. The operational semantics of the calculus
are the expected call-by-name semantics with commuting conversions
\cite{DBLP:journals/tcs/MaraistOTW99}. We write $\Downarrow$ for the
evaluation relation.
We write $\cong$ for the
largest typed congruence that is consistent with the observables of
type $\mathbf{2}$ \B{(i.e. a so-called Morris-style equivalence as in
  \cite{DBLP:journals/acta/BergerHY05})}. 

\subsection{Encoding Linear-F into Session $\pi$-Calculus}\label{sec:ftopi} 
We define a translation from Linear-F to Poly$\pi$
generalising the one from 
\cite{DBLP:conf/fossacs/ToninhoCP12}, accounting for 
polymorphism and multiplicative pairs.
We translate typing
derivations of $\lambda$-terms to those of $\pi$-calculus
terms (we omit the full typing derivation for the sake of
readability).

Proof theoretically, the $\lambda$-calculus corresponds to a proof
term assignment for natural deduction presentations of logic, whereas
the session $\pi$-calculus from \S~\ref{sec:sessionpi} corresponds to
a proof term assignment for sequent calculus. Thus, we obtain a
translation from $\lambda$-calculus to the session $\pi$-calculus by
considering the proof theoretic content of the constructive proof of
soundness of the sequent calculus wrt natural deduction. Following
Gentzen \cite{GentzenND35}, the translation from natural deduction to
sequent calculus maps introduction rules to the corresponding right
rules and elimination rules to a combination of the corresponding left
rule, cut and/or identity.

Since typing in the session calculus identifies a distinguished
channel along which a process offers a session, the translation of
$\lambda$-terms is parameterised by a ``result'' channel along which
the behaviour of the $\lambda$-term is implemented. Given a
$\lambda$-term $M$, the process $\lb M \rb_z$ encodes the behaviour of
$M$ along the session channel $z$.
 We enforce that
  the type $\mathbf{2}$ of booleans and its two constructors are
  consistently translated to their polymorphic Church encodings before
  applying the translation to Poly$\pi$. Thus, type $\mathbf{2}$ is first
  translated to $\forall X.\bang X \!\lolli \,\bang X\!
  \lolli X$, the value $\TT$ to $\Lambda X.\lambda u{:}\bang X . \lambda
  v{:}\bang X . \llet{\bang x = u}{\llet{\bang y = v}{x}}$ and the
  value $\FF$
  to $\Lambda X.\lambda u{:}\bang X . \lambda
  v{:}\bang X . \llet{\bang x = u}{\llet{\bang y = v}{y}}$. Such
  representations of the booleans are adequate 
  up to parametricity \cite{DBLP:journals/lmcs/BirkedalMP06} \B{and
    suitable for our purposes of relating the session calculus (which
    has no primitive notion of value or result type) with
    the $\lambda$-calculus precisely due to the tight correspondence
    between the two calculi}.

\begin{definition}[From Linear-F to Poly$\pi$]\rm
\label{def:ltopi}
$\lb\Omega\rb ; \lb\Ga\rb ; \lb\D\rb \vdash \lb M \rb_z ::
z{:} A$ denotes the translation of contexts, types and terms from
Linear-F to the polymorphic session calculus. 
The translations on
contexts and types are the identity function. 
Booleans and their values are first translated to their Church
encodings as specified above. The translation on $\lambda$-terms is
given below:

{\small
\[
\begin{array}{lcllcl}
\lb x\rb_z & \triangleq & [x\leftrightarrow z]  & 
\lb M \, N \rb_z  
\triangleq 
&\hspace{-2cm} 
\hspace{-2cm} 
& \hspace{-2cm} 
(\nub x)(\lb M \rb_x \mid (\nub y)x\langle y \rangle.(\lb N \rb_y
                                  \mid
                                 [x\leftrightarrow z]))
\\
 \lb u \rb_z & \triangleq & (\nub x)u\langle x
                             \rangle.[x\leftrightarrow z] & 
\lb \llet{\bang u = M}{N}\rb_z & \triangleq & 
   (\nub x)(\lb  M \rb_x \mid \lb N \rb_z\{x/u\})\\
\lb \lambda x{:}A.M \rb_z & \triangleq & z(x).\lb M \rb_z & 
\lb  \mpair{M}{N} \rb_z & \triangleq & 
(\nub y)z\langle  y \rangle.(\lb M \rb_y \mid \lb N\rb_z)\\
 \lb \bang M \rb_z & \triangleq & \bang z(x).\lb M \rb_x & 
\lb \llet{x\tensor y = M}{N} \rb_z & \triangleq & 
(\nub w)(\lb M \rb_y \mid y(x).\lb N \rb_z )\\
\lb \Lambda X.M\rb_z & \triangleq & z(X).\lb M \rb_z & 
\lb M[A] \rb_z & \triangleq & (\nub x)(\lb M \rb_x \mid x\langle A
                              \rangle.[x\leftrightarrow z])\\
\lb \pack{A}{M} \rb_z & \triangleq & z\langle A \rangle.\lb M\rb_z & \lb \llet{(X,y) = M}{N} \rb_z & \triangleq & 
(\nub x)(\lb M\rb_y \mid y(X).\lb N\rb_z)\\
\lb \munit \rb_z & \triangleq & \zero &
\lb \llet{\one = M}{N} \rb_z & \triangleq & (\nub x)(\lb M \rb_x \mid \lb N \rb_z)

\end{array}
\]}
\end{definition}
\noindent 
To translate a (linear) $\lambda$-abstraction $\lambda x{:}A.M$, which
corresponds to the proof term for the introduction rule for $\lolli$,
we map it to the corresponding $\rgt\lolli$ rule, thus obtaining a
process $z(x).\lb M \rb_z$ that inputs along the result channel $z$ a
channel $x$ which will be used in $\lb M \rb_z$ to access the function
argument. To encode the application $M\,N$, we compose (i.e. $\cut$)
$\lb M \rb_x$, where $x$ is a fresh name, with a process that provides
the (encoded) function argument by outputting along $x$ a channel $y$
which offers the behaviour of $\lb N \rb_y$. After the output is
performed, the type of $x$ is now that of the function's codomain and
thus we conclude by forwarding (i.e. the $\m{id}$ rule) between $x$
and the result channel $z$.

The encoding for polymorphism follows a similar pattern: To encode the
abstraction $\Lambda X.M$, we receive along the result channel a type
that is bound to $X$ and proceed inductively. To encode type
application $M[A]$ we encode the abstraction $M$ in parallel with a
process that sends $A$ to it, and forwards accordingly. Finally, the
encoding of the existential package $\pack{A}{M}$ maps to an output of
the type $A$ followed by the behaviour $\lb M \rb_z$, with the
encoding of the elimination form $\llet{(X,y) = M}{N}$ composing the
translation of the term of existential type $M$ with a process
performing the appropriate type input and proceeding as $\lb N \rb_z$.


\B{\begin{example}[Encoding of Linear-F]
\label{ex:lnearF}
Consider the following $\lambda$-term
corresponding to a polymorphic pairing function (recall that we write
$\ov{z}\langle w \rangle.P$ for $(\nub w)z\langle w \rangle.P$):
\[
\small
\begin{array}{lcllcl}
M \triangleq  \Lambda X.\Lambda Y.\lambda x{:}X.\lambda
  y{:}Y.\mpair{x}{y} 
&
\mbox{and}
& N \triangleq  ((M[A][B]\,M_1)\, M_2)
\end{array}
\]
Then we have, with $\tilde{x}=x_1x_2x_3x_4$:
\[
\small
\begin{array}{rcll}
  \lb N \rb_z & \equiv &
(\nub \tilde{x})
(&\hspace{-1mm}
\lb M\rb_{x_1} \mid x_1\langle
  A\rangle.[x_1\leftrightarrow x_2]
   \mid x_2\langle B
  \rangle.[x_2\leftrightarrow x_3] \mid\\
 & & &\hspace{-1mm}\ov{x_3}\langle x\rangle.(\lb
  M_1 \rb_x \mid [x_3 \leftrightarrow x_4]) \mid \ov{x_4}\langle y
  \rangle.(\lb M_2\rb_y \mid [x_4\leftrightarrow z]))\\
  & \equiv & 
             (\nub \tilde{x})(
             &\hspace{-1mm} x_1(X).x_1(Y).x_1(x).x_1(y).
             \ov{x_1}\langle w\rangle.([x\leftrightarrow w] \mid [y
             \leftrightarrow x_1]) \mid 
x_1\langle
             A\rangle.[x_1\leftrightarrow x_2] \mid \\
  & &  &\hspace{-1mm} 
x_2\langle B
  \rangle.[x_2\leftrightarrow x_3] \mid \ov{x_3}\langle x\rangle.(\lb
      M_1 \rb_x \mid [x_3 \leftrightarrow x_4]) \mid 
\ov{x_4}\langle y \rangle.(\lb M_2\rb_y \mid [x_4\leftrightarrow z]))
\end{array}
\]
We can observe that $N \tra{}^+ (((\lambda x{:}A.\lambda
y{:}B.\mpair{x}{y})\,M_1)\,M_2) \tra{}^+ \mpair{M_1}{M_2}$. At the
process level, each reduction corresponding to the redex of type
application is simulated by two reductions, obtaining:
\[
\small
\begin{array}{lcll}
  \lb N \rb_z  & \tra{}^+ & 
   (\nub x_3,x_4)(&\hspace{-1mm}x_3(x).x_3(y).\ov{x_3}\langle w \rangle.([x
  \leftrightarrow w] \mid [y\leftrightarrow x_3]) \mid\\
& &  &\hspace{-3mm}\ov{x_3}\langle x\rangle.(\lb
     M_1 \rb_x \mid [x_3 \leftrightarrow x_4]) \mid
  \ov{x_4}\langle y
  \rangle.(\lb M_2\rb_y \mid [x_4\leftrightarrow z])) = P 
\end{array}
\]
The reductions corresponding to the $\beta$-redexes clarify the way in
which the encoding represents substitution of terms for variables via
fine-grained name passing. Consider 
$\lb \mpair{M_1}{M_2} \rb_z \triangleq  \ov{z}\langle
                                              w\rangle.(\lb M_1\rb_w
                                             \mid \lb M_2\rb_z)
                                              $ 
and 
\[
\small
  \begin{array}{rcl}
    P & \tra{}^+ & (\nub x,y)(\lb M_1\rb_x \mid \lb M_2\rb_y \mid
                   \ov{z}\langle w \rangle.([x\leftrightarrow w] \mid
                   [y\leftrightarrow z]))
  \end{array}
\]
The encoding of the pairing of $M_1$ and $M_2$ outputs a fresh name
$w$ which will denote the behaviour of (the encoding of) $M_1$, and
then the behaviour of the encoding of $M_2$ is offered on $z$. The
reduct of $P$ outputs a fresh name $w$ which is then identified with
$x$ and thus denotes the behaviour of $\lb M_1\rb_w$. The channel $z$
is identified with $y$ and thus denotes the behaviour of $\lb
M_2\rb_z$, making the two processes listed above equivalent.
 This informal reasoning exposes the insights that justify the
operational correspondence of the encoding. Proof-theoretically, these
equivalences simply map to commuting conversions which push the
processes $\lb M_1\rb_x$ and $\lb M_2\rb_z$ under the output on $z$.
\end{example}}


\begin{theorem}[Operational Correspondence]~
\label{thm:ftopioc}
\begin{itemize}
\item If $\Omega ; \Ga ; \D \vdash M : A$ and $M \tra{} N$ then $\lb M
  \rb_z \wtra{} P$ such that $\lb N\rb_z \logsim P$
\item If $\lb M \rb_z \tra{} P$ then $M \tra{}^+ N$ and $\lb N\rb_z
  \logsim P$
\end{itemize}

\end{theorem}

\subsection{Encoding Session $\pi$-calculus to Linear-F}\label{sec:pitof}

Just as the proof theoretic content of the soundness of sequent
calculus wrt natural deduction induces a translation from
$\lambda$-terms to session-typed processes, the \emph{completeness} of
the sequent calculus wrt natural deduction induces a translation from
the session calculus to the $\lambda$-calculus. This mapping identifies
sequent calculus right rules  with the introduction rules of
natural deduction and left rules with elimination rules
combined with (type-preserving) substitution.
\B{Crucially, the mapping is defined on \emph{typing derivations},
  enabling us to consistently identify when a process uses a session
  (i.e. left rules) or, dually, when a process offers a session
  (i.e. right rules).}

\begin{figure}[t]

  \[
{\footnotesize
    \begin{array}{l}
\pif{\inferrule*[left=$(\rgt\lolli)$]{ \D , x{:}A\vdash  \blue
    P  ::z{:}B }{\D \vdash  \blue{z(x).P} ::
    z{:}A\lolli B}}
\triangleq
\inferrule*[left=$(\lolli I)$]{\D , x{:}A \vdash \blue{\llp P \rrp}_{
                                                        \D,x{:}A\vdash z{:}B}
                                                           : B  }{ \D
                                                          \vdash
                                                          \blue{\lambda
                                                          x{:}A. \llp
                                                                    P \rrp}_{\D,x{:}A
                                                           \vdash z{:}B}
                                                           : A
                                                          \lolli B}
    \\[2em]
   \pif{\inferrule[$(\lft\lolli)$]{ \D_1 \vdash \blue P :: y{:}A \quad 
                                                         \D_2 , x{:}B
                                                         \vdash \blue Q :: z{:}C}{ \D_1,\D_2 , x{:}A\lolli B \vdash \blue{(\nub
  y)x\langle y \rangle.(P\mid Q)} :: z{:}C}}

 \triangleq 
\\[2em]                                                       
\inferrule[(subst)]
{\D_2 , x{:}B \vdash \blue{\llp Q \rrp}_{ \D_2 , x{:}B \vdash z{:}C} : C \quad
\inferrule[$(\lolli E)$]
{ x{:}A\lolli B \vdash \blue x {:} A \lolli B \quad
  \D_1 \vdash \blue{\llp P \rrp}_{ \D_1 \vdash y{:}A} : B }
{\D_1 , x {:} A\lolli B \vdash \blue{x\,\llp P \rrp}_{ \D_1 \vdash y{:}A} : B }}
{ \D_1 , \D_2 , x{:}A\lolli B \vdash \blue{\llp Q \rrp}_{ \D_2 , x{:}B \vdash z{:}C}
\blue{\{(x\,\llp P \rrp}_{ \D_1 \vdash y{:}A} \blue{)/x\}} : C} 
  \end{array}}
\]
\caption{Translation on Typing Derivations (Excerpt -- See
 ~\cite{longversion})\label{fig:derivpitof}}
\end{figure}

\begin{definition}[From Poly$\pi$ to Linear-F]\rm
\label{def:pitof}
  We write
  $\llp \Omega \rrp ; \llp \Ga \rrp ; \llp \D \rrp \vdash \llp P \rrp :
  A$
  for the translation from typing derivations in Poly$\pi$ to derivations in
  Linear-F. The translations on types and contexts are the identity
  function. The translation on processes is given below, where the
  leftmost column indicates the typing rule at the root of the
  derivation (see Fig.~\ref{fig:derivpitof} for an excerpt of the
  translation on typing derivations, where we write  $\llp P \rrp_{\Omega ; \Ga
  ; \D \vdash z{:}A}$ to denote the translation of $\Omega ; \Ga ;
  \D \vdash P :: z{:}A$. We omit $\Omega$ and $\Ga$ when unchanged).

{\small
\[
\begin{array}{llclllcl}
 (\rgt\one) & \llp \zero \rrp & \triangleq & \munit &
(\lft\lolli) & \llp (\nub y)x\langle y \rangle.(P\mid Q) \rrp & \triangleq &
                                      \llp Q \rrp \{(x\,\llp P
                                                              \rrp)/x\}\\
(\m{id}) & \llp [x\leftrightarrow y] \rrp & \triangleq & x &
(\rgt\lolli) & \llp z(x).P \rrp & \triangleq & \lambda x{:}A. \llp P \rrp \\ 
(\lft\one) & \llp P \rrp & \triangleq & \llet{\one = x}{\llp P \rrp} &
(\rgt\tensor) & \llp (\nub x)z\langle x \rangle.(P \mid Q) \rrp & \triangleq & 
\mpair{\llp P\rrp}{\llp Q \rrp}\\  
(\rgt\bang ) &  \llp \bang z(x).P \rrp & \triangleq & \bang \llp P \rrp 
 & 
(\lft\tensor) & \llp x(y).P\rrp & \triangleq & \llet{ x\tensor y =
                                               x}{\llp P\rrp}\\ 
(\lft\bang) & \llp P\{u/x\} \rrp & \triangleq & \llet{\bang u =
                                                x}{\llp P\rrp} & 
(\cpy) & \llp (\nub x)u\langle x \rangle.P \rrp & \triangleq & \llp P
                                                      \rrp\{u/x\}\\
(\rgt\forall) & \llp z(X).P \rrp & \triangleq & \Lambda X.\llp P\rrp & 
(\lft\forall) & \llp x\langle B \rangle.P \rrp & \triangleq &\llp P\rrp \{(x[B])/x\}\\ 
(\rgt\exists) & \llp z\langle B\rangle.P \rrp & \triangleq &
                                                             \pack{B}{\llp P \rrp} & 
(\lft\exists) & \llp x(Y).P \rrp & \triangleq & \llet{(Y,x) = x}{ \llp P \rrp}\\ 
(\cut) & \llp (\nub x)(P\mid Q) \rrp  & \triangleq & \llp Q \rrp
                                                     \{\llp P \rrp/x\}
                                                     & 
(\cut^\bang) & \llp (\nub u)(\bang u(x).P \mid Q) \rrp & \triangleq &
                          \llp Q \rrp\{\llp P \rrp/u\}\\ 
\end{array}
\]}
\end{definition}

\noindent For instance, the encoding of a process $z(x).P :: z{:} A \lolli B$,
typed by rule $\rgt\lolli$, results in the
corresponding $\lolli\! I$ introduction rule in the $\lambda$-calculus and thus is
$\lambda x{:}A.\llp P \rrp$. To encode the process
$(\nub y)x\langle y \rangle.(P \mid Q)$, typed by rule $\lft\lolli$,
we make use of substitution: 
Given that
the sub-process $Q$ is typed as
$\Omega ; \Ga ; \D' , x{:}B \vdash Q :: z{:}C$, the encoding of the
full process is given by $\llp Q \rrp\{(x\,\llp P \rrp)/x\}$. The term
$x\, \llp P \rrp$ consists of the application of $x$ (of function
type) to the argument $\llp P \rrp$, thus ensuring that the term
resulting from the substitution is of the appropriate type. We note
that, for instance, the encoding of rule $\lft\tensor$ does not need
to appeal to substitution -- the $\lambda$-calculus $\m{let}$ style
rules can be mapped directly. Similarly, rule $\rgt\forall$ is mapped
to type abstraction, whereas rule $\lft\forall$ which types a
process of the form $x\langle B \rangle.P$  maps to a
substitution of the type application $x[B]$ for $x$ in $\llp P
\rrp$. The encoding of existential polymorphism is simpler
due to the $\m{let}$-style elimination.
\B{We also highlight the encoding of the $\m{cut}$ rule which embodies
  parallel composition of two processes sharing a linear name, which
  clarifies the use/offer duality of the intuitionistic calculus --
  the process that offers $P$ is encoded and substituted into the
  encoded user $Q$. }

\begin{restatable}{theorem}{typsoundpitof}\label{thm:typsoundpitof}
If $\Omega ; \Ga ; \D \vdash P :: z{:}A$ then $\llp \Omega \rrp ; \llp
\Ga \rrp ; \llp \D \rrp \vdash \llp P \rrp : A$.
\end{restatable}

\begin{example}[Encoding of Poly$\pi$]
\label{ex:polypi}
Consider the following
processes%
\[\small
  \begin{array}{l}
  P  \triangleq  z(X).z(Y).z(x).z(y).\ov{z}\langle w
                   \rangle.([x\leftrightarrow w] \mid
    [y\leftrightarrow z]) 
\quad   
  Q  \triangleq  z\langle \one \rangle.z\langle \one \rangle.\ov{z}\langle 
                   x\rangle.\ov{z}\langle y \rangle.z(w).[w\leftrightarrow
                   r]
\end{array}
\]
with {\small $\vdash P :: z{:}\forall X.\forall Y.X\lolli Y \lolli
X \tensor Y$} and {\small $z{:}\forall X.\forall Y.X\lolli Y \lolli
X \tensor Y \vdash Q :: r{:}\one$}. 
$\small
\begin{array}{ll}
\mbox{Then:}
   &\llp P \rrp   =  
                   \Lambda X.\Lambda Y.\lambda x{:}X.\lambda
                               y{:}Y.\mpair{x}{y}
                               \quad \quad 
    \llp Q \rrp  =  \llet{x\tensor y =
                      z[\one][\one]\,\langle\rangle\,\langle\rangle}
{\llet{\one = y}{x}}\\

  &\llp (\nub z)(P \mid Q) \rrp  =  \llet{x\tensor y =
                      (\Lambda X.\Lambda Y.\lambda x{:}X.\lambda
                               y{:}Y.\mpair{x}{y})[\one][\one]\,\langle\rangle\,\langle\rangle}
{\llet{\one = y}{x}}
  \end{array}
$\\[1mm]
\B{By the behaviour of $(\nub z)(P\mid Q)$, which consists
  of a sequence of cuts, and its encoding, }
we have that $\llp (\nub z)(P \mid Q) \rrp
\tra{}^+ \langle\rangle$ and $(\nub
z)(P \mid Q) \tra{}^+ \zero = \llp \langle \rangle \rrp$.
\end{example}

In general, the translation of Def.~\ref{def:pitof} can introduce some
distance between the immediate operational behaviour of a process and
its corresponding $\lambda$-term, insofar as the translations of cuts
(and left rules to non $\m{let}$-form elimination rules) make use of
substitutions that can take place deep within the resulting term. 
\B{Consider the process at the root of the following typing judgment  
$\D_1 , \D_2 , \D_3 \vdash (\nub x)(x(y).P_1 \mid (\nub y)x\langle y
  \rangle.(P_2 \mid w(z).\zero)) :: w{:}\one\lolli \one$, derivable through
  a $\m{cut}$ on session $x$ between instances of $\rgt\lolli$ and
  $\lft\lolli$, where the continuation
  process $w(z).\zero$ offers a session $w{:}\one\lolli \one$ (and so
  must use rule $\lft\one$ on $x$). We have that:
$(\nub x)(x(y).P_1 \mid (\nub y)x\langle y
  \rangle.(P_2 \mid w(z).\zero)) \tra{} (\nub x,y)(P_1 \mid P_2 \mid
  w(z).\zero)$. However, the translation of the process above results in
  the term $\lambda z{:}\one.\llet{\one = ((\lambda y{:}A.\llp P_1\rrp)\,\llp P_2
\rrp)}{\llet{\one = z}{\munit}}$, where the redex that corresponds to the process reduction is
present but hidden under the binder for $z$ (corresponding to the
input along $w$). Thus, to establish operational
completeness we consider full $\beta$-reduction, denoted by 
$\tra{}_\beta$, i.e. enabling $\beta$-reductions under binders.}

\begin{restatable}[Operational Completeness]{theorem}{thmopcs}
\label{thm:opc1}
Let $\Omega ; \Ga ; \D \vdash P :: z{:}A$.
If $P\tra{}Q$ then $\llp P \rrp \tra{}_\beta^* \llp Q\rrp$.
\end{restatable}

\B{In order to study the soundness direction it is instructive to
  consider 
typed process 
$x{:}\one\lolli\one \vdash
      \ov{x}\langle y \rangle.(\nub z)(z(w).\zero \mid \ov{z}\langle
      w\rangle.\zero) :: v{:}\one$ 
and its translation:
  \[
\small
    \begin{array}{c}
      \llp \ov{x}\langle y \rangle.(\nub z)(z(w).\zero \mid \ov{z}\langle w\rangle.\zero) \rrp =
      \llp (\nub z)(z(w).\zero \mid \ov{z}\langle w\rangle.\zero) \rrp \{(x\,\munit)
      /x\}\\
      = \llet{\one = (\lambda w{:}\one.\llet{\one = w}{\munit})\,\munit}{\llet{\one = x\,\munit}\munit}
    \end{array}
  \]
  The process above cannot reduce due to the output prefix on
  $x$, which cannot synchronise with a corresponding input action
  since there is no provider for $x$ (i.e. the channel is in the
  left-hand side context). However, its encoding can exhibit the $\beta$-redex
  corresponding to the synchronisation along $z$, hidden by the prefix
  on $x$. 
  The corresponding reductions hidden under prefixes in the encoding can 
  be \emph{soundly} exposed in the session calculus by appealing to the
  commuting conversions of linear logic (e.g. in the
  process above, the instance of rule $\lft\lolli$ corresponding to
  the output on $x$ can be commuted with the $\m{cut}$ on $z$).}

\B{  As shown in
  \cite{DBLP:conf/esop/PerezCPT12}, commuting conversions are sound
  wrt observational equivalence, and thus we formulate operational
  soundness through a notion of \emph{extended} process reduction,
  which extends process reduction with the reductions that are induced
  by commuting conversions. Such a relation was also used for similar
  purposes in \cite{DBLP:journals/acta/BergerHY05} and in
  \cite{DBLP:conf/icfp/LindleyM16}, in a classical linear logic
  setting. For conciseness, we define extended reduction as a relation
on \emph{typed} processes modulo $\equiv$.}

\begin{definition}[Extended Reduction \cite{DBLP:journals/acta/BergerHY05}]\rm
We define $\mapsto$ as the type preserving
relations on typed processes modulo $\equiv$ generated by: 
\begin{enumerate}
\item 
{\small$\mathcal{C}[(\nub y)x\langle y\rangle.P] \mid x(y).Q \mapsto \mathcal{C}[(\nub y)(P
\mid Q)]$}; 
\item 
{\small$\mathcal{C}[(\nub y)x\langle y \rangle.P] \mid\, \bang x(y).Q \mapsto \mathcal{C}[(\nub
  y)(P \mid Q)] \mid \,\bang x(y).Q$; and (3)
 $(\nub x)(\bang x(y).Q) \mapsto \mathbf{0}$}
\end{enumerate}
where $\mathcal{C}$ is a (typed) process context
which does not capture the bound
name $y$.
\end{definition}

\begin{restatable}[Operational Soundness]{theorem}{thmopcc}
\label{thm:opc2}
Let $\Omega ; \Ga ; \D \vdash P :: z{:}A$ and
$\llp P \rrp \tra{} M$, there exists $Q$ such that $P\mapsto^* Q$ and
  $\llp Q \rrp =_\alpha M$.
\end{restatable}
\subsection{Inversion and Full Abstraction}\label{sec:fullabs}
Having established the operational preciseness of the encodings 
to-and-from Poly$\pi$ and Linear-F, 
we establish our main results for the encodings. Specifically, we
show that the encodings are mutually inverse up-to behavioural 
equivalence (with {\em fullness} as its corollary), which then enables us to establish {\em full
  abstraction} for \emph{both} encodings. 

\begin{restatable}[Inverse]{theorem}{thminv}
\label{thm:inv}\label{thm:inv1}\label{thm:inv2}~
If $\Omega ; \Ga ; \D \vdash M : A$ then $\Omega ; \Ga ; \D \vdash \llp\lb
M \rb_z\rrp \cong M : A$. Also,
if $\Omega ; \Ga ; \D \vdash P :: z{:}A$ then $\Omega ; \Ga ; \D \vdash
\lb\llp  P\rrp\rb_z \logsim P :: z{:}A$ 
\end{restatable}




\begin{corollary}[Fullness]
\label{cor:fullb}
~
Let $\Omega ; \Ga ; \D \vdash P :: z{:}A$. $\exists M$ s.t.
$\Omega ; \Ga ; \D \vdash M : A$ and $\Omega ; \Ga ; \D \vdash \lb M
\rb_z \obseq P :: z{:}A$
Also, let $\Omega ; \Ga ; \D \vdash M : A$. $\exists P$ s.t.
$\Omega ; \Ga ; \D \vdash P :: z{:}A$ and $\Omega ; \Ga ; \D \vdash\llp
P \rrp \cong M : A$ 

\end{corollary}



We now state our full abstraction results. 
Given two Linear-F terms of the same type,
equivalence in the image of the $\lb{-}\rb_z$ translation can be used
as a proof technique for contextual equivalence in Linear-F. This is
called the \emph{soundness} direction of full abstraction in
the literature \cite{DBLP:journals/mscs/GorlaN16} and proved 
by showing 
the relation generated by $\lb M \rb_z \logsim \lb N \rb_z$  forms 
$\cong$;  we then establish the \emph{completeness} direction 
by contradiction, using fullness. 




\begin{restatable}[Full Abstraction]{theorem}{thmfaltp}
\label{thm:fa_ltp1}
\label{thm:fa_ltp2}
$\Omega ; \Ga ; \D \vdash M \cong N : A$ iff $\Omega ; \Ga ; \D \vdash
\lb M \rb_z \logsim \lb N \rb_z :: z{:}A$. 
\end{restatable}
We can straightforwardly combine the above full abstraction with
Theorem~\ref{thm:inv} to obtain 
full abstraction of the $\llp{-}\rrp$ translation. 


\begin{restatable}[Full Abstraction]{theorem}{thmfaptl}
  \label{thm:fa_ptl}
$\Omega ; \Ga ; \D \vdash P \logsim Q :: z{:}A$ iff $\Omega ; \Ga ; \D \vdash \llp P \rrp \cong \llp Q \rrp : A$.
\end{restatable}


\section{Applications of the Encodings}\label{sec:apps}
In this section we develop applications of the encodings of the
previous sections. Taking advantage of full abstraction and mutual
inversion, we apply non-trivial properties from the theory of the
$\lambda$-calculus to our session-typed process setting.

In \S~\ref{sec:indcoind} we study 
 inductive and
coinductive sessions, arising through encodings of initial
$F$-algebras and final $F$-coalgebras in the polymorphic
$\lambda$-calculus.

In \S~\ref{sec:hovals} we study encodings for an extension of the core
session calculus with term passing, where terms are derived from a
simply-typed $\lambda$-calculus. Using the development of
\S~\ref{sec:hovals} as a stepping stone, we generalise the encodings
to a \emph{higher-order} session calculus (\S~\ref{sec:hopi}), where
processes can send, receive and execute other processes.
We show full abstraction and mutual
inversion theorems for the encodings from higher-order to first-order.
As a consequence, we can straightforwardly derive a strong
normalisation property for the higher-order process-passing calculus.

\subsection{Inductive and Coinductive Session Types}\label{sec:indcoind}
The study of polymorphism in the $\lambda$-calculus
\cite{DBLP:journals/tcs/BainbridgeFSS90,DBLP:journals/mscs/Hasegawa94,DBLP:conf/tlca/PlotkinA93,DBLP:journals/lmcs/BirkedalMP06}
has shown that parametric polymorphism is expressive enough to encode
both inductive and coinductive types in a precise way, through a
faithful representation of initial and final (co)algebras
\cite{DBLP:conf/lics/Mendler87}, without extending the language of
terms nor the semantics of the calculus, \B{giving a logical
  justification to the Church encodings of inductive datatypes such as
  lists and natural numbers}.  The polymorphic session calculus can
express fairly intricate communication behaviours, including generic
protocols through both existential and universal polymorphism
(i.e. protocols that are parametric in their sub-protocols).  Using
our fully abstract encodings between the two calculi, we show that
session polymorphism is expressive enough to encode inductive and
coinductive sessions, ``importing'' the results for the
$\lambda$-calculus, \B{which may then be instantiated to provide a
  session-typed formulation of the encodings of data structures in the
  $\pi$-calculus of
  \cite{DBLP:journals/iandc/MilnerPW92}.}

\noindent\mypara{Inductive and Coinductive Types in System F}
Exploring an algebraic interpretation of polymorphism where types are
interpreted as functors, it can be shown that given a type $F$ with a
free variable $X$ that occurs only positively (i.e. occurrences of $X$
are on the left-hand side of an even number of function arrows), the
polymorphic type $\forall X.((F(X) \rightarrow X) \rightarrow X)$
forms an initial $F$-algebra
\cite{DBLP:journals/iandc/ReynoldsP93,DBLP:journals/tcs/BainbridgeFSS90}
(we write $F(X)$ to denote that $X$ occurs in $F$). This enables the
representation of \emph{inductively} defined structures using an
algebraic or categorical justification. For instance, the natural
numbers can be seen as the initial $F$-algebra of $F(X) = \one + X$
(where $\one$ is the unit type and $+$ is the coproduct), and are thus
\emph{already present} in System F, in a precise sense, as the type
$\forall X.((\one + X) \rightarrow X) \rightarrow X$ (noting that both
$\one$ and $+$ can also be encoded in System F). A similar story can
be told for \emph{coinductively} defined structures, which correspond
to final $F$-coalgebras and are representable with the polymorphic
type $\exists X. (X \rightarrow F(X)) \times X$, where $\times$ is a
product type. In the remainder of this section we assume the
positivity requirement on $F$ mentioned above.

While the complete formal development of the representation of
inductive and coinductive types in System F would lead us to far
astray, we summarise here the key concepts as they apply to the
$\lambda$-calculus (the interested reader can
refer to \cite{DBLP:journals/mscs/Hasegawa94} for the full categorical details). 
\begin{figure}[t]
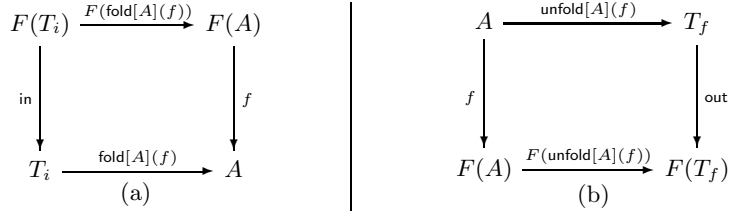

\centering 
\begin{tabular}{c|c}
\begin{tabular}{c}
\begin{diagram}
F(T_i) & \rTo^{F(\mathsf{fold}[A](f))} & F(A)\\
\dTo^{\mathsf{in}} & & \dTo_{{f}}\\
T_i & \rTo^{\mathsf{fold}[A](f)} & A\\ 
\end{diagram}
\\
(a)
\end{tabular}
\quad\quad\quad\quad&\quad\quad\quad\quad
\begin{tabular}{c}
\begin{diagram}
A & \rTo^{\mathsf{unfold}[A](f)} & T_f\\
\dTo^f & & \dTo_{\mathsf{out}}\\
F(A) & \rTo^{F(\mathsf{unfold}[A](f))} & F(T_f)\\ 
\end{diagram}\\
(b)
\end{tabular}
\end{tabular}
\caption{Diagrams for Initial $F$-algebras and Final $F$-coalgebras \label{fig:diagrams}
}
\end{figure}

To show that the polymorphic type $T_i \triangleq \forall X.((F(X) \rightarrow X)
\rightarrow X)$ is an initial $F$-algebra, one exhibits a
pair of $\lambda$-terms, often dubbed $\m{fold}$ and $\m{in}$, such that
the diagram in Fig.~\ref{fig:diagrams}(a) commutes (for any $A$, where $F(f)$, where $f$
is a $\lambda$-term, denotes the
functorial action of $F$ applied to $f$),  
and, crucially, that $\m{fold}$ is \emph{unique}. When these
conditions hold, we are justified in saying that $T_i$ is a least
fixed point of $F$. Through a fairly simple calculation, it is easy to
see that: 
\[\small
\begin{array}{rcl}
\m{fold}  & \triangleq  & \Lambda X.\lambda x{:}F(X)\rightarrow
X.\lambda t{:}T_i.t[X](x)\\ 
\m{in} & \triangleq & \lambda
x{:}F(T_i).\Lambda X.\lambda y{:} F(X)\rightarrow
X.y\,(F(\m{fold}[X](x))(x))
\end{array}
\]
satisfy the necessary
equalities. To show uniqueness one appeals to \emph{parametricity},
which allows us to prove that any function of the appropriate type
is equivalent to $\m{fold}$. This property is often dubbed initiality
or universality.

The construction of final $F$-coalgebras and their justification as
\emph{greatest} fixed points is dual. Assuming products in the
calculus and taking $T_f \triangleq \exists X. (X
\rightarrow F(X)) \times X$, we produce the $\lambda$-terms
\[\small
\begin{array}{rcl}
\m{unfold} & \triangleq & \Lambda X.\lambda f{:}X\rightarrow F(X).\lambda
x{:}T_f.\pack{X}{(f,x)}\\ 
\m{out} & \triangleq & 
\lambda t : T_f .\llet{(X,(f,x)) = t}{F(\m{unfold}[X](f))\,(f(x))}
\end{array}
\]
such that the diagram in Fig.~\ref{fig:diagrams}(b) commutes and $\m{unfold}$ is unique
(again, up to parametricity).
%
While the argument above applies to System F, a similar development
can be made in Linear-F \cite{DBLP:journals/lmcs/BirkedalMP06} by
considering $T_i \triangleq  \forall X.\bang (F(X) \lolli X) \lolli X$ and $T_f
\triangleq  \exists X.\bang(X\lolli F(X)) \tensor X$. Reusing the same
names for the sake of conciseness, the associated
\emph{linear} $\lambda$-terms are:
\[\small
\begin{array}{rcl}
\m{fold} & \triangleq & \Lambda X.\lambda u{:}\bang (F(X)\lolli
X).\lambda y{:}T_i.(y[X]\,u) : \forall X.\bang (F(X)\lolli X) \lolli T_i \lolli X\\
\m{in} & \triangleq & \lambda x{:}F(T_i).\Lambda X.\lambda
y{:}\bang(F(X)\lolli X).\llet{\bang u = y}{k\,(F\,(\m{fold}[X](\bang
  u))(x))}
: 
F(T_i) \lolli T_i\\
\m{unfold} & \triangleq & \Lambda X.\lambda u{:}\bang(X\lolli F(X)).\lambda x{:}X.
\pack{X}{ \mpair{u}{x}} : \forall X.\bang(X\lolli F(X))\lolli X \lolli T_f\\
\m{out} 
& \triangleq & 
\lambda t : T_f .\llet{(X,(u,x)) = t}{
\llet{\bang f = u}{F(\m{unfold}[X](\bang f))\,(f(x))}} : T_f \lolli F(T_f)
\end{array}
\]

\mypara{Inductive and Coinductive Sessions for Free}
As a consequence of full abstraction we may appeal to the
$\lb{-}\rb_z$ encoding to derive representations of $\m{fold}$ and
$\m{unfold}$ that satisfy the necessary algebraic properties. The derived
processes are (recall that we write $\ov{x}\langle y
\rangle.P$ for $(\nub y)x\langle y \rangle.P$):
\[\small
\begin{array}{rcl}
\lb \m{fold} \rb_z & \triangleq & 
z(X).z(u).z(y).(\nub w)((\nub x)([y\leftrightarrow x] \mid 
x\langle X\rangle.[x\leftrightarrow w]) \mid \ov{w}\langle v \rangle.(
[u\leftrightarrow v] \mid [w\leftrightarrow z]) )
\\ %
\lb \m{unfold} \rb_z & \triangleq & 
z(X).z(u).z(x).z\langle X \rangle.\ov{z}\langle y \rangle.([u \leftrightarrow y] \mid [x\leftrightarrow z])\\
\end{array}
\]

We can then show universality of the two constructions.
We write $P_{x,y}$ to single out that $x$ and
$y$ are free in $P$ and $P_{z,w}$ to denote the result of 
employing capture-avoiding substitution on $P$, substituting $x$ and $y$ by
$z$ and $w$.
Let:
\[\small
\begin{array}{rcl}
\m{foldP}(A)_{y_1,y_2} &\triangleq &(\nub x)(\lb \m{fold}\rb_x \mid 
x\langle A\rangle.\ov{x}\langle v\rangle.(\ov{u}\langle y\rangle.[y\leftrightarrow v] \mid 
\ov{x}\langle z\rangle.([z\leftrightarrow y_1] \mid [x\leftrightarrow y_2])))\\
\m{unfoldP}(A)_{y_1,y_2} & \triangleq & (\nub x)(\lb\m{unfold}\rb_x \mid x\langle A\rangle.
\ov{x}\langle v\rangle.(\ov{u}\langle y \rangle.[y\leftrightarrow v] \mid 
\ov{x}\langle z \rangle.([z\leftrightarrow y_1] \mid [x\leftrightarrow y_2])))
\end{array}
\]
where $\m{foldP}(A)_{y_1,y_2}$ corresponds to the application of $\m{fold}$ to an
$F$-algebra $A$ with the associated morphism $F(A)\lolli A$ available
on the shared channel $u$, consuming an ambient session $y_1{:}T_i$ 
and offering $y_2{:}A$. Similarly, $\m{unfoldP}(A)_{y_1,y_2}$ corresponds to
the application of $\m{unfold}$ to an $F$-coalgebra $A$ with the
associated morphism $A \lolli F(A)$ available on the shared channel
$u$, consuming an ambient session $y_1{:}A$ and offering $y_2{:}T_f$.

\begin{theorem}[Universality of $\m{foldP}$]~\label{thm:unifold}
$\forall Q$ such that $X;u{:}F(X)\lolli X;y_1{:}T_i \vdash Q ::
y_2{:}X$ we have
  $X ; u{:}F(X)\lolli X;y_1{:}T_i \vdash Q \logsim
 \m{foldP}(X)_{y_1,y_2} :: y_2 {:} X$
\end{theorem}

\begin{theorem}[Universality of $\m{unfoldP}$]~\label{thm:uniunfold}
$\forall Q$ and $F$-coalgebra $A$ 
s.t $\cdot ; \cdot ; y_1{:}A \vdash Q :: y_2 {:}T_f$
we have that 
$\cdot ; u{:}F(A)\lolli A ; y_1{:}A \vdash Q \logsim \m{unfoldP}(A)_{y_1,y_2} ::
y_2 :: T_f$.
\end{theorem}

\begin{example}[Natural Numbers]
\label{sec:exnat}
We show how to represent the natural numbers as an inductive session type
using $F(X) = \one \oplus X$, making use of $\m{in}$: 
\[\small
\begin{array}{c}
\m{zero}_x  \triangleq (\nub z)(z.\m{inl};\zero \mid \lb \m{in}(z)\rb_x ) \quad
\m{succ}_{y,x} \triangleq  (\nub s)(s.\m{inr};[y\leftrightarrow s]
                              \mid \lb \m{in}(s)\rb_x)
\end{array}
\]
with $\m{Nat} \triangleq \forall X.\bang ((\one \oplus X) \lolli X)
\lolli X$ 
where $\vdash \m{zero}_x :: x{:}\m{Nat}$ and $y{:} \m{Nat} \vdash
\m{succ}_{y,x} :: x{:}\m{Nat}$ encode the representation of $0$ and
successor, respectively. The natural $1$ would thus be
represented by
$\m{one}_x \triangleq (\nub y)(\m{zero}_y \mid
\m{succ}_{y,x})$.
The behaviour of type $\m{Nat}$ can be
seen as a that of a sequence of internal choices of arbitrary (but
finite) length. 
We can then observe that the $\m{foldP}$
process acts as a recursor. For instance consider:
\[\small
\begin{array}{l}
\m{stepDec}_d  \triangleq  d(n).n.\m{case}( \m{zero}_d ,
                         [n\leftrightarrow d]  ) \quad
\m{dec}_{x,z} \triangleq (\nub u)(\bang u(d).\m{stepDec}_d \mid
                          \m{foldP}(\m{Nat})_{x,z})
\end{array}
\]
with $\m{stepDec}_d :: d{:}(\one \oplus \m{Nat}) \lolli \m{Nat}$ and
$x {:} \m{Nat} \vdash \m{dec}_{x,z} :: z{:}\m{Nat}$, where 
$\m{dec}$ decrements a given natural number session
on channel $x$. We have that:
\[\small
\begin{array}{l}
(\nub x)(\m{one}_x \mid \m{dec}_{x,z}) \equiv 
(\nub x,y.u)(\m{zero}_y \mid
\m{succ}_{y,x}\bang u(d).\m{stepDec}_d \mid
                          \m{foldP}(\m{Nat})_{x,z})
\logsim \m{zero}_z
\end{array}
\]

\B{We note that the resulting encoding is reminiscent of the encoding
of lists of \cite{DBLP:journals/iandc/MilnerPW92} (where $\m{zero}$
is the empty list and $\m{succ}$ the cons cell). The main
differences in the encodings arise due to our primitive notions of
labels and forwarding, as well as due to the generic nature of
$\m{in}$ and $\m{fold}$. }
\end{example}

\begin{example}[Streams]
\label{ex:streams}
We build on Example~\ref{sec:exnat} by representing
\emph{streams} of natural numbers as a coinductive session type.  
We encode infinite streams of
naturals with $F(X) = \m{Nat} \tensor X$. Thus: 
$\m{NatStream}  \triangleq  \exists X.\bang (X\lolli
                                     (\m{Nat} \tensor X)) \tensor X$.
The behaviour of a session of type $\m{NatStream}$ amounts to an
infinite sequence of outputs of channels of type $\m{Nat}$.
Such an encoding enables us to construct the stream of all naturals
$\m{nats}$ (and the stream of all non-zero naturals $\m{oneNats}$):
\[\small
\begin{array}{lcl}
\m{genHdNext}_z & \triangleq 
& z(n).\ov{z}\langle y
                             \rangle.(\ov{n}\langle
                             n'\rangle.[n'\leftrightarrow y] \mid
                             \,\bang z(w).\ov{n}\langle n' \rangle.\m{succ}_{n',w})\\
\m{nats}_y & \triangleq & (\nub x,u) (\m{zero}_x \mid \,\bang
                          u(z).\m{genHdNext}_z \mid  \m{unfoldP}(\bang\m{Nat})_{x,y})\\
\m{oneNats}_y& \triangleq & (\nub x,u) (\m{one}_x \mid \,\bang
                          u(z).\m{genHdNext}_z \mid  \m{unfoldP}(\bang\m{Nat})_{x,y})
\end{array}
\]
with $\m{genHdNext}_z :: z{:} \bang \m{Nat} \lolli \m{Nat}
\tensor \bang\m{Nat}$ and both $\m{nats}_y$ and $\m{oneNats} :: y{:} \m{NatStream}$. $\m{genHdNext}_z$
consists of a helper that generates the current head of a stream and
the next element. As expected, the following process 
implements a session that ``unrolls'' the stream once, providing
the head of the stream and then behaving as the rest of
the stream (recall that $\m{out} : T_f \lolli F(T_f)$).
\begin{center}
$
(\nub x)(\m{nats}_x \mid \lb \m{out}(x)\rb_y) :: y{:} \m{Nat} \tensor \m{NatStream}
$
\end{center}

We note a peculiarity of the interaction of linearity with the 
stream encoding: a process that begins to
deconstruct a stream has no
way of ``bottoming out'' and stopping. One cannot, for instance,
extract the first element of a stream of naturals and stop unrolling
the stream in a well-typed way.
We can, however, easily encode a
``terminating'' stream of all natural numbers via $F(X) = (\m{Nat}
\tensor \bang X)$ by 
replacing the $\m{genHdNext}_z$ with
the generator given as:
\[
\small
\begin{array}{rcl}
\m{genHdNextTer}_z 
& \triangleq &
z(n).\ov{z}\langle y
                             \rangle.(\ov{n}\langle
                             n'\rangle.[n'\leftrightarrow y] \mid
                             \,\bang z(w).\bang w(w').\ov{n}\langle n'
                             \rangle.\m{succ}_{n',w'})
\end{array}
\]
It is then easy to see that a usage of $\lb\m{out}(x)\rb_y$
results in a session of type $\m{Nat} \tensor \bang\m{NatStream}$,
enabling us to discard the stream as needed. 
One can replay this
argument with the operator $F(X) = (\bang\m{Nat} \tensor 
X)$ to enable discarding of stream elements. Assuming such
modifications, we can then show:
\[
(\nub y)((\nub x)(\m{nats}_x \mid \lb \m{out}(x)\rb_y) \mid y(n).[y\leftrightarrow z]) 
\logsim \m{oneNats}_z :: z{:}\m{NatStream}
\]
\end{example}
\subsection{Communicating Values -- \valpi}\label{sec:hovals}
We now consider a session calculus extended with a data layer 
obtained from a $\lambda$-calculus (whose terms are ranged over by
$M,N$ and types by $\tau,\sigma$). We dub this calculus \valpi.
\[
\small
\begin{array}{ll}
\begin{array}{lcl}
P,Q & ::= & \dots \mid x\langle M \rangle.P \mid x(y).P\\
M,N & ::= & \lambda x {:}\tau.M \mid M\,N \mid x
\end{array}
&\quad\quad 
\begin{array}{lcl}
A,B & ::= & \dots \mid \tau \wedge A \mid \tau \supset A\\
\tau,\sigma & ::= & \dots \mid \tau \rightarrow \sigma 
\end{array}
\end{array}
\]
Without loss of generality, we consider the data layer to be
simply-typed, with a call-by-name semantics, satisfying the usual type
safety properties. The typing judgment for this calculus is
$\Psi \vdash M : \tau$.  We omit session polymorphism for the sake of
conciseness, restricting processes to communication of data and
(session) channels.  The typing judgment for processes is thus
modified to $\Psi ; \Ga ; \D \vdash P :: z{:}A$, where $\Psi$ is an
intuitionistic context that accounts for variables in the data
layer. The rules for the relevant process constructs are (all other
rules simply propagate the $\Psi$ context from conclusion to
premises):
\[
\small
\begin{array}{c}
\infer[(\rgt{\wedge})]
{\Psi ; \Ga ; \D \vdash z\langle M \rangle.P :: z{:}\tau \wedge A}
{\Psi \vdash M : \tau & \Psi ; \Ga ; \D \vdash P :: z{:}A}
\quad
\infer[(\lft{\wedge})]
{\Psi ; \Ga ; \D , x{:}\tau\wedge A \vdash x(y).Q :: z{:}C}
{\Psi , y{:}\tau ; \Ga ; \D , x{:}A \vdash Q :: z{:}C}\\[0.5em]
\infer[(\rgt\supset)]
{\Psi ; \Ga ; \D \vdash z(x).P :: z{:}\tau \supset A}
{\Psi , x{:}\tau ; \Ga ; \D \vdash P :: z{:}A}
\quad
\infer[(\lft\supset)]
{\Psi ; \Ga ; \D , x{:}\tau\supset A \vdash x \langle M\rangle.Q :: z{:}C}
{\Psi \vdash M : \tau & \Psi  ; \Ga ; \D , x{:}A \vdash Q :: z{:}C}
\end{array}
\]
With the reduction rule given by:\footnote{For simplicity, in this
  section, we define
  the process semantics through a reduction relation.}
$
x\langle M \rangle. P \mid x(y).Q \tra{} P \mid Q\{M/y\}
$.
With a simple extension to our encodings we may eliminate the data layer by encoding
the data objects as processes, showing that from an expressiveness
point of view, data communication is orthogonal to the framework.  We
note that the data language we are considering is \emph{not} linear,
and the usage discipline of data in processes is itself also not linear.

%
\paragraph{\bf To First-Order Processes}
We now introduce our encoding for \valpi, defined inductively on session types,
processes, types and $\lambda$-terms (we omit the purely inductive
cases on session types and processes for conciseness). As before, the
encoding on processes is defined on \emph{typing derivations}, where
we indicate the typing rule at the root of the typing derivation. 
\[
\small
\begin{array}{l}
\lb \tau \wedge A \rb  \triangleq  \bang \lb \tau \rb \tensor \lb A
\rb  \qquad 
\lb \tau \supset A \rb  \triangleq  \bang \lb \tau \rb \lolli \lb A
\rb
\qquad 
 \lb \tau \rightarrow \sigma \rb  \triangleq  
\bang \lb \tau \rb \lolli \lb \sigma \rb
\end{array}
\]
\[
\small
\begin{array}{lrcllrcl}
(\rgt{\wedge}) & \lb z\langle M \rangle . P \rb & \triangleq & \ov{z}\langle x \rangle.(
\bang x(y).\lb M \rb_y \mid \lb P \rb) \quad&
(\lft\wedge) & \lb x(y).P \rb & \triangleq & 
 x(y).\lb P \rb\\
(\rgt\supset) & \lb z(x).P \rb & \triangleq & z(x).\lb P \rb &
(\lft\supset) & \lb x\langle M\rangle.P \rb & \triangleq & \ov{x}\langle y \rangle.( \bang y(w).\lb M\rb_w \mid \lb P \rb)
\end{array}
\]
\[
\small
\begin{array}{lll}
 \lb x \rb_z  \triangleq  \ov{x}\langle y \rangle.[y \leftrightarrow z]  
\quad \quad \quad \lb \lambda x {:} \tau . M \rb_z  \triangleq   z(x).\lb M \rb_z\\
\lb M \, N\rb_z  \triangleq  (\nub y)(\lb M \rb_y \mid 
\ov{y}\langle x \rangle.(\bang x(w).\lb N \rb_w \mid [y\leftrightarrow z]) )
\end{array}
\]
The encoding addresses the non-linear usage of data elements in
processes by encoding the types $\tau \wedge A$ and $\tau \supset A$
as $\bang \lb \tau \rb \tensor \lb A\rb$ and $\bang \lb \tau\rb \lolli
\lb A \rb$, respectively. Thus, sending and receiving of data is
codified as the sending and receiving of channels of type $\bang$,
which therefore can be used non-linearly. Moreover, since data
terms are themselves non-linear, the $\tau \rightarrow \sigma$ type
is encoded as $\bang \lb \tau\rb \lolli
\lb \sigma\rb$, following Girard's embedding of
intuitionistic logic in linear logic \cite{DBLP:journals/tcs/Girard87}.

At the level of processes, offering a session of type $\tau \wedge A$
(i.e. a process of the form $z\langle M \rangle.P$) is encoded
according to the translation of the type: we first send a \emph{fresh}
name $x$ which will be used to access the encoding of the term
$M$. Since $M$ can be used an arbitrary number of times by the
receiver, we guard the encoding of $M$ with a replicated input,
proceeding with the encoding of $P$ accordingly. Using a session of
type $\tau \supset A$ follows the same principle. The input cases (and
the rest of the process constructs) are completely homomorphic.

The encoding of $\lambda$-terms follows 
Girard's decomposition
of the intuitionistic function space 
 \cite{DBLP:conf/fossacs/ToninhoCP12}.  
The $\lambda$-abstraction is
translated as input. Since variables in a
$\lambda$-abstraction may be used non-linearly, the case for variables
and application is slightly more intricate: to encode the application
$M\, N$ we compose $M$ in parallel with a process that will send the
``reference'' to the function argument $N$ which will be encoded using
replication, in order to handle the potential for $0$ or more usages
of variables in a function body. Respectively, a variable is encoded
by performing an output to trigger the replication and forwarding
accordingly. Without loss of generality, we assume variable names and their
corresponding replicated counterparts match, which can be achieved
through $\alpha$-conversion before applying the
translation. 
We exemplify our encoding as follows:
\[
\small
\begin{array}{c}
\lb z(x).z\langle x \rangle.z\langle (\lambda y{:}\sigma
. x)\rangle.\zero \rb  =  
z(x).\ov{z}\langle w \rangle.(\bang w(u).\lb x\rb_u \mid \ov{z}\langle v\rangle.(
\bang v(i).\lb \lambda y{:}\sigma.x\rb_i \mid \zero))\\ 
\qquad\qquad\qquad\qquad =  
z(x).\ov{z}\langle w\rangle.(\bang w(u).\ov{x}\langle y \rangle.[y\leftrightarrow u] \mid 
\ov{z}\langle v \rangle.(\bang v(i).i(y).\ov{x}\langle t\rangle.[t\leftrightarrow i] \mid \zero))
\end{array}
\]

\noindent\mypara{Properties of the Encoding} 
We discuss the correctness of our encoding. 
We can straightforwardly establish that the encoding
preserves typing.

To show that our encoding is operationally sound and complete, we
capture the interaction between substitution on $\lambda$-terms and
the encoding into processes through logical equivalence.  Consider
the following reduction of a process:
\begin{eqnarray}
\small
& & (\nub z)(z(x).z\langle x \rangle.z\langle (\lambda y{:}\sigma . x)\rangle.\zero
\mid z\langle \lambda w{:}\tau_0.w \rangle.P)\nonumber\\ 
\label{eq:valred}
& & \hspace*{3cm}\tra{}
(\nub z)(z\langle \lambda w{:}\tau_0.w  \rangle.z\langle (\lambda y{:}\sigma . \lambda w{:}\tau_0.w )\rangle.\zero
\mid P)
\end{eqnarray}

\noindent Given that substitution in the target 
session $\pi$-calculus amounts
to renaming, whereas in the $\lambda$-calculus we replace 
a variable for a term, the relationship between
the encoding of a substitution $M\{N/x\}$ and the encodings of $M$ and
$N$ corresponds to the composition of the encoding of $M$ with that of
$N$, but where the encoding of $N$ is guarded by a replication,
codifying a form of explicit non-linear substitution.

 \begin{restatable}[Compositionality]{lemma}{lemcompvals}
 \label{lem:comp}~
 Let $\Psi , x{:}\tau \vdash M : \sigma$ and $\Psi \vdash N : \tau$. We
 have that $\lb M\{N/x\}\rb_z \logsim (\nub x)(\lb M\rb_z \mid \bang
 x(y).\lb N\rb_y)$
 \end{restatable}

Revisiting the process to the left of the arrow in
Equation~\ref{eq:valred} we have:
\[
\small
\begin{array}{l}
\lb (\nub z)(z(x).z\langle x \rangle.z\langle (\lambda y{:}\sigma . x)\rangle.\zero
\mid z\langle \lambda w{:}\tau_0.w \rangle.P) \rb\\ = 
(\nub z)(\lb z(x).z\langle x \rangle.z\langle (\lambda y{:}\sigma . x)\rangle.\zero \rb_z \mid \ov{z}\langle x \rangle.(\bang x(b).\lb \lambda w{:}\tau_0 .w\rb_b \mid \lb P \rb))\\ \tra{}
(\nub z,x)(\ov{z}\langle w \rangle.(\bang w(u).\ov{x}\langle y \rangle.[y\leftrightarrow u] \mid \ov{z}\langle v\rangle.(
\bang v(i).\lb \lambda y{:}\sigma.x\rb_i \mid \zero) \mid \,\bang x(b).\lb \lambda w{:}\tau_0 .w\rb_b \mid \lb P \rb))
\end{array}
\]
whereas the process to the right of the  arrow is encoded as:
\[
\small
\begin{array}{l}
\lb (\nub z)(z\langle \lambda w{:}\tau_0.w  \rangle.z\langle (\lambda y{:}\sigma . \lambda w{:}\tau_0.w )\rangle.\zero
\mid P)\rb \\ = (\nub z)(\ov{z}\langle w\rangle.( \bang w(u).\lb \lambda w{:}\tau_0.w\rb_u \mid  \ov{z}\langle v\rangle.(\bang v(i).\lb \lambda y{:}\sigma . \lambda w{:}\tau_0.w \rb_i \mid \lb P \rb)))
\end{array}
\]
While the reduction of the encoded process and the encoding of the
reduct differ syntactically, they are observationally equivalent --
the latter inlines the replicated process behaviour that is accessible
in the former on $x$. Having characterised 
substitution, we establish  operational \B{correspondence for the
  encoding.}
\vspace{-3mm}
\B{
  \begin{theorem}[Operational Correspondence]\label{thm:opcorrvals}
    \begin{enumerate}
\item If $\Psi \vdash M : \tau$ and $\lb M\rb_z \tra{} Q$ then
$M \tra{}^+ N$ such that $\lb N \rb_z \logsim Q$
\item If $\Psi ; \Ga ; \D \vdash P :: z{:}A$ and $\lb P \rb \tra{} Q$ then
  $P \tra{}^+ P'$ such that $\lb P'\rb \logsim Q$
  \item If $\Psi \vdash M : \tau$ and $M \tra{} N$ then $\lb M \rb_z
  \wtra{} P$ such that $P \logsim \lb N\rb_z$
\item If $\Psi ; \Ga ; \D \vdash P :: z{:}A$ and $P \tra{} Q$ then
$\lb P \rb \tra{}^+ R$ with $R \logsim \lb Q \rb$
    \end{enumerate}
  \end{theorem}
}



The process equivalence in
Theorem~\ref{thm:opcorrvals} above need not be
extended to account for data (although it would be relatively simple to
do so), since the processes in the image of the encoding are fully
erased of any data elements.  

\paragraph{\bf Back to $\lambda$-Terms.}
We extend our encoding of processes to $\lambda$-terms
to \valpi. Our extended translation maps
processes to linear $\lambda$-terms, with the session type $\tau \wedge A$
interpreted as a pair type where the first component is
replicated. Dually, $\tau \supset A$ is interpreted as a function type
where the domain type is replicated.  The remaining session constructs
are translated as in \S~\ref{sec:pitof}.
\[
 \small
\begin{array}{c}
  \llp \tau \wedge A \rrp \triangleq  \ \bang \llp \tau \rrp \tensor
  \llp A \rrp
  \quad\quad
\llp \tau \supset A \rrp  \triangleq  \ \bang \llp \tau \rrp \lolli \llp  A \rrp
\quad\quad\llp \tau \rightarrow \sigma \rrp  \triangleq  \ \bang \llp
\tau \rrp \lolli \llp \sigma \rrp
\end{array}
\]
\[
\small
\begin{array}{llclllcl}
(\lft\wedge) & \llp x(y).P \rrp & \triangleq & \llet{y\tensor x =
                                               x}{\llet{\bang y =
                                               y}{\llp P \rrp}}
&
(\rgt\wedge) & \llp z\langle M \rangle.P \rrp & \triangleq & 
 \langle \bang \llp M \rrp \tensor \llp P \rrp  \rangle 
\\
(\rgt\supset) & \llp x(y).P \rrp & \triangleq & 
  \lambda x{:}\bang \llp \tau \rrp.\llet{\bang x = x}{\llp P \rrp}  &
(\lft\supset) & \llp x\langle M \rangle.P \rrp & \triangleq & 
\llp P \rrp\{(x\,\bang\llp M \rrp)/x\}
\end{array}
\]
\[
\begin{array}{c}
\llp \lambda x{:}\tau. M\rrp  \triangleq  \lambda x{:}\bang \llp
                                              \tau \rrp .\llet{\bang x
                                              = x}{\llp M \rrp}

                                            \quad
 \llp M\,N \rrp  \triangleq  \llp M \rrp \,\bang\llp N\rrp
\quad \llp x \rrp  \triangleq  x
\end{array}
\]

The treatment of non-linear components of processes is identical to
our previous encoding: non-linear functions $\tau \rightarrow \sigma$
are translated to linear functions of type
$\bang \tau \lolli \sigma$; a process offering a session of type
$\tau \wedge A$ (i.e. a process of the form $z\langle M \rangle.P$,
typed by rule $\rgt\wedge$) is translated to a pair where the first
component is the encoding of $M$ prefixed with $\bang$ so that it
may be used non-linearly, and the second is the encoding of $P$.
Non-linear variables are handled at the respective
binding sites: a process using a session of type $\tau \wedge A$ is
encoded using the elimination form for the pair and the elimination
form for the exponential; similarly, a process offering a session of
type $\tau \supset A$ is encoded as a $\lambda$-abstraction where the
bound variable is of type $\bang \llp \tau \rrp$. Thus, we use the
elimination form for the exponential, ensuring that the typing is correct.
We illustrate our encoding:
\[
\small
\begin{array}{rcl}
\llp z(x).z\langle x \rangle.z\langle (\lambda y{:}\sigma . x)\rangle.\zero \rrp
& = & \lambda x{:}\bang \llp \tau \rrp.\llet{\bang x = x}{\mpair{\bang x\,}
{\mpair{\bang \llp \lambda y{:}\sigma.x\rrp\,}
{\langle\rangle}}}\\ 
& = & 
\lambda x{:}\bang \llp \tau \rrp.\llet{\bang x = x}{\mpair{\bang x\,}
{\mpair{\bang ( \lambda y{:}\bang\llp\sigma\rrp.\llet{\bang y = y}{x} )\,}
{\langle\rangle}}}
\end{array}
\]

\noindent\mypara{Properties of the Encoding}
Unsurprisingly due to the logical correspondence between natural
deduction and sequent calculus presentations of logic, our encoding
\B{satisfies both type soundness and operational correspondence (c.f. Theorems~\ref{thm:typsoundpitof},~\ref{thm:opc1},
  and~\ref{thm:opc2}). The full development can be found in \cite{longversion}.}

\vspace{-3mm}
\paragraph{{\bf Relating the Two Encodings.}}
We prove the two encodings are
mutually inverse and preserve the full abstraction properties (we
write $=_\beta$ and $=_{\beta\eta}$ for $\beta$- and
$\beta\eta$-equivalence, respectively).

\begin{restatable}[Inverse]{theorem}{thminvv}
  \label{thm:inv1}\label{thm:inv2}
If $\Psi ; \Ga ; \D \vdash P :: z{:}A$
then $\lb\llp P \rrp\rb_z \logsim \lb P\rb$. Also,
if $\Psi \vdash M : \tau$ then $\llp\lb M \rb_z\rrp =_{\beta} \llp M \rrp$.
\end{restatable}

The equivalences above are formulated between the
composition of the encodings applied to $P$ (resp. $M$) and the
process (resp. $\lambda$-term) \emph{after} applying the translation 
embedding the non-linear components into their linear counterparts.
This formulation matches more closely that of \S~\ref{sec:fullabs}, which
applies to linear calculi for which the \emph{target} languages of
this section are a strict subset (and avoids the formalisation of
process equivalence with terms). We also note that in this setting,
observational equivalence and $\beta\eta$-equivalence coincide 
\cite{barberdill96,DBLP:conf/rta/OhtaH06}. Moreover, the extensional
flavour of $\logsim$ includes $\eta$-like principles at the process level.

\begin{restatable}{theorem}{thmfav}
  \label{thm:fav}\label{thm:fav2}
Let $\cdot \vdash M : \tau$ and $\cdot \vdash N : \tau$. 
$ \llp M\rrp =_{\beta\eta} \llp N\rrp$ iff $\lb M \rb_z \logsim \lb N
\rb_z$. Also, let $\cdot \vdash P :: z{:}A$ and $\cdot \vdash Q ::
z{:}A$.
We have
that $\lb P\rb \logsim \lb Q\rb$ iff $\llp P \rrp =_{\beta\eta} \llp Q \rrp$.
\end{restatable}

We establish
full abstraction for the encoding of $\lambda$-terms into processes
(Theorem~\ref{thm:fav}) in two steps: The completeness direction
(i.e. from left-to-right) follows from operational completeness and
strong normalisation of the $\lambda$-calculus. The soundness direction
uses operational soundness. The proof of Theorem~\ref{thm:fav2} uses
the same strategy of Theorem~\ref{thm:fa_ptl}, appealing to the
inverse theorems. 

\subsection{Higher-Order Session Processes -- \hopi}\label{sec:hopi}

We extend the value-passing framework
of the previous section, accounting for process-passing 
(i.e.~the higher-order) in a session-typed setting. As
shown in \cite{DBLP:conf/esop/ToninhoCP13}, we achieve this by adding
to the data layer a \emph{contextual monad} that encapsulates (open)
session-typed processes as data values, with a corresponding
elimination form in the process layer. We dub this calculus \hopi.
\[
  \begin{array}{rclrcl}
    P,Q  & ::= &  \dots \mid x\leftarrow M\leftarrow \ov{y_i};Q & \quad\quad
    M.N  & ::= & \dots \mid \{x \leftarrow P \leftarrow
                \ov{y_i{:}A_i}\} \\
   \tau,\sigma & ::= & \dots \mid \{ \ov{x_j{:}A_j}
                        \vdash z{:}A\}
  \end{array}
\]
The type $\{ \ov{x_j{:}A_j} \vdash z{:}A\}$ is the type of a term
which encapsulates an open process that uses the linear channels
$\ov{x_j{:}A_j}$ and offers $A$ along channel $z$.  
This formulation has the added benefit of formalising
the integration of session-typed processes in a functional language
and forms the basis for the concurrent programming language {\tt SILL}
\cite{DBLP:conf/fossacs/PfenningG15,DBLP:conf/esop/ToninhoCP13}.  The
typing rules for the new constructs are (for simplicity we assume no
shared channels in process monads):
{\small\[
\begin{array}{c}
\infer[\{\}I]
{\Psi \vdash \{ z \leftarrow P \leftarrow \ov{x_i{:}A_i}\} : \{ \ov{x_i{:}A_i}
                        \vdash z{:}A\}}
{\Psi ; \cdot ;\ov{x_i{:}A_i} \vdash P :: z{:}A  }\\[0.5em]
\infer[\{\}E]
{\Psi ; \Ga  ; \D_1 , \D_2 \vdash x\leftarrow M \leftarrow \ov{y_i};Q :: z{:}C}
{\Psi \vdash M :\{ \ov{x_i{:}A_i} \vdash x{:}A\} & 
 \D_1 = \ov{y_i {:}A_i} &
 \Psi ; \Ga ; \D_2 , x{:}A \vdash Q :: z{:}C }
\end{array}
\]}

Rule $\{\}I$ embeds processes in the term language by essentially
quoting an open process that is well-typed according to the type
specification in the monadic type. Dually, rule $\{\}E$ allows for
processes to use monadic values through composition that
\emph{consumes} some of the ambient channels in order to provide the
monadic term with the necessary context (according to its type).
These constructs are discussed in substantial
detail in \cite{DBLP:conf/esop/ToninhoCP13}. 
The reduction semantics of the process
construct is given by (we tacitly assume that the names $\ov{y}$ and
$c$ do not occur in $P$ and omit the congruence case):
\[
\begin{array}{c}
(c\leftarrow \{ z \leftarrow P \leftarrow \ov{x_i{:}A_i}\} \leftarrow \ov{y_i};Q) \tra{}
  (\nub c)(P\{\ov{y}/\ov{x_i}\{c/z\}\} \mid Q)
\end{array}
\]
The semantics allows for the underlying monadic term $M$ to
evaluate to a (quoted) process $P$. The process $P$ is then executed in
parallel with the continuation $Q$, sharing the linear channel $c$ for
subsequent interactions.
We illustrate the higher-order extension with following typed process (we
write $\{x\leftarrow P\}$ when $P$ does not depend on any linear
channels and assume $\vdash Q :: d{:}\m{Nat}\wedge \one$):
\begin{equation}\label{eq:ho}
P  \triangleq  (\nub c)(c\langle \{d \leftarrow Q\}\rangle.c(x).\zero \mid c(y).
d \leftarrow y; d(n).c\langle n \rangle.\zero)
\end{equation}
Process $P$ above gives
an abstract view of a communication idiom where a process (the
left-hand side of the parallel composition) sends
another process $Q$ which potentially encapsulates some complex
computation. The receiver then \emph{spawns} the execution of
the received process and inputs from it a result value that is sent back to
the original sender.
An execution of $P$ is given by:
\[
\small
\begin{array}{rcl}
P \tra{} (\nub c)(c(x).\zero \mid d\leftarrow \{d \leftarrow Q\} ;
d(n).c\langle n\rangle.\zero) & \tra{} & 
(\nub c)(c(x).\zero \mid (\nub d)(Q \mid d(n).c\langle
n\rangle.\zero)) \\
& \tra{}^+ & 
(\nub c)(c(x).\zero \mid c\langle 42\rangle.\zero) \tra{} \zero
\end{array}
\]
Given the seminal work of Sangiorgi
\cite{sangiorgipi}, 
such a representation naturally begs the
question of whether or not we can develop a \emph{typed} encoding of
higher-order processes into the first-order setting.  Indeed, we can
achieve such an encoding with a fairly simple extension of the
encoding of \S~\ref{sec:hovals} to \hopi\ by observing that monadic values are
processes that need to be potentially provided with extra sessions in
order to be executed correctly. For instance, a term of type
$\{x{:}A\vdash y{:}B\}$ denotes a process that given a session $x$ of
type $A$ will then offer $y{:}B$. Exploiting this observation we
encode this type as the session $A\lolli B$, ensuring
subsequent usages of such a term are consistent with this interpretation.
\[
 \small
  \begin{array}{lcl}
  \lb \{ \ov{x_j{:}A_j}
    \vdash z{:}A\} \rb & \triangleq & \ov{\lb A_j\rb} \lolli
                                      \lb A\rb\\[0.5em]
    \lb \{ x \leftarrow P \rightarrow \ov{y_i}\}\rb_z & \triangleq &
                 z(y_0).\dots.z(y_n).\lb P\{z/x\}\rb
    \quad (z\not\in \fn{P})\\[0.5em]

    \lb x\leftarrow M \leftarrow \ov{y_i};Q \rb & \triangleq &
      (\nub x)(\lb M \rb_x \mid \ov{x}\langle a_0\rangle.([a_0\leftrightarrow
                                                             y_0] \mid
                                                             \dots
\mid x\langle a_n\rangle.([a_n \leftrightarrow y_n] \mid \lb Q \rb)
                                                               \dots ))                 \end{array}
\]

To encode the monadic type $\{ \ov{x_j{:}A_j} \vdash z{:}A\}$,
denoting the type of process $P$ that is typed by
$\ov{x_j{:}A_j} \vdash P :: z{:}A$, 
we require that the session in the image of the translation specifies
a sequence of channel inputs with behaviours $\ov{A_j}$ that make up
the linear context. After the contextual aspects of the type are
encoded, the session will then offer the (encoded) behaviour of
$A$. Thus, the encoding of the monadic type is
$\lb A_0\rb \lolli \dots \lolli \lb A_n\rb \lolli \lb A\rb$, which we
write as $\ov{\lb A_j\rb} \lolli \lb A\rb$.  The encoding of monadic
expressions adheres to this behaviour, first performing the necessary
sequence of inputs and then proceeding inductively.  Finally, the
encoding of the elimination form for monadic expressions behaves
dually, composing the encoding of the monadic expression with a
sequence of outputs that instantiate the consumed names accordingly
(via forwarding).  The encoding of process $P$ from Equation~\ref{eq:ho}
is thus:
\[\small
\begin{array}{l}
  \lb P \rb  =  (\nub c)(\lb c\langle \{d \leftarrow
  Q\}\rangle.c(x).\zero\rb \mid \lb c(y).  d \leftarrow y;
  d(n).c\langle n \rangle.\zero\rb)\\
 =  (\nub c)(\ov{c}\langle w\rangle.(\bang w(d).\lb Q\rb \mid
c(x).\zero)
 c(y).(\nub d)(\ov{y}\langle b\rangle.[b\leftrightarrow d] \mid 
d(n).\ov{c}\langle m \rangle.(\ov{n}\langle e\rangle.[e\leftrightarrow m] \mid \zero)))
\end{array}
\]

\paragraph{{\bf Properties of the Encoding.}}
As in our previous development, we can show that
our encoding for \hopi\
\B{is type sound and satisfies operational correspondence. 
  The full development is omitted but can be found in \cite{longversion}.}



We encode \hopi\ into $\lambda$-terms, extending \S~\ref{sec:hovals} with:
\[\small
\begin{array}{l}
\llp \{\ov{x_i{:}A_i} \vdash z{:} A\} \rrp  \triangleq  
\ov{\llp A_i \rrp} \lolli \llp A\rrp
\\[0.5mm]
\llp x\leftarrow M \leftarrow \ov{y_i};Q\rrp  \triangleq  \llp Q \rrp\{(\llp
                                                     M \rrp\,\ov{y_i})
                                                    /x\}\
\quad 
\llp \{x\leftarrow P \leftarrow \ov{w_i}\}\rrp \triangleq 
  \lambda w_0 . \dots . \lambda w_n.\llp P \rrp                                      \end{array}
\]
The encoding translates the monadic type $\{\ov{x_i{:}A_i} \vdash z{:}
A\}$ as a linear function $\ov{\llp A_i \rrp} \lolli \llp A\rrp$, which
captures the fact that the underlying value must be provided with
terms satisfying the requirements of the linear context. At the level
of terms, the encoding for the monadic term constructor follows its
type specification, generating a nesting of $\lambda$-abstractions that
closes the term and proceeding inductively. For the process encoding,
we translate the monadic application construct analogously to the
translation of a linear $\cut$, but applying the appropriate variables
to the translated monadic term (which is of function type). We remark
the similarity between our encoding and that of the previous section,
where monadic terms are translated to a sequence of inputs (here a
nesting of $\lambda$-abstractions).
%
\B{Our encoding satisfies type soundness and operational correspondence,
as usual.}
Further showcasing the applications of our development, we obtain a
novel strong normalisation result for this higher-order
session-calculus ``for free'', through encoding to the
$\lambda$-calculus.




\begin{restatable}[Strong Normalisation]{theorem}{thmsn}
  Let $\Psi ; \Ga ; \D \vdash P :: z{:}A$.
There is no infinite reduction sequence starting from $P$.
\end{restatable}
\vspace{-3mm}


\begin{restatable}[Inverse Encodings]{theorem}{thminvencsho}
\label{thm:inv3}\label{thm:inv4}~
If $\Psi ; \Ga ; \D \vdash P :: z{:}A$
then $\lb\llp P \rrp\rb_z \logsim \lb P\rb$. Also,
if $\Psi \vdash M : \tau$ then $\llp\lb M \rb_z\rrp =_\beta \llp M \rrp$.

\end{restatable}
\vspace{-3mm}



\begin{theorem}
  Let $\vdash M : \tau$,  $\vdash N : \tau$, $\vdash P :: z{:}A$
  and $ \vdash Q :: z{:}A$. $ \llp M\rrp =_{\beta\eta} \llp N\rrp$ iff
  $\lb M \rb_z \logsim \lb N \rb_z$ and  $\lb P\rb \logsim \lb Q\rb$ iff $\llp P \rrp =_{\beta\eta} \llp Q \rrp$.
\end{theorem}



\section{Related Work and Concluding Remarks}\label{sec:related}
%
\mypara{Process Encodings of Functions} Toninho et
al.~\cite{DBLP:conf/fossacs/ToninhoCP12} study encodings of the
simply-typed $\lambda$-calculus in a logically motivated session
$\pi$-calculus, via encodings to the linear
$\lambda$-calculus.  Our work differs since they do not study
polymorphism nor reverse encodings; and we provide deeper insights
through applications of the encodings. Full abstraction or inverse
properties are not studied.

\B{Sangiorgi \cite{DBLP:conf/mfps/Sangiorgi93} uses a fully abstract
  compilation from the higher-order $\pi$-calculus (HO$\pi$) to the
  $\pi$-calculus to study full abstraction for Milner's encodings of
  the $\lambda$-calculus. The work shows that Milner's encoding of the
  lazy $\lambda$-calculus can be recovered by restricting the semantic
  domain of processes (the so-called \emph{restrictive} approach) or
  by enriching the $\lambda$-calculus with suitable constants. }
\B{This work was later refined in
  \cite{DBLP:conf/birthday/Sangiorgi00}, which does not use HO$\pi$
  and considers an operational equivalence on $\lambda$-terms called
  \emph{open applicative bisimulation} which coincides with 
  L\'{e}vy-Longo tree equality.  The work
  \cite{DBLP:conf/concur/SangiorgiX14} studies general conditions
  under which encodings of the $\lambda$-calculus in the
  $\pi$-calculus are fully abstract wrt L\'{e}vy-Longo and B\"{o}hm
  Trees, which are then applied to several encodings of (call-by-name)
  $\lambda$-calculus.  The works above deal with {\em untyped
    calculi}, and so reverse encodings are unfeasible. In a broader
  sense, our approach takes the restrictive approach using linear
  logic-based session typing and the induced observational
  equivalence. We use a $\lambda$-calculus with booleans as
  observables and reason with a Morris-style equivalence instead of
  tree equalities.
It would be an interesting future work to apply the conditions 
in \cite{DBLP:conf/concur/SangiorgiX14} in our typed 
setting. }

Wadler \cite{DBLP:journals/jfp/Wadler14} shows
a correspondence between a linear functional language with session
types GV and a session-typed process calculus with polymorphism based
on classical linear logic CP. Along the lines of this work, Lindley
and Morris~\cite{DBLP:conf/icfp/LindleyM16}, in an exploration of
inductive and coinductive session types through the addition of least
and greatest fixed points to CP and GV, develop an encoding from a
linear $\lambda$-calculus with session primitives (Concurrent $\mu$GV)
to a pure linear $\lambda$-calculus (Functional $\mu$GV) via a CPS
transformation. They also develop translations between $\mu$CP and
Concurrent $\mu$GV, extending
\cite{DBLP:conf/esop/LindleyM15}. Mapping to the terminology used in
our work 
\cite{DBLP:journals/iandc/Gorla10}, 
their encodings
are shown to be operationally complete, but no results are shown for
the operational soundness directions \B{and neither full abstraction nor
inverse properties are studied}.
In addition, their operational
characterisations do not compose across encodings. For instance, while
strong normalisation of Functional $\mu$GV implies the same property
for Concurrent $\mu$GV through their operationally complete encoding, the encoding
from $\mu$CP to $\mu$GV does not necessarily preserve this property.

Types for $\pi$-calculi delineate sequential behaviours by
restricting composition and name usages, limiting the
contexts in which processes can interact.  Therefore typed
equivalences offer a {\em coarser} semantics than untyped semantics.
Berger et al.~\cite{DBLP:journals/acta/BergerHY05} study an encoding
of System F in a polymorphic linear $\pi$-calculus, showing it to be
fully abstract based on game semantics techniques.  Their typing
system and proofs are more complex due to the fine-grained constraints
from game semantics.  Moreover, they do not study a reverse
encoding.

Orchard and Yoshida \cite{DBLP:conf/popl/OrchardY16} develop
embeddings to-and-from PCF with parallel effects and a session-typed
$\pi$-calculus, but only develop operational correspondence and
semantic soundness results, leaving the full abstraction problem open.


\noindent\mypara{Polymorphism and Typed Behavioural Semantics} The work of
\cite{DBLP:conf/esop/CairesPPT13} studies parametric session
polymorphism for the intuitionistic setting, developing a behavioural
equivalence that captures parametricity, which is used (denoted as
$\logsim$) in our paper.  The work \cite{DBLP:journals/jacm/PierceS00}
introduces a typed bisimilarity for polymorphism in the
$\pi$-calculus.  Their bisimilarity is of an intensional flavour,
whereas the one used in our work follows the extensional style of
Reynolds \cite{DBLP:conf/ifip/Reynolds83}.  Their typing discipline
(originally from~\cite{turnerpoly96}, which also develops
type-preserving encodings of polymorphic $\lambda$-calculus into
polymorphic $\pi$-calculus) differs significantly from the linear
logic-based session typing of our work (e.g. theirs does not ensure
deadlock-freedom). A key observation in their work is the coarser
nature of typed equivalences with polymorphism (in analogy to those
for IO-subtyping \cite{DBLP:journals/mscs/PierceS96}) and their
interaction with channel aliasing, suggesting a use of typed
semantics and encodings of the $\pi$-calculus for fine-grained
analyses of program behaviour.
%
%

\noindent\mypara{F-Algebras and Linear-F}
The use of initial and final (co)algebras to give a semantics to
inductive and coinductive types dates back to Mendler
\cite{DBLP:conf/lics/Mendler87}, with their strong definability in System F
appearing in \cite{DBLP:journals/tcs/BainbridgeFSS90} and
\cite{DBLP:journals/mscs/Hasegawa94}. 
The definability of inductive and coinductive types using
parametricity also appears in \cite{DBLP:conf/tlca/PlotkinA93} in the
context of a logic for parametric polymorphism and later in
\cite{DBLP:journals/lmcs/BirkedalMP06} in a linear variant of such a
logic. The work of \cite{DBLP:conf/aplas/ZhaoZZ10} studies
parametricity for the polymorphic linear $\lambda$-calculus of this
work, developing encodings of a few inductive types but not the
initial (or final) algebraic encodings in their full generality.
  Inductive and coinductive session types in a logical process setting
  appear in \cite{DBLP:conf/tgc/ToninhoCP14} and
  \cite{DBLP:conf/icfp/LindleyM16}. Both works consider a calculus
  with built-in recursion -- the former in an intuitionistic setting
  {where a process that offers a (co)inductive protocol is composed
    with another that consumes the (co)inductive protocol} and the
  latter in a classical framework {where composed recursive
    session types are dual each other}.
\label{sec:conc}

\noindent\mypara{Conclusion and Future Work}
This work answers the question of what kind of
  type discipline of the $\pi$-calculus can exactly capture and is
  captured by $\lambda$-calculus behaviours.
Our answer
is given by showing the first mutually inverse and fully abstract
encodings between two calculi with polymorphism, one being the
Poly$\pi$ session
calculus based on intuitionistic linear logic, and the other (a
linear) System F.  This further demonstrates that the linear
logic-based articulation of name-passing interactions originally
proposed by \cite{DBLP:conf/concur/CairesP10} (and studied extensively
thereafter
e.g.~\cite{DBLP:conf/esop/ToninhoCP13,DBLP:conf/tgc/ToninhoCP14,DBLP:conf/esop/PerezCPT12,DBLP:journals/mscs/CairesPT16,DBLP:journals/jfp/Wadler14,DBLP:conf/esop/CairesPPT13,DBLP:conf/esop/LindleyM15})
provides a clear and applicable tool for message-passing concurrency.
By exploiting the proof theoretic equivalences between natural
deduction and sequent calculus we develop mutually inverse
and fully abstract encodings, which naturally extend to more intricate
settings such as process passing (in the sense of HO$\pi$).
Our encodings also enable us to derive properties of the $\pi$-calculi
``for free''. Specifically, we show how to obtain adequate
representations of least and greatest fixed points in Poly$\pi$
through the encoding of initial and final (co)algebras in the
$\lambda$-calculus.  We also straightforwardly derive a strong
normalisation result for the higher-order session calculus, which
otherwise involves non-trivial proof techniques
\cite{DBLP:journals/jlp/DemangeonHS10,DBLP:conf/birthday/DemangeonHS09,DBLP:conf/esop/PerezCPT12,DBLP:conf/esop/CairesPPT13,DBLP:journals/acta/BergerHY05}.
Future work includes extensions to the classical linear logic-based 
framework, including multiparty session types \cite{CLMSW16,CMSY2015}. 
Encodings of session $\pi$-calculi to the $\lambda$-calculus
have been used to implement session primitives in functional
languages such as Haskell (see a recent survey \cite{OY17}), OCaml
\cite{padovani16simple,padovani17context,IYY2017} and Scala
\cite{SDHY2017}.
Following this line of work, we plan to develop 
encoding-based implementations of this work as
embedded DSLs. 
This would potentially enable an exploration of algebraic
constructs beyond initial and final co-algebras in a session
programming setting. In particular, we wish to further study the
meaning of functors, natural transformations and related constructions
in a session-typed setting, both from a more fundamental
viewpoint but also in terms of programming patterns.

\noindent{\small {\bf Acknowledgements.} 
The authors would like to thank the reviewers for their
comments, suggestions and pointers to related works. 
This work is partially supported by EPSRC EP/K034413/1,
EP/K011715/1, EP/L00058X/1, EP/N027833/1 and EP/N028201/1.}

\bibliographystyle{splncs03}

 \newpage
 \appendix

\hspace{0pt}
\vfill
\begin{center}
{\Huge \bf Appendix}\\[1em]

On Polymorphic Sessions and Functions\\
A Tale of Two (Fully Abstract) Encodings\\[1em]

Additional definitions and proofs of the main materials.

\end{center}
\vfill
\hspace{0pt}

\newpage

\section{Appendix}

\subsection{Additional Definitions for \S~\ref{sec:sessionpi} --
  Structural Congruence}\label{app:procstuff}

\begin{definition}[Structural congruence]
  \label{def:struct-cong}   
  is the least congruence relation generated by the following laws: 
  $P \para \zero  \equiv P$; ~~ 
  $P \equiv_{\alpha} Q \Rightarrow P \equiv Q$; ~~
  $P \para Q \equiv Q \para P$; ~~
  $P \para (Q \para R) \equiv (P \para Q) \para R$; ~~
   $(\nub x)(\nub y)P \equiv (\nub y)(\nub x)P$; ~~
  $x\not\in\fn{P} \Rightarrow P \para (\nub x)Q \equiv (\nub x)(P \para Q)$; 
   $(\nub x)\zero \equiv \zero$; 
  and
    $[x\leftrightarrow y] \equiv [y\leftrightarrow x]$. 
\end{definition}

\begin{definition}[Extended Structural Congruence]
We write $\equiv_\bang$ 
for the least congruence relation on processes which results from extending 
structural congruence $\equiv$ (Def. \ref{def:struct-cong})
with the following  axioms, dubbed the Sharpened Replication Axioms \cite{sangiorgipi}: 
{\small 
\begin{enumerate}
\item $(\nub u)(!u(z).P \para (\nu y)(Q \para R)) \equiv_\bang(\nu y)((\nu u)(!u(z).P \para Q) \para (\nu u)(!u(z).P \para R))$
\item $(\nu u)(!u(y).P \para (\nu v)(!v(z).Q \para R)) 
\equiv_\bang (\nu v)((!v(z).(\nu u)(!u(y).P \para Q)) \para (\nu u)(!u(y).P \para R))$
\item $(\nu u)(!u(y).Q \para P) \equiv_\bang P$ ~~~if $u \not\in \fn{P}$
\end{enumerate}}
\end{definition}
Axioms (1) and (2) represent principles for the distribution of shared
servers among processes, while (3) formalises the garbage collection
of shared servers which cannot be invoked by any process. The axioms
embody distributivity, contraction and weakening of shared resources
and are sound wrt (typed) observational equivalence \cite{DBLP:conf/esop/PerezCPT12}.

\subsection{Additional Definitions for \S~\ref{sec:sessionpi} --
  Typing Rules}
\label{app:pityprules}

Below we list the typing rules for the calculus of section
\S~\ref{sec:sessionpi}. We note that the judgment
$\Omega \vdash B\,\m{type}$ simply requires that free variables in $B$
be in $\Omega$. Moreover, typing treats processes quotiented by
structural congruence -- given a well-typed process
$\Omega ; \Ga ; \D \vdash P :: T$, subject reduction ensures that for
all possible reductions $P \tra{\tau} P'$, there exists a process $Q$
where $P' \equiv Q$ such that $\Omega ; \Ga ; \D \vdash Q :: T$.
Related properties hold wrt general transitions $P \tra{\alpha} P'$.
We refer the reader to
\cite{DBLP:journals/mscs/CairesPT16,DBLP:conf/concur/CairesP10} for
additional details on this matter.

\[
  \begin{array}{c}
      \inferrule*[left=$(\m{id})$]
{\, }
  { \Omega ; \Ga ; x{:}A \vdash [x\leftrightarrow z] :: z{:}A }
  \quad
\inferrule*[left=$(\rgt\one)$]
{\, }
{\Omega ; \Ga ; \cdot \vdash \zero :: z{:}\one}
\quad
\inferrule*[left=$(\lft\one)$]
{\Omega ; \Ga ; \D \vdash P :: z{:}C}
{\Omega ; \Ga ; \D , x{:}\one \vdash P :: z{:}C}\\[1em]
\inferrule*[left=$(\rgt\lolli)$]
{\Omega ; \Ga ; \D , x{:}A \vdash P :: z{:}B}
{\Omega ; \Ga ; \D \vdash z(x).P :: z{:}A\lolli B}
\quad
\inferrule*[left=$(\lft\lolli)$]
{\Omega ; \Ga ; \D_1 \vdash P :: y{:}A \quad 
 \Omega ; \Ga ; \D_2 , x{:} B \vdash Q :: z{:}C}
{\Omega ; \Ga ; \D_1 , \D_2 , x{:}A \lolli B \vdash (\nub
  y)x\langle y \rangle.(P \mid Q) :: z{:}C}\\[1em]
\inferrule*[left=$(\rgt\tensor)$]
{\Omega ; \Ga ; \D_1 \vdash P :: y{:}A \quad 
 \Omega ; \Ga ; \D_2 \vdash Q :: z{:}B}
{\Omega ; \Ga ; \D_1 , \D_2 \vdash (\nub x)z\langle y \rangle.(P \mid
  Q) :: z{:}A\tensor B}
\quad
\inferrule*[left=$(\lft\tensor)$]
{\Omega ; \Ga ; \D ,y{:}A , x{:}B \vdash P :: z{:}C}
{\Omega ; \Ga ; \D , x{:}A\tensor B \vdash x(y).P :: z{:}C}\\[1em]
\inferrule*[left=$(\rgt\with)$]
{\Omega ; \Ga ; \D \vdash P :: z{:}A \quad 
 \Omega ; \Ga ; \D \vdash Q :: z{:}B}
{\Omega ; \Ga ; \D \vdash z.\m{case}(P,Q) :: z{:}A\with B}
\quad
\inferrule*[left=$(\lft\with_1)$]
{\Omega ; \Ga ; \D , x{:}A \vdash P :: z{:}C}
{\Omega ; \Ga ; \D , x{:}A \with B \vdash x.\m{inl};P :: z{:}C}\\[1em]
\inferrule*[left=$(\lft\with_2)$]
{\Omega ; \Ga ; \D , x{:}A \vdash P :: z{:}C}
{\Omega ; \Ga ; \D , x{:}A \with B \vdash x.\m{inr};P :: z{:}C}
\quad
\inferrule*[left=$(\rgt\oplus_1)$]
{\Omega ; \Ga ; \D\vdash P :: z{:}A}
{\Omega ; \Ga ; \D  \vdash z.\m{inl};P :: z{:}A\oplus B}\\[1em]
\inferrule*[left=$(\rgt\oplus_2)$]
{\Omega ; \Ga ; \D\vdash P :: z{:}B}
{\Omega ; \Ga ; \D  \vdash z.\m{inr};P :: z{:}A\oplus B}
\quad
\inferrule*[left=$(\lft\oplus)$]
{\Omega ; \Ga ; \D , x{:}A \vdash P :: z{:}C 
\quad 
\Omega ; \Ga ; \D , x{:}B \vdash Q :: z{:}C }
{\Omega ; \Ga ; \D , x{:}A\oplus B \vdash x.\m{case}(P,Q)  ::
  z{:}C}\\[1em]
\inferrule*[left=$(\rgt\bang)$]
{\Omega ; \Ga ; \cdot \vdash P :: x{:}A}
{\Omega ; \Ga ; \cdot \vdash \bang z(x).P :: z{:}\bang A}
\quad
\inferrule*[left=$(\lft\bang)$]
{\Omega ; \Ga , u{:}A ; \D \vdash P :: z{:}C}
{\Omega ; \Ga ; \D , x{:}\bang A \vdash P\{x/u\} :: z{:}C}
\\[1em]
\inferrule*[left=$(\m{copy})$]
{\Omega ; \Ga , u{:}A ; \D , y{:}A \vdash P :: z{:}C}
{\Omega ; \Ga , u{:}A ; \D \vdash (\nub y)u\langle y \rangle.P ::
  z{:}C}
\\[1em]
\inferrule*[left=$(\rgt\forall)$]
{\Omega , X ; \Ga ; \D \vdash P :: z{:}A}
{\Omega ; \Ga ; \D \vdash z(X).P :: z{:}\forall X.A}
\quad
\inferrule*[left=$(\lft\forall)$]
{\Omega \vdash B\,\m{type}\quad \Omega ; \Ga ; \D , x{:}A\{B/X\} \vdash P :: z{:}C}
{\Omega ; \Ga ; \D , x{:}\forall X . A \vdash x\langle B\rangle.P ::
  z{:}C}
\\[1em]
\inferrule*[left=$(\rgt\exists)$]
{\Omega \vdash B\,\m{type}\quad \Omega ; \Ga ; \D \vdash P :: z{:}A\{B/X\}}
{\Omega ; \Ga ; \D \vdash z\langle B\rangle.P ::
  z{:}\exists X.A}
\quad
\inferrule*[left=$(\lft\exists)$]
{\Omega , X ; \Ga ; \D , x{:}A \vdash P :: z{:}C}
{\Omega ; \Ga ; \D , x{:}\exists X.A \vdash x(X).P :: z{:}C}
  \\[1em]
\inferrule*[left=$(\cut)$]
{\Omega ; \Ga ; \D_1 \vdash P :: x{:}A \quad 
 \Omega ; \Ga ; \D_2 , x{:}A \vdash Q :: z{:}C}
{\Omega ; \Ga ; \D_1 , \D_2 \vdash (\nub x)(P \mid Q) :: z{:}C}
\\[1em] \inferrule*[left=$(\cut^\bang)$]
 {\Omega ; \Ga ; \cdot \vdash P :: x{:}A \quad 
  \Omega ; \Ga , u{:} A ; \D \vdash Q :: z{:}C }
 {\Omega ; \Ga; \D \vdash (\nub u)(\bang u(x).P \mid Q) :: z{:}C}
\end{array}
\]

\subsection{Additional Definitions for \S~\ref{sec:sessionpi} --
  Typed Barbed Congruence}
\label{subsec:barb}
\begin{definition}[Type-respecting Relations \cite{DBLP:conf/esop/CairesPPT13}]
A type-respecting relation over processes, written
$\{\mathcal{R}_S\}_S$ is defined as a family of relations over
processes indexed by $S$. We often write $\mathcal{R}$ to refer to the
whole family, and write $\Omega ; \Ga ; \D \vdash P \mathcal{R} Q :: T$
to denote $\Omega ; \Ga ; \D \vdash P , Q :: T$ and $(P,Q) \in
\mathcal{R}_{\Omega ; \Ga ; \D \vdash T}$.
\end{definition}

We say that a type-respecting relation is an equivalence if it
satisfies the usual properties of reflexivity, transitivity and
symmetry. In the remainder of this development we often omit
``type-respecting''.

\begin{definition}[$\tau$-closed \cite{DBLP:conf/esop/CairesPPT13}]\label{def:tauclosed}
Relation $\mathcal{R}$
is  \emph{$\tau$-closed} if
$\Omega ; \Ga ; \D \vdash P \mathcal{R} Q :: T$ and $P \tra{} P'$ imply
there exists a $Q'$ such that $Q \wtra{} Q'$ and 
$\Omega ; \Ga ; \D \vdash P' \mathcal{R} Q' :: T$. 
\end{definition}

Our definition of basic observable on processes, or \emph{barb}, is
given below.

\begin{definition}[Barbs \cite{DBLP:conf/esop/CairesPPT13}]\label{def:barbs}
Let 
$O_x = \{\ov{x}, x, \ov{x.\m{inl}}, \ov{x.\m{inr}}, x.\m{inl}, x.\m{inr}\}$
be the set of  \emph{basic observables} under name $x$.
Given a well-typed process $P$, we write:
(i)~$\barb{P}{\ov{x}}$, if $P \tra{\ov{(\nub y)x \out y}}
  P'$; 
(ii)~
$\barb{P}{\ov{x}}$, if $P \tra{\ov{x\langle A
      \rangle}} P'$, for some $A,P'$;
(iii)~
$\barb{P}{x}$, if $P \tra{x(A)} P'$, for some $A , P'$;
(iv)~
$\barb{P}{x}$, if $P \tra{x(y)} P'$, for some $y, P'$; 
(v)~
$\barb{P}{\alpha}$, if $P \tra{\, \alpha \, } P'$, for some $P'$ and $\alpha \in O_x \setminus \{x, \ov{x}\}$.
Given some $o \in O_x$, 
we write $\wbarb{P}{o}$ if there exists a $P'$ such that 
$P \wtra{} P'$ and $\barb{P'}{o}$ holds.
\end{definition}

\begin{definition}[Barb preserving relation]
Relation $\mathcal{R}$
is a \emph{barb preserving} if, for every name $x$, 
$\Omega ; \Ga ; \D \vdash P \mathcal{R} Q :: T$ and $ \barb{P}{o}$ imply
$\wbarb{Q}{o}$, for any $o \in O_x$. 
\end{definition}

\begin{definition}[Contextuality]
A relation $\mathcal{R}$ is contextual if 
$\Omega ; \Ga ; \D \vdash P \mathcal{R} Q :: T$ implies $\Omega ; \Ga ;
\D' \vdash \mathcal{C}[P] \mathcal{R} \mathcal{C}[Q] :: T'$, for every $\D'$ $T'$ and
typed context $\mathcal{C}$.
\end{definition}

\begin{definition}[Barbed Congruence]\label{d:bcong}
\emph{Barbed congruence}, noted $\tbcong$, is 
the largest 
equivalence on well-typed processes
symmetric 
type-respecting relation 
that is $\tau$-closed, barb preserving, and contextual.
\end{definition}

\subsection{Additional Definitions for \S~\ref{sec:sessionpi} --
  Logical Equivalence}
\label{app:logeq}

The full definition for the base case of logical equivalence is given below:
\[
\begin{array}{rcl}
\logeq{P}{Q}{z{:}X}{\eta : \omega\candit\omega'} & \text{iff} & (P,Q) \in \eta(X)(z)\\
 \logeq{P}{Q}{z{:}\one}{\eta : \omega\candit\omega'}
 &\text{iff} &   \forall P',Q'.~(P \wtra{} P' \land P' \not\redd \land
 \;Q \wtra{} Q' \land Q'\not\redd) \Rightarrow\\
 && (P' \equiv_\bang \zero \land Q' \equiv_\bang \zero)\\
 \logeq{P}{Q}{z{:}A\lolli B}{\eta : \omega\candit\omega'}&\text{iff} & \forall P', y.~(P \tra{z(y)} P') \Rightarrow
 \exists Q'. Q \wtra{z(y)} Q' \, s.t.\,\\
 && \forall R_1,R_2 .~~ \logeq{R_1}{R_2}{y{:}A}{\eta : \omega
   \candit \omega'} \\ && \logeq{(\nu y)(P' \mid R_1)}{(\nu
   y)(Q' \mid R_2)}{z{:}B}{\eta : \omega\candit\omega'}\\ 
 \logeq{P}{Q}{z{:}A\tensor B}{\eta : \omega\candit\omega'}&
 \textrm{iff} & \forall P', y.~~(P \tra{\ov{(\nu y)z\out y}} P') \Rightarrow
 \exists Q'. Q \wtra{\ov{(\nu y)z\out y}} Q' \, s.t.\,\\
 && \exists P_1,P_2,Q_1,Q_2. P' \equiv_\bang P_1 \mid P_2 \wedge Q' \equiv_\bang Q_1 \mid Q_2\\
 && \logeq{P_1}{Q_1}{y{:}A}{\eta : \omega\candit\omega'} \wedge 
 \logeq{P_2}{Q_2}{z{:}B}{\eta : \omega\candit\omega'} \\
 \logeq{P}{Q}{z{:}\bang A}{\eta : \omega\candit\omega'}&\text{iff}  &
  \forall P'. ~~(P \tra{z(y)} P') \Rightarrow  \exists Q'. Q \wtra{z(y)}
  Q' \land \logeq{P' }{Q'}{y{:}A}{\eta :
   \omega\candit\omega'} \\
\logeq{P}{Q}{z{:}A\with B}{\eta : \omega\candit\omega'}&\text{iff}  & \\
 (\forall P'. (P \tra{z.\inl} P')  & \Rightarrow & \exists
 Q'. (Q\wtra{z.\inl~}Q' \land \logeq{P'}{Q'}{z{:}A}{\eta : \omega\candit\omega'})) \land  \\
  (\forall P'. (P \tra{z.\inr} P') &  \Rightarrow &  \exists
 Q'. (Q\wtra{z.\inr~}Q' \land \logeq{P'}{Q'}{z{:}B}{\eta : \omega\candit\omega'}))\\
\logeq{P}{Q}{z{:}A \oplus B}{\eta : \omega \candit \omega'}&\text{iff} & \\
(\forall P'. (P \tra{\overline{z.\inl}~} P')  & \Rightarrow & \exists
 Q'. (Q\wtra{\overline{z.\inl}~}Q' \land \logeq{P'}{Q'}{z{:}A}{\eta : \omega\candit\omega'})) \land \\
 (\forall P'. (P \tra{\overline{z.\inr}~} P')  & \Rightarrow & \exists
 Q'. (Q\wtra{\overline{z.\inr}~~}Q' \land \logeq{P'}{Q'}{z{:}B}{\eta : \omega\candit\omega'}))\\
\logeq{P}{Q}{z{:}\forall X. A}{\eta : \omega\candit\omega'}&
\text{iff} &
\forall B_1 , B_2 ,P',\mathcal{R} ::{-}{:}B_1\candit B_2.~~( P
\tra{z(B_1)} P' ) \,\mathit{implies} \\ 
\exists Q' . Q \wtra{z(B_2)} Q', ~P' & \logsim & Q'::z{:}A[\eta[X\mapsto \mathcal{R}] : \omega[X \mapsto
B_1] \candit \omega'[X\mapsto B_2]]
\\
\logeq{P}{Q}{z{:}\exists X . A}{\eta : \omega\candit\omega'} & \text{iff} & \exists
B_1,B_2, \mathcal{R} :: {-}{:}B_1 \candit B_2.~~(P \tra{\overline{z\out{B_1}}} P')
\,\mathit{implies} \\
\exists Q'. Q \wtra{\overline{z\out{B_2}}} Q',
 ~P' & \logsim & Q'::z{:}A[\eta[X\mapsto \mathcal{R}] : \omega[X \mapsto
B_1] \candit \omega'[X\mapsto B_2]]

\end{array}
\]

\subsection{Typing Rules for Linear-F}\label{app:linftyping}

\[
\small
\begin{array}{c}
\inferrule[(var)]
{\,}
{\Omega ; \Ga ; x{:}A \vdash x {:}A}
\quad
\inferrule[($\lolli I$)]
{\Omega ; \Ga ; \D , x{:}A \vdash M : B}
{\Omega ; \Ga; \D \vdash \lambda x{:}A.M : A \lolli B}
\quad
\inferrule[($\lolli E$)]
{\Omega ; \Ga ; \D \vdash M : A \lolli B \quad \Omega ; \Ga ; \D' \vdash N : A}
{\Omega ; \Ga ; \D,\D' \vdash M\, N : B}\\[1em]
\inferrule[$(\tensor I)$]
{\Omega ; \Ga ; \D \vdash M : A \quad \Omega ; \Ga ; \D' \vdash N : B}
{\Omega ; \Ga ; \D , \D' \vdash \langle M \tensor N\rangle : A \tensor B
  }\quad
\inferrule[$(\tensor E)$]
{\Omega ; \Ga ; \D \vdash M : A \tensor B \quad 
 \Omega ; \Ga ; \D , x{:}A , y{:}B \vdash N : B'}
{\Omega ; \Ga ; \D , \D'\vdash \m{let}\,x\tensor y = M\,\m{in}\, N :
  B'}\\[1em]
\inferrule[$(\bang I)$]
{\Omega ; \Ga ; \cdot \vdash M : A}
{\Omega ; \Ga ; \cdot \vdash \bang M : \bang A}
\quad
\inferrule[$(\bang E)$]
{\Omega ; \Ga ; \D \vdash M : \bang A \quad \Omega ; \Ga , u {:} A ; \D' \vdash N : B}
{\Omega ; \Ga ; \D, \D' \vdash \m{let}\,\bang u = M\,\m{in}\, N :B }
\quad
\inferrule[(uvar)]
{\, }
{\Omega ; \Ga , u{:} A ; \cdot \vdash u{:}A}
\\[1em]
\inferrule[$(\forall I)$]
{\Omega , X ; \Ga ; \D \vdash M : A}
{\Omega ; \Ga ; \D \vdash \Lambda X.M : \forall X.A}
\quad
\inferrule[$(\forall E)$]
{\Omega \vdash A\,\m{type} \quad \Omega ; \Ga ; \D \vdash M :
  \forall X.B}
{\Omega ; \Ga ; \D \vdash M[A] : B\{A/X\} }\\[1em]
\inferrule[$(\exists I)$]
{\Omega \vdash A\,\m{type} 
\quad \Omega ; \Ga ; \D \vdash M : B\{A/X\}}
{\Omega ; \Ga ; \D \vdash  \m{pack}\,A\,\m{with}\,M : \exists
  X.B}
\quad
\inferrule[$(\exists E)$]
{\Omega ; \Ga ; \D \vdash M : \exists X.A
\quad
\Omega , X ; \Ga ; \D , y{:}A \vdash N : B
\quad
\Omega \vdash B\,\m{type}
}
{\Omega ; \Ga ; \D \vdash \m{let}\,(X,y) = M\,\m{in}\,N : B }\\[1em]
\inferrule[$(\one I)$]
{\, }
{\Omega ; \Ga ; \cdot \vdash \munit : \one}
  \quad
  \inferrule[$(\one E)$]
  {\Omega ; \Ga ; \D \vdash M : \one
  \quad \Omega ; \Ga ; \D' \vdash N : C}
  {\Omega ; \Ga ; \D,\D' \vdash \llet{\one = M}{N} : C}
 \quad
\inferrule[$(\mathbf{2}I_1)$]{\, }{\Omega ; \Ga ; \cdot \vdash \TT : \mathbf{2}}
\quad
\inferrule[$(\mathbf{2}I_2)$]{\, }{\Omega ; \Ga ; \cdot \vdash \FF : \mathbf{2}}
\end{array}
\]

\subsection{Operational Correspondence for $\lb{-}\rb_z$}

The results follow from a straightforward extension to the development
in \cite{DBLP:conf/fossacs/ToninhoCP12}.

\begin{lemma}
\label{lem:compos}
\begin{enumerate}
\item Let $\Omega ; \Ga ; \D_1 , x{:}A \vdash M : B$ and $\Omega ; \Ga ; \D_2 \vdash N : A$. We
have that $\Omega ; \Ga ; \D_1 , \D_2 \vdash \lb M\{N/x\}\rb_z \logsim (\nub
x)(\lb M\rb_z \mid \lb N\rb_x) :: z{:}B$.
\item Let $\Omega ; \Ga , u{:}A ; \D \vdash M : B$ and $\Omega ; \Ga ; \cdot
\vdash N : A$. we have that $\Omega ; \Ga ; \D \vdash \lb M \{N/u\}
\rb_z \logsim (\nub u)(\lb M \rb_z \mid \bang u(x).\lb N \rb_x) ::
z{:}B$.
\end{enumerate}
\end{lemma}

\begin{proof} 
By induction on the structure of $M$, exploiting the fact that
commuting conversions and $\equiv_\bang$ are sound $\logsim$
equivalences.  See Lemma~\ref{lem:comp} for further details.
\end{proof}

\subsection{Additional Definitions for \S~\ref{sec:pitof} -- Encoding
  on Typing Derivations}
\label{app:enctypd}

The encoding on typing derivations is given in
Figures~\ref{fig:pitofderiv1} and~\ref{fig:pitofderiv2} (for
readability purposes, the processes are highlighted in \blue{blue}). The encoding
makes use of the two admissible substitution principles denoted by the
following rules:
\[\small
  \inferrule[(subst)]
  {\Omega ; \Ga ; \D_1 , x{:}B \vdash M : A \quad
   \Omega ; \Ga ; \D_2 \vdash N : B }
  {\Omega ; \Ga ; \D_1 ,\D_2 \vdash M\{N/x\} : A }
\qquad
  \inferrule[(subst{$^\bang$})]
  {\Omega ; \Ga , u{:}B ; \D \vdash M : A \quad
   \Omega ; \Ga ; \cdot \vdash N : B }
  {\Omega ; \Ga ; \D \vdash M\{N/u\} : A }
\]

\begin{figure}[t]
\[
\footnotesize
\begin{array}{l}
\pif{\inferrule[$(\rgt\one)$]{ }{\Omega ; \Ga ; \cdot \vdash \blue{\mathbf{0}} ::
  z{:}\one} }
\triangleq
\inferrule[$(\one I)$]{ }{\Omega ; \Ga ; \cdot \vdash \blue{\langle\rangle} :
  \one}\\[2em]
\pif{
\inferrule[$(\lft\one)$]{\Omega ; \Ga ; \D \vdash \blue P :: z{:}C}
{\Omega ; \Ga ; \D , x{:}\one \vdash \blue P :: z{:}C}}
\triangleq
\inferrule[$(\one E)$]
{ \Omega ; \Ga ; x{:}\one \vdash \blue x : \one \quad 
  \Omega ; \Ga ; \D \vdash \blue{\llp P \rrp_{\Omega ; \Ga ; \D \vdash z{:}C}} : C }
{\Omega ; \Ga ; \D , x{:}\one \vdash \blue{\llet{\one = x }
{ \llp P \rrp_{\Omega ; \Ga ; \D \vdash z{:}C}}} : C}

\\[2em]

\pif{\inferrule[(id)]{ }{\Omega ; \Ga ; x{:}A \vdash  \blue{[x\leftrightarrow
  z]}  ::
  z{:}A}}  \triangleq  \inferrule[(var)]{ }{\Omega ; \Ga ; x{:}A
                         \vdash \blue x {:}A}
\\[2em]
\pif{ \inferrule[$(\rgt\bang)$]
{\Omega ; \Ga ; \cdot \vdash \blue P :: x{:}A}
{\Omega ; \Ga ; \cdot \vdash \blue{\bang z(x).P} :: z{:}\bang A}}   \triangleq
\inferrule[$(\bang I)$]
{\Omega ; \Ga ; \cdot \vdash \blue{\llp P \rrp}_{\Omega ; \Ga ; \cdot
  \vdash x{:} A} : A}
{\Omega ; \Ga ; \cdot \vdash \blue{\bang \llp P \rrp}_{\Omega ; \Ga ; \cdot
  \vdash z{:}\bang A} : \bang A }
\\[2em]
\pif{\inferrule*[left=$(\rgt\lolli)$]{\Omega ; \Ga ; \D , x{:}A\vdash  \blue
    P  ::z{:}B }{\Omega ; \Ga ; \D \vdash  \blue{z(x).P} ::
    z{:}A\lolli B}}
\triangleq
\inferrule*[left=$(\lolli I)$]{\Omega ; \Ga ; \D , x{:}A \vdash \blue{\llp P \rrp}_{\Omega , \Ga
                                                       ; \D,x{:}A\vdash z{:}B}
                                                           : B  }{\Omega
                                                          ; \Ga ; \D
                                                          \vdash
                                                          \blue{\lambda
                                                          x{:}A. \llp
                                                                    P \rrp}_{\Omega , \Ga
                                                           ; \D,x{:}A
                                                           \vdash z{:}B}
                                                           : A
                                                          \lolli B}
\\[2em] 
\pif{\inferrule[$(\lft\lolli)$]{\Omega ; \Ga ; \D_1 \vdash \blue P :: y{:}A \quad \Omega ; \Ga ;
                                                         \D_2 , x{:}B
                                                         \vdash \blue Q :: z{:}C}{\Omega ; \Ga ; \D_1,\D_2 , x{:}A\lolli B \vdash \blue{(\nub
  y)x\langle y \rangle.(P\mid Q)} :: z{:}C}}

 \triangleq 
\\[2em]
\hspace{1cm}                                                       
\inferrule[(subst)]
{\Omega ; \Ga ; \D_2 , x{:}B \vdash \blue{\llp Q \rrp}_{\Omega ; \Ga ; \D_2 , x{:}B \vdash z{:}C} : C \quad
\inferrule[$(\lolli E)$]
{\Omega ; \Ga ; x{:}A\lolli B \vdash \blue x {:} A \lolli B \quad
 \Omega ; \Ga ; \D_1 \vdash \blue{\llp P \rrp}_{\Omega ; \Ga ; \D_1 \vdash y{:}A} : B }
{\Omega ; \Ga ; \D_1 , x {:} A\lolli B \vdash \blue{x\,\llp P \rrp}_{\Omega ; \Ga ; \D_1 \vdash y{:}A} : B }}
{\Omega ; \Ga ; \D_1 , \D_2 , x{:}A\lolli B \vdash \blue{\llp Q \rrp}_{\Omega ; \Ga ; \D_2 , x{:}B \vdash z{:}C}
\blue{\{(x\,\llp P \rrp}_{\Omega ; \Ga ; \D_1 \vdash y{:}A} \blue{)/x\}} : C}\\[2em] 
\pif{\inferrule[$(\rgt\tensor)$]
{\Omega ; \Ga ; \D_1 \vdash \blue P :: x{:}A \quad \Omega ; \Ga ; \D_2
\vdash 
\blue Q :: z{:}B}
{\Omega ; \Ga ; \D_1 , \D_2 \vdash \blue{(\nub x)z\langle x \rangle.(P \mid Q)} ::
  z{:}A\tensor B}}  \triangleq
\inferrule[$(\tensor I)$]
{\Omega ; \Ga ; \D_1 \vdash \blue{\llp P\rrp}_{\Omega ; \Ga ; \D_1 \vdash
  x{:}A} : A \quad 
\Omega ; \Ga ; \D_1 \vdash \blue{\llp Q\rrp}_{\Omega ; \Ga ; \D_2 \vdash
  z{:}B} : B}
{\Omega ; \Ga ; \D_1 , \D_2 \vdash \blue{\langle \llp P\rrp}_{\Omega ; \Ga ; \D_1 \vdash
  x{:}A} \blue{\tensor \llp
  Q\rrp}_{\Omega ; \Ga ; \D_2 \vdash
  z{:}B} \blue{\rangle} : A\tensor B}
\\[2em] 

\pif{\inferrule[$(\lft\tensor)$]
{\Omega ; \Ga ; \D , y{:}A . x{:}B \vdash \blue P :: z{:}C}
{\Omega ; \Ga ; \D , x{:}A\tensor B \vdash \blue {x(y).P} :: z{:}C}} 

\triangleq  
\inferrule[$(\tensor E)$]
{\Omega ; \Ga ; x{:} A\tensor B \vdash \blue x : A\tensor B
 \quad \Omega ; \Ga ; \D , y{:}A , x{:}B \vdash \blue{\llp P \rrp}_{\Omega ; \Ga ; \D
  , y{:}A . x{:}B \vdash z{:}C} : C }
{\Omega ; \Ga ; \D_1 , x{:}A \tensor B 
\vdash \blue{\m{let}\,x\tensor y = x\,\m{in}\, \llp P \rrp}_{\Omega ; \Ga ; \D
  , y{:}A . x{:}B \vdash z{:}C} : C } \\[2em] 

\end{array}
\]
\caption{Translation on Typing Derivations from Poly$\pi$ to Linear-F 
(Part 1)\label{fig:pitofderiv1}}
\end{figure}
\begin{figure}
\[
\footnotesize
\begin{array}{l}
\pif{
\inferrule*[left=$(\lft\bang)$]
{\Omega ; \Ga , u{:}A ; \D \vdash \blue P :: z{:}C}
{\Omega ; \Ga ; \D , x{:}\bang A \vdash \blue{P\{u/x\}}:: z{:}C}}
  \triangleq   
\inferrule*[left=$(\bang E)$]
{\Omega ; \Ga ; x{:}\bang A \vdash \blue x : \bang A \quad
\Omega ; \Ga , u {:}A ; \D \vdash \blue{\llp P \rrp}_{\Omega ; \Ga , u{:}A ; \D
  \vdash z{:}C} : C }
{\Omega ; \Ga ; \D , x{:} \bang A \vdash 
\blue{\m{let}\,\bang u = x\,\m{in}\, \llp P \rrp}_{\Omega ; \Ga , u{:}A ; \D
  \vdash z{:}C} : C}\\[1.5em] 

\pif{
\inferrule*[left=$(\cpy)$]
{\Omega ; \Ga ,u{:}A ; \D , x{:} A \vdash \blue P :: z{:}C}
{\Omega ; \Ga , u{:}A ; \D \vdash \blue{(\nub x)u\langle x \rangle.P} :: z{:}C}}
  \triangleq  \\[1em]
\qquad\qquad\qquad  
\inferrule*[left=(subst)]
{ \Omega ; \Ga , u{:}A ; \D , x{:}A \vdash \blue{\llp P \rrp}_{\Omega ; \Ga ,u{:}A ; \D ,  x{:} A\vdash z{:}C} : C \quad 
\Omega ; \Ga , u{:}A ; \cdot \vdash \blue u{:}A}
{\Omega ; \Ga ,u{:}A ; \D \vdash \blue{\llp P \rrp}_{\Omega ; \Ga ,u{:}A ; \D ,
  x{:} A\vdash z{:}C}\blue{\{u/x\}} : C}
\\[1.5em]

\pif{
\inferrule*[left=$(\rgt\forall)$]
{\Omega , X ; \Ga ; \D \vdash \blue P :: z{:}A}
{\Omega ; \Ga ; \D \vdash \blue{z(X).P} :: z{:}\forall X.A} }
 \triangleq 
\inferrule*[left=$(\forall I)$]
{\Omega ,X ; \Ga ; \D \vdash \blue{\llp P\rrp}_{\Omega , X ; \Ga ; \D
  \vdash z{:}A } : A}
{\Omega ; \Ga ; \D \vdash \blue{\Lambda X.\llp P\rrp}_{\Omega , X ; \Ga ; \D
  \vdash z{:}A } : \forall X.A}
\\[1.5em] 

\pif{
\inferrule*[left=$(\lft\forall)$]
{\Omega \vdash B\,\m{type} \quad \Omega ; \Ga ; \D , x{:}A\{B/X\}
  \vdash \blue P :: z{:}C}
{\Omega ; \Ga ; \D , x{:} \forall X.A \vdash \blue {x\langle B \rangle.P} :: z{:}C}
}
\triangleq \\
\inferrule*[left=(subst)]
  {\Omega ; \Ga ; \D , x{:}A\{B/X\} \vdash \blue{\llp P \rrp}_{\Omega ; \Ga ;
  \D , x{:}A\{B/X\} \vdash z{:}C} : C \quad
\inferrule*[left=$(\forall E)$]
{\Omega ; \Ga , x{:}\forall X.A\vdash \blue x {:} \forall X.A \quad \Omega
  \vdash B\,\m{type}}
{\Omega ; \Ga ; x{:}\forall X.A \vdash \blue{x[B]} : A\{B/X\}} }
{\Omega ; \Ga ; \D , x{:}\forall X .A \vdash  \blue{\llp P \rrp}_{\Omega ; \Ga ; \D , x{:}A\{B/X\} \vdash z{:}C}\blue{\{(x[B]/x)\}} : C}
\\[1.5em] 
\pif{\inferrule*[left=$(\rgt\exists)$]
{\Omega \vdash B\,\m{type} \quad \Omega ; \Ga ; \D \vdash \blue P :: z{:}A\{B/X\}}
{\Omega ; \Ga ; \D \vdash \blue{z\langle B \rangle.P} :: z{:}\exists X.A }}
\triangleq
\inferrule*[left=$(\exists I)$]
{\Omega \vdash B\,\m{type} \quad \Omega ; \Ga ; \D \vdash \blue{\llp P \rrp}_{\Omega ;
  \Ga ; \D\vdash z{:}A\{B/X\}} : A\{B/X\}}
{\Omega ; \Ga ; \D \vdash \blue{\m{pack}\,B\,\m{with}\,\llp P \rrp}_{\Omega ;
  \Ga ; \D\vdash z{:}A\{B/X\}} : \exists X.A}
\\[1.5em] 

\pif{
\inferrule*[left=$(\lft\exists)$]
{\Omega , Y; \Ga ; \D , x{:}A \vdash \blue P :: z{:}C }
{\Omega ; \Ga ; \D , x{:}\exists X. A \vdash \blue{x(Y).P} :: z{:} C}
}
 \triangleq 
  \\[1em]
  \qquad\qquad\qquad
\inferrule*[left=$(\exists E)$]
{\Omega ; \Ga ; x{:}\exists Y.A \vdash \blue x {:}\exists Y.A \quad \Omega ,Y
  ; \Ga ; \D , x{:}A \vdash \blue{\llp P \rrp}_{\Omega , Y ; \Ga ; \D , x{:}A
  \vdash z{:}C} : C }
{\Omega  ; \Ga ; \D , x{:} \exists Y.A \vdash 
\blue{\m{let}\,(Y,x) = x\,\m{in}\, \llp P \rrp}_{\Omega , Y ; \Ga ; \D , x{:}A
  \vdash z{:}C} : C}

\\[1.5em] 

\pif{
\inferrule[$(\cut)$]
{\Omega ; \Ga ; \D_1 \vdash \blue P :: x{:}A \quad
 \Omega ; \Ga ; \D_2 , x{:}A \vdash \blue Q :: z{:}C}
{\Omega ; \Ga ; \D_1 , \D_2 \vdash \blue{(\nub x)(P\mid Q)} :: z{:}C} }
\triangleq
\\[1em]
\qquad\qquad\qquad\inferrule[(subst)]
{\Omega ; \Ga ; \D_2 , x{:}A \vdash \blue{\llp Q\rrp}_{\Omega ; \Ga ; \D_2 ,
  x{:}A \vdash z{:}C} : C
\quad
\Omega ; \Ga ; \D_1 \vdash \blue{\llp P \rrp}_{\Omega ; \Ga ; \D_1 \vdash
  x{:}A} : A}
{
\Omega ; \Ga ; \D_1 , \D_2 \vdash \blue{\llp Q\rrp}_{\Omega ; \Ga ; \D_2 , x{:}A \vdash z{:}C} 
\blue{\{\llp P \rrp}_{\Omega ; \Ga ; \D_1 \vdash x{:}A} \blue{/x\}} : C}
\\[1.5em] 

\pif{
\inferrule[$(\cut^\bang)$]
{\Omega ; \Ga ; \cdot \vdash \blue P :: x{:}A \quad 
 \Omega ; \Ga , u{:}A ; \D \vdash \blue Q :: z{:}C}
{\Omega ; \Ga ; \D \vdash \blue{(\nub u)(\bang u(x).P \mid Q)} :: z{:}C}
}  \triangleq 
 \inferrule[(subst$^\bang$)]
{\Omega ; \Ga , u{:}A  ; \D  \vdash \blue{\llp Q\rrp}_{\Omega ; \Ga , u{:}A ; \D
  \vdash z{:}C} : C
\quad
\Omega ; \Ga ; \cdot \vdash \blue{\llp P \rrp}_{\Omega ; \Ga ; \D_1 \vdash
  x{:}A} : A}
{
\Omega ; \Ga ; \D \vdash \blue{\llp Q\rrp}_{\Omega ; \Ga , u{:}A ; \D\vdash z{:}C} 
\blue{\{\llp P \rrp}_{\Omega ; \Ga ; \cdot \vdash x{:}A}\blue{ /u\}}}
\end{array}
\]
\caption{Translation on Typing Derivations from Poly$\pi$ 
to Linear-F (Part 2)\label{fig:pitofderiv2}}
\end{figure}

\subsection{Proofs for \S~\ref{sec:pitof} -- Encoding from Poly$\pi$
  to Linear-F}

\typsoundpitof*
\begin{proof}
Straightforward induction.
\end{proof}

\thmopcs*

\begin{proof}
Induction on typing and case analysis on the possibility of reduction.
\begin{description}

\item[Case:] 
\[
\inferrule*[left=$(\cut)$]
{\Omega ; \Ga ; \D_1 \vdash P_1 :: x{:}A \quad
 \Omega ; \Ga ; \D_2 , x{:}A\vdash P_2 :: z{:}C}
{\Omega ; \Ga ; \D_1 , \D_2 \vdash (\nub x)(P_1 \mid P_2) :: z{:}C }
\]
\noindent where $P_1 \tra{} P_1'$ or $P_2 \tra{} P_2'$.

\begin{tabbing}
$\llp (\nub x)(P_1 \mid P_2)\rrp = \llp P_2\rrp\{\llp P_1\rrp /x\}$ \` by definition\\
{\bf Subcase:} $P_1 \tra{} P_1'$\\
\quad $(\nub x)(P_1 \mid P_2) \tra{} (\nub x)(P_1' \mid P_2)$\\
\quad $\llp P_1 \rrp \rightarrow_\beta^* \llp P_1'\rrp $ \` by i.h.\\
\quad $\llp P_2\rrp \{\llp P_1 \rrp /x\} \rightarrow_\beta^* \llp P_2\rrp \{\llp
P_1'\rrp /x\}$ \` by definition\\
\quad $\llp (\nub x)(P_1' \mid P_2)\rrp = \llp P_2\rrp \{\llp P_1'\rrp
/x\}$ \` by definition\\
{\bf Subcase:} $P_2 \tra{} P_2'$\\
\quad $(\nub x)(P_1 \mid P_2) \tra{} (\nub x)(P_1 \mid P_2')$ \\
\quad $\llp P_2 \rrp \rightarrow_\beta^*  \llp P_2'\rrp$ \` by i.h.\\
\quad $\llp P_2 \rrp\{\llp P_1 \rrp /x\} \rightarrow_\beta^* \llp
P_2'\rrp\{\llp P_1\rrp/x\}$ \` by definition\\
\quad $\llp (\nub x)(P_1 \mid P_2')\rrp = \llp P_2'\rrp \{\llp P_1
\rrp /x\}$ \` by definition\\
\end{tabbing}

\item[Case:] 
\[
\inferrule*[left=$(\cut)$]
{\Omega ; \Ga ; \D_1 \vdash x(y).P_1 :: x{:}A\lolli B \quad
 \Omega ; \Ga ; \D_2 , x{:}A\lolli B \vdash (\nub y)x\langle y \rangle.(Q_1 \mid Q_2) :: z{:}C}
{\Omega ; \Ga ; \D_1 , \D_2 \vdash (\nub x)(x(y).P_1 \mid (\nub y)x\langle y \rangle.(Q_1 \mid Q_2)) :: z{:}C }
\]

\begin{tabbing}
  $(\nub x)(x(y).P_1 \mid (\nub y)x\langle y \rangle.(Q_1 \mid Q_2)) \tra{} 
  (\nub x)((\nub y)(Q_1 \mid P_1) \mid Q_2)$ \` by reduction\\
  $\llp (\nub x)(x(y).P_1 \mid (\nub y)x\langle y \rangle.(Q_1 \mid Q_2))\rrp = 
  (\llp Q_2\rrp \{(x\,\llp Q_1\rrp)/x\})\{(\lambda y.\llp P_1\rrp)/x\}$ \` by definition\\
  $(\llp Q_2\rrp \{(x\,\llp Q_1\rrp)/x\})\{(\lambda y.\llp P_1\rrp)/x\} = 
  \llp Q_2 \rrp\{((\lambda y.\llp P_1 \rrp)\,\llp Q_1\rrp)/x\}$ \\
  $\llp (\nub x)((\nub y)(Q_1 \mid P_1) \mid Q_2) \rrp = 
  \llp Q_2\rrp \{(\llp P_1\rrp \{\llp Q_1 \rrp /y\})/x\}$ \` by definition\\
  $\llp Q_2\rrp \{((\lambda y.\llp P_1\rrp)\,\llp Q_1\rrp)/x\} 
  \rightarrow_\beta \llp Q_2\rrp\{(\llp P_1\rrp \{\llp Q_1\rrp/y\})/x\}$ \` redex\\
  $\llp (\nub x)((\nub y)(Q_1 \mid P_1) \mid Q_2) \rightarrow_\beta^* \llp Q_2\rrp\{(\llp
  P_1\rrp \{\llp Q_1\rrp/y\})/x\}$ \` by definition
\end{tabbing}

\item[Case:] 
\[
\inferrule*[left=$(\cut)$]
{\Omega ; \Ga ; \D_1 \vdash (\nub y)x\langle y \rangle.(P_1 \mid P_2) :: x{:}A\tensor B \quad
 \Omega ; \Ga ; \D_2 , x{:}A\tensor B \vdash x(y).Q_1 :: z{:}C}
{\Omega ; \Ga ; \D_1 , \D_2 \vdash (\nub x)((\nub y)x\langle y \rangle.(P_1 \mid P_2) 
\mid x(y).Q_1) :: z{:}C }
\]
\begin{tabbing}
$(\nub x)((\nub y)x\langle y \rangle.(P_1 \mid P_2) 
\mid x(y).Q_1) \tra{} (\nub x)(P_2 \mid (\nub y)(P_1 \mid Q_1))$ \` by reduction\\
$\llp (\nub x)((\nub y)x\langle y \rangle.(P_1 \mid P_2) 
\mid x(y).Q_1)\rrp =  \llet{x\tensor y = \mpair{ \llp P_2\rrp }{\llp
    P_1 \rrp}}{\llp Q_1\rrp}$\\
$\llp (\nub x)(P_2 \mid (\nub y)(P_1 \mid Q_1))\rrp = 
\llp Q_1\rrp \{\llp P_2\rrp/x\}\{\llp P_1\rrp/y\} $ \` by def.\\
$\llet{x\tensor y = \mpair{ \llp P_2\rrp }{\llp
    P_1 \rrp}}{\llp Q_1\rrp} \tra{} 
\llp Q_1\rrp\{\llp P_2\rrp/x\}\{\llp P_1\rrp /y\}$\\
\end{tabbing}

\item[Case:]
\[
\inferrule*[left=$(\cut^\bang)$]
{\Omega ; \Ga ; \cdot \vdash  P_1 :: x{:} A \quad
 \Omega ; \Ga , u{:}A ; \D  \vdash (\nub x)u\langle x \rangle.Q_1 :: z{:}C}
{\Omega ; \Ga ; \D \vdash (\nub u)(\bang u(x).P_1 \mid
(\nub x)u\langle x \rangle.Q_1) :: z{:}C }
\]
\begin{tabbing}
$(\nub u)(\bang u(x).P_1 \mid (\nub x)u\langle x \rangle.Q_1) \tra{} 
(\nub u)(\bang u(x).P_1 \mid (\nub x)(P_1 \mid Q_1))$ \` by reduction\\
$\llp (\nub u)(\bang u(x).P_1 \mid (\nub x)u\langle x \rangle.Q_1) \rrp = 
\llp Q_1\rrp\{u/x\}\{\llp P_1\rrp /u\}$\\
$= \llp Q_1\rrp\{\llp P_1\rrp/x,\llp P_1\rrp/u\}$ \`  by def.\\
$\llp (\nub u)(\bang u(x).P_1 \mid (\nub x)(P_1 \mid Q_1))\rrp = 
(\llp Q_1\rrp\{\llp P_1\rrp /x\})\{\llp P_1\rrp/u\}$\\
\end{tabbing}
\item[Case:] 
\[
\inferrule*[left=$(\cut)$]
{\Omega ; \Ga ; \D_1 \vdash x(Y).P_1 :: x{:}\forall Y.A \quad
 \Omega ; \Ga ; \D_2 , x{:}\forall Y.A \vdash x\langle B\rangle.Q_1 :: z{:}C}
{\Omega ; \Ga ; \D_1 , \D_2 \vdash (\nub x)(x(Y).P_1
\mid  x\langle B\rangle.Q_1) :: z{:}C }
\]
\begin{tabbing}
$(\nub x)(x(Y).P_1
\mid  x\langle B\rangle.Q_1) \tra{} (\nub x)(P_1\{B_1/Y\} \mid Q_1)$ \` by reduction\\
$\llp (\nub x)(x(Y).P_1
\mid  x\langle B\rangle.Q_1)\rrp  = (\llp Q_1\rrp\{x[B]/x\})\{(\Lambda
Y.\llp P_1\rrp)/x\}$\\
$= \llp Q_1\rrp\{(\Lambda Y.\llp P_1\rrp [B])/x\} \rightarrow_\beta
\llp Q_1\rrp\{\llp P_1\rrp \{B_1/Y\}/x\}$ \` by definition\\
$\llp (\nub x)(P_1\{B_1/Y\} \mid Q_1)\rrp = \llp Q_1\rrp\{\llp P_1\rrp
\{B_1/Y\}/x\}$
\end{tabbing}
\item[Case:]
\[
\inferrule*[left=$(\cut)$]
{\Omega ; \Ga ; \D_1 \vdash x\langle B \rangle.P_1 :: x{:}\exists Y.A \quad
 \Omega ; \Ga ; \D_2 , x{:}\exists Y.A \vdash x(Y).Q_1 :: z{:}C}
{\Omega ; \Ga ; \D_1 , \D_2 \vdash (\nub x)( x\langle B \rangle.P_1
\mid  x(Y).Q_1) :: z{:}C }
\]
\begin{tabbing}
$(\nub x)( x\langle B \rangle.P_1
\mid  x(Y).Q_1) \tra{} (\nub x)(P_1 \mid Q_1\{B/Y\})$ \` by reduction\\
$\llp (\nub x)( x\langle B \rangle.P_1
\mid  x(Y).Q_1) \rrp =  
\llet{(Y,x) = \pack{B}{\llp P_1\rrp}}{\llp Q_1\rrp}$ \` by def.\\
$(\pack{B}{\llp P_1\rrp}{\llp Q_1\rrp} \rightarrow_\beta
\llp Q_1\rrp \{\llp P_1\rrp/x,B/Y\}$\\
$\llp(\nub x)(P_1 \mid Q_1\{B/Y\})\rrp = \llp Q_1\rrp\{B/Y\})\{\llp P_1\rrp/x\}$
\end{tabbing}
\end{description}
\end{proof}

\thmopcc*

\begin{proof}
By induction on typing.

\begin{description}
\item[Case:]
\[
\inferrule*[left=$(\lft\lolli)$]
{\Omega ; \Ga ; \D_1  \vdash  P_1 :: y{:}A \quad
 \Omega ; \Ga ; \D_2 , x{:}B \vdash P_2 :: z{:}C}
{ \Omega ; \Ga ; \D_1 , \D_2 , x{:}A\lolli B \vdash (\nub y)x\langle y
  \rangle.(P_1\mid P_2) :: z{:}C}
\]
\begin{tabbing}
$\llp (\nub y)x\langle y \rangle.(P_ 1\mid P_2)\rrp = \llp P_2\rrp
\{(x\,\llp P_1\rrp)
/ x\}$ 
with $\llp P_2\rrp \{(x\,\llp P_1\rrp) / x\} = M \tra{} M'$\\
\` by assumption\\
{\bf Subcase:} $M\tra{} M'$ due to redex in $\llp P_1\rrp$\\
$\llp P_1\rrp \tra{} M_0$ \` by assumption\\
$\exists Q_0$ such that $P_1 \mapsto^* Q_0$ and $\llp Q_0\rrp \equiv_\alpha M_0$
\` by i.h.\\
$(\nub y)x\langle y \rangle.(P_1\mid P_2) \mapsto^* (\nub y)x\langle y
\rangle.(Q_0\mid P_2)$ \` by compatibility of $\mapsto$\\
$\llp (\nub y)x\langle y \rangle.(Q_0\mid P_2)\rrp =
\llp P_2\rrp\{(x\,\llp Q_0\rrp)/ x\} = \llp P_2\rrp\{(x\,M_0)/ x\}$\\
{\bf Subcase:} $M \tra{} M'$ due to redex in $\llp P_2\rrp$\\
$\llp P_2\rrp\tra{} M_0$\` by assumption\\
$\exists Q_0$ such that $P_2 \mapsto^* Q_0$ and $\llp Q_0\rrp = M_0$ \` by i.h\\
$(\nub y)x\langle y \rangle.(P_1\mid P_2) \mapsto^* (\nub y)x\langle y
\rangle.(P_1 \mid Q_0)$ \` by compatibility of $\mapsto$\\
$\llp (\nub y)x\langle y \rangle.(P_1\mid Q_0)\rrp =
\llp Q_0\rrp\{(x\,\llp P_1\rrp)/ x\} = M_0\{x\,\llp P_1\rrp)/ x\}$\\
\end{tabbing}

\item[Case:] 
\[
\inferrule*[left=$(\m{copy})$]
{\Omega ; \Ga , u{:}A ; \D ,x{:}A\vdash P_1 :: z{:}C}
{\Omega ; \Ga , u{:}A ; \D \vdash (\nub x)u\langle x \rangle.P_1 :: z{:}C}
\]
\begin{tabbing}
$\llp (\nub x)u\langle x \rangle.P_1\rrp  = \llp P_1\rrp\{u/x\} = M \tra{} M'$ \`
by assumption\\
$\llp P_1\rrp \tra{} M_0$ \` by inversion on $\tra{}$\\
$\exists Q_0$ such that $P_1 \mapsto^* Q_0$ and $\llp Q_0\rrp =_\alpha M_0$ \` by
i.h.\\
$(\nub x)u\langle x \rangle.P_1 \mapsto^* (\nub x)u\langle x
\rangle.Q_0$ \` by compatibility\\
$\llp (\nub x)u\langle x \rangle.Q_0\rrp = \llp Q_0\rrp\{u/x\} = M_0\{u/x\}$
\end{tabbing}
\item[Case:] 
\[
\inferrule*[left=$(\lft\forall)$]
{\Omega \vdash B\,\m{type} \quad \Omega  ; \Ga ; \D , x{:}A\{B/X\}
  \vdash P_1 :: z{:}C}
{\Omega ; \Ga ; \D , x{:}\forall X.A \vdash x\langle B\rangle.P_1 :: z{:}C}
\]
\begin{tabbing}
$\llp x\langle B\rangle.P_1\rrp = \llp P_1\rrp \{x[B]/x\}$ with
$\llp P_1\rrp\{x[B]/x\} \tra{} M$ \` by assumption\\
$\llp P_1\rrp \tra{} M_0$ \` by inversion\\
$\exists Q_0$ such that $P_1 \mapsto^* Q_0$ and $\llp Q_0\rrp =_\alpha M_0$ \` by
i.h.\\
$x\langle B\rangle.P_1 \mapsto^* x\langle B\rangle.Q_0$ \` by
compatibility\\
$\llp x\langle B\rangle.Q_0\rrp = \llp Q_0\rrp\{x[B]/x\} = M_0\{x[B]/x\}$
\end{tabbing}

\item[Case:] 
\[
\inferrule*[left=$(\m{cut})$]
{\Omega ; \Ga ; \D_1 \vdash P_1 :: x{:}A \quad \Omega ; \Ga ; \D_2 , x{:}A
  \vdash P_2 :: z{:}C}
{\Omega ; \Ga ; \D_1 , \D_2 \vdash (\nub x)(P_1 \mid P_2) :: z{:}C}
\]
\begin{tabbing}
$\llp (\nub x)(P_1 \mid P_2)\rrp = \llp P_2\rrp\{\llp P_1\rrp/x\}$ with
$\llp P_2\rrp\{\llp P_1\rrp/x\} = M \tra{} M'$ \` by assumption\\
{\bf Subcase:} $M \tra{} M'$ due to redex in $\llp P_1\rrp$\\
$\llp P_1\rrp \tra{} M_0$ \` by assumption\\
$\exists Q_0$ such that $P_1 \mapsto^* Q_0$ and $\llp Q_0\rrp =_\alpha M_0$ \` by
i.h.\\
$(\nub x)(P_1 \mid P_2) \mapsto^* (\nub x)(Q_0 \mid P_2)$ \` by
reduction\\
$\llp (\nub x)(Q_0 \mid P_2)\rrp = \llp P_2\rrp\{\llp Q_0\rrp/x\} =
\llp P_2\rrp\{M_0/x\}$\\
{\bf Subcase:} $M \tra{} M'$ due to redex in $\llp P_2\rrp$\\
$\llp P_2\rrp \tra{} M_0$ \` by assumption\\
$\exists Q_0$ such that $P_2 \mapsto^* Q_0$ and $\llp Q_0\rrp = M_0$ \` by
i.h.\\
$(\nub x)(P_1 \mid P_2) \mapsto^* (\nub x)(Q_0 \mid P_2)$ \` by
compatibility\\
$\llp (\nub x)(P_1 \mid Q_0)\rrp  = \llp Q_0\rrp \{\llp P_1\rrp/x\} =
M_0\{\llp P_1\rrp/x\}$\\
{\bf Subcase:} $M \tra{} M'$ where the redex arises due to the
substitution of $\llp P_1\rrp$ for $x$\\
{\bf Subsubcase:} Last rule of deriv. of $P_2$ is a left rule on
$x$:\\
In all cases except $\lft\bang$, a top-level process reduction is exposed
(viz. Theorem~\ref{thm:opc1}). \\
If last rule is $\lft\bang$, then either $x$ does not occur in $P_2$
and we conclude by $\mapsto$.\\
{\bf Subsubcase:} Last rule of deriv. of $P_2$ is not a left rule on
$x$:\\
For rule $(\m{id})$ we have a process reduction immediately. In all
other cases either\\ there is no possible $\beta$-redex or we can 
conclude via compatibility of $\mapsto$.
\end{tabbing}

\item[Case:] 
\[
\inferrule*[left=$(\m{cut}^\bang)$]
{\Omega ; \Ga ; \cdot \vdash P_1 :: x{:}A \quad \Omega ; \Ga , u{:}A; \D 
  \vdash P_2 :: z{:}C}
{\Omega ; \Ga ; \D \vdash (\nub u)(\bang u(x).P_1 \mid P_2) :: z{:}C}
\]
\begin{tabbing}
$\llp (\nub u)(\bang u(x).P_1 \mid P_2)\rrp = \llp P_2\rrp\{\llp P_1\rrp/u\}$ with
$\llp P_2\rrp\{\llp P_1\rrp/u\} \tra{} M$ \` by assumption\\
{\bf Subcase:} $M \tra{} M'$ due to redex in $\llp P_1\rrp$\\
$\llp P_1\rrp \tra{} M_0$ \` by assumption\\
$\exists Q_0$ such that $P_1 \mapsto^* Q_0$ and $\llp Q_0\rrp =_\alpha M_0$ \` by
i.h.\\
$(\nub u)(\bang u(x).P_1 \mid P_2) \mapsto^* (\nub u)(\bang u(x).Q_0
\mid P_2)$ \` by compatibility\\
$\llp (\nub u)(\bang u(x).Q_0 \mid P_2)\rrp = \llp P_2\rrp\{\llp Q_0\rrp/u\} =
\llp P_2\rrp\{M_0/u\}$ \\
{\bf Subcase:} $M \tra{} M'$ due to redex in $\llp P_2\rrp$\\
$\llp P_2\rrp \tra{} M_0$ \` by assumption\\
$\exists Q_0$ such that $P_2 \mapsto^* Q_0$ and $\llp Q_0\rrp = M_0$ \` by
i.h.\\
$(\nub u)(\bang u(x).P_1 \mid P_2) \mapsto^* (\nub u)(\bang u(x).P_1
\mid Q_0)$ \` by compatibility\\
$\llp (\nub u)(\bang u(x).P_1 \mid Q_0)\rrp = \llp Q_0\rrp\{\llp P_1\rrp/u\} =
M_0\{\llp P_1\rrp/u\}$ \\
{\bf Subcase:} $M \tra{} M'$ where the redex arises due to the
substitution of $\llp P_1\rrp$ for $u$\\
If last rule in deriv. of $P_2$ is $\m{copy}$ then we have $=$ terms
in $0$ process reductions.\\
Otherwise, the result follows by compatibility of $\mapsto$.
\end{tabbing}

In all other cases the $\lambda$-term in the image of the translation
does not reduce.
\end{description}
\end{proof}

\subsection{Proofs for \S~\ref{sec:fullabs} -- Inversion and Full Abstraction}\label{app:fabs}

The proofs below rely on the fact that all commuting conversions of
linear logic are sound observational equivalences in the sense of $\logsim$.

\thminv*

We prove (1) and (2) above separately.

\begin{theorem}
If $\Omega ; \Ga ; \D \vdash M : A$ then $\Omega ; \Ga ; \D \vdash \llp\lb
M \rb_z\rrp \cong M : A$
\end{theorem}
\begin{proof}
By induction on the given typing derivation.

\begin{description}
\item[Case:] Linear variable
\begin{tabbing}
$\llp \lb x\rb_z\rrp = x \cong x$
\end{tabbing}
\item[Case:] Unrestricted variable

\begin{tabbing}
$\lb u \rb_z = (\nub x)u\langle x \rangle.[x\leftrightarrow z]$ \` by
def.\\
$\llp (\nub x)(u\langle x \rangle.[x\leftrightarrow z])\rrp = u \cong u$
\end{tabbing}

\item[Case:] $\lambda$-abstraction

\begin{tabbing}
$\lb \lambda x.M \rb_z = z(x).\lb M\rb_z$\` by def.\\
$\llp z(x).\lb M\rb_z\rrp = \lambda x . \llp \lb M\rb_z\rrp \cong \lambda x.M$
\` by i.h. and congruence\\
\end{tabbing}

\item[Case:] Application
\begin{tabbing}
$\lb M\, N\rb_z = (\nub x)(\lb M \rb_x \mid (\nub y)x\langle y
\rangle.(\lb N\rb_y \mid [x\leftrightarrow z]))$ \` by def.\\
$\llp (\nub x)(\lb M \rb_x \mid (\nub y)x\langle y
\rangle.(\lb N\rb_y \mid [x\leftrightarrow z]))\rrp = \llp\lb
M\rb_x\rrp\,\llp\lb N\rb_y\rrp$ \` by def.\\
$\llp\lb  M\rb_x\rrp\,\llp\lb N\rb_y\rrp \cong M\,N$ \` by i.h. and
congruence\\
\end{tabbing}

\item[Case:] Exponential
\begin{tabbing}
$\lb \bang M\rb_z = \bang z(x).\lb M\rb_x$ \` by def.\\
$\llp \bang z(x).\lb M\rb_x\rrp = \bang \llp\lb M\rb_x\rrp \cong \llp\lb \bang
M\rb_z\rrp$ \` by def, i.h. and congruence\\
\end{tabbing}

\item[Case:] Exponential elim.
\begin{tabbing}
$\lb \llet{\bang u = M}{N}\rb_z = (\nub x)(\lb M\rb_x
\mid \lb N\rb_z\{x/u\})$ \` by def.\\
$\llp (\nub x)(\lb M\rb_x \mid \lb N\rb_z\{x/u\})\rrp = 
\llet{\bang u = \llp\lb M\rb_x\rrp} {\llp\lb N \rb_z\rrp}$ \` by def.\\
$\llet{\bang u
= \llp\lb M\rb_x\rrp}{\llp\lb N \rb_z\rrp} \cong \llet{\bang u = M}
{N}$ \` by congruence and i.h.\\
\end{tabbing}
\item[Case:] Multiplicative Pairing

\begin{tabbing}
$\lb \mpair{M}{N} \rb_z = (\nub y)z\langle y
\rangle.(\lb M\rb_y \mid \lb N\rb_z)$ \` by def.\\
$\llp (\nub y)z\langle y
\rangle.(\lb M\rb_y \mid \lb N\rb_z)\rrp = \mpair{\llp\lb M\rb_y\rrp}{
\llp \lb N\rb_z\rrp}$ \` by def.\\
$\mpair{ \llp\lb M\rb_y\rrp}{ \llp\lb N\rb_z\rrp} \cong\mpair{M}{N}$ \` by i.h. and congruence\\
\end{tabbing}

\item[Case:] Mult. Pairing Elimination
\begin{tabbing}
$\lb\llet{x\tensor y  = M}{N}\rb_z = (\nub y)(\lb M\rb_x
\mid x(y).\lb N \rb_z)$\` by def.\\
$\llp (\nub y)(\lb M\rb_x
\mid x(y).\lb N \rb_z)\rrp = \llet{ x\tensor y  = \llp\lb M\rb_x\rrp}
{\llp\lb N \rb_z\rrp}$ \` by def.\\
$\llet{ x\tensor y  = \llp \lb M\rb_x\rrp}{\llp\lb N \rb_z\rrp}
\cong \llet{ x\tensor y  = M}{N}$ \` by i.h. and congruence\\
\end{tabbing}

\item[Case:] $\Lambda$-abstraction

\begin{tabbing}
$\llp \lb \Lambda X.M \rb_z\rrp = \Lambda X.\llp \lb M\rb_z\rrp \cong \Lambda
X.M$ \` by i.h. and congruence
\end{tabbing}

\item[Case:] Type application

\begin{tabbing}
$\llp \lb M[A] \rb_z\rrp = \llp\lb M\rb_z\rrp [A] \cong M[A]$ \` by i.h. and
congruence\\
\end{tabbing}
\item[Case:] Existential Intro.

\begin{tabbing}
$\llp\lb \pack{A}{M}\rb_z\rrp =
\pack{A}{\llp \lb M\rb_z\rrp} \cong \pack{A}{M}$ \` by i.h. and congruence\\
\end{tabbing}

\item[Case:] Existential Elim.

\begin{tabbing}
$\llp\lb \llet{(X,y) = M}{N}\rb_z\rrp = \llet{(X,y) = \llp\lb
M\rb_x\rrp}{\llp\lb N\rb_z\rrp} \cong 
\llet{(X,y) = M}{N}$\\
\` by i.h. and congruence\\
\end{tabbing}
\end{description}
\end{proof}

\begin{theorem}
If $\Omega ; \Ga ; \D \vdash P :: z{:}A$ then $\Omega ; \Ga ; \D \vdash
\lb\llp  P\rrp\rb_z \obseq P :: z{:}A$
\end{theorem}

\begin{proof}
By induction on the given typing derivation.

\begin{description}

\item[Case:] $(\m{id})$ or any right rule

Immediate by definition in the case of $(\m{id})$ and by i.h. and
congruence in all other cases.

\item[Case:] $\lft\lolli$

\begin{tabbing}
$\llp (\nub y)x\langle y \rangle.(P \mid Q)\rrp = \llp Q\rrp\{(x\,\llp
P\rrp)/x\}$ \`
by def.\\
$\lb \llp Q\rrp \{(x\,\llp P\rrp))/x\} \rb_z \logsim (\nub a)(\lb (x\,\llp P\rrp)\rb_a \mid
\lb \llp Q\rrp \rb_z\{a/x\})$ \` by Lemma~\ref{lem:compos}, with $a$ fresh\\
$= (\nub a)( (\nub w)([x\leftrightarrow w] \mid (\nub y)w\langle y
\rangle.(\lb\llp P\rrp\rb_y \mid [w\leftrightarrow a]) ) \mid \lb \llp
Q\rrp
\rb_z\{a/x\})$ \` by def.\\
$\tra{} (\nub a)( (\nu y)x\langle y \rangle.(\lb \llp P\rrp \rb_y \mid
[x\leftrightarrow a]) \mid \lb \llp Q\rrp \rb_z\{a/x\})$ \` by reduction\\
$\obseq (\nub y)x\langle y \rangle.(\lb \llp P\rrp\rb_y \mid
\lb\llp Q\rrp\rb_z)$ \` commuting conversion + reduction\\
$\obseq (\nub y)x\langle y \rangle.(P \mid
Q)$ \` by i.h. + congruence\\
\end{tabbing}
\item[Case:] $\lft\tensor$

\begin{tabbing}
$\llp x(y).P\rrp =\llet{x\tensor y = x}{\llp P\rrp}$ \` by def.\\
$\lb\llet{x\tensor y = x}{\llp P\rrp}\rb_z = (\nub
w)([x\leftrightarrow w] \mid w(y).\lb \llp P\rrp\rb_z)$ \` by def.\\ 
$\tra{} x(y).\lb \llp P\rrp \rb_z \obseq x(y).P$ \` by i.h. and congruence\\
\end{tabbing}

\item[Case:] $\lft\bang$

\begin{tabbing}
$\llp P\{x/u\}\rrp = \llet{\bang u = x}{\llp P\rrp}$ \` by def.\\
$\lb \llet{\bang u = x}{\llp P\rrp} \rb_z = (\nub
w)([x\leftrightarrow w]
\mid \lb\llp P\rrp\rb_z\{w/u\})$ \` by def.\\
$\tra{} \lb\llp P\rrp\rb_z\{x/u\} \logsim P\{x/u\}$ \` by i.h.
\end{tabbing}

\item[Case:] $\m{copy}$

\begin{tabbing}
$\llp (\nub x) u\langle x\rangle.P \rrp = \llp P\rrp\{u/x\}$ \` by
def.\\
$\lb  \llp P\rrp\{u/x\} \rb_z \logsim (\nub x)(\ov{u}\langle
w\rangle.[w\leftrightarrow x] \mid \lb \llp P \rrp\rb_z)$ \` by
Lemma~\ref{lem:compos}\\
$\logsim (\nub x)(\ov{u}\langle w \rangle.[w\leftrightarrow x] \mid
P)$ \` by i.h. and congruence\\
$\logsim (\nub x)u\langle x \rangle.P$ \` by definition of $\logsim$
for open processes\\ \` (i.e. closing for $u{:}A$ and observing that no
actions on $z$ are blocked)
\end{tabbing}

\item[Case:] $\lft\forall$ 

\begin{tabbing}
  $\llp x\langle B\rangle.P \rrp = \llp P\rrp\{(x[B])/x\}$ \` by def.\\
  $\lb\llp P\rrp\{(x[B])/x\}\rb_z \logsim (\nub a)( \lb x[B]\rb_a \mid \lb
  \llp P\rrp\rb_z\{a/x\})$ 
  \` by Lemma~\ref{lem:compos}, with $a$ fresh\\ 
  $(\nub a)( (\nub w)([x\leftrightarrow w] \mid w\langle B
  \rangle.[w\leftrightarrow a]) \mid \lb \llp P\rrp \rb_z\{a/x\})$ \` by def.\\
$\tra{} (\nub a)(x\langle B\rangle.[x\leftrightarrow a] \mid \lb \llp P\rrp\rb_z\{a/x\})$ \\
$\obseq x\langle B \rangle.\lb\llp P\rrp\rb_z$ \` commuting conversion + reduction\\
$\obseq x\langle B\rangle.P$ \` by i.h. + congruence
\end{tabbing}

\item[Case:] $\lft\exists$

\begin{tabbing}
$\llp x(Y).P\rrp =\llet{(Y,x) = x}{\llp P\rrp}$ \` by def.\\
$\lb \llet{(Y,x) = x}{\llp P\rrp} \rb_z = 
(\nub y)([x\leftrightarrow y] \mid y(Y).\lb \llp P\rrp\rb_z)$ \` by def.\\
$\tra{} x(Y).\lb\llp P\rrp\rb_z\{y/x\})$ \` by reduction\\
$\obseq x(Y).P$ \` by i.h. + congruence
\end{tabbing}

\item[Case:] $\m{cut}$

\begin{tabbing}
$\llp (\nub x)(P \mid Q)\rrp = \llp Q\rrp \{\llp P\rrp/x\}$ \` by definition\\
$\lb\llp Q\rrp\{\llp P\rrp/x\}\rb_z \logsim (\nub y)(\lb\llp P\rrp \rb_y
\mid 
\lb\llp Q\rrp\rb_z\{y/x\})$ \` by Lemma~\ref{lem:compos}, with $y$ fresh\\
$\equiv (\nub x)(P \mid Q)$ \` by i.h. + congruence and $\equiv_\alpha$
\end{tabbing}

\item[Case:] $\m{cut}^!$

\begin{tabbing}
$\#((\nub u)(\bang u(x).P \mid Q)) = \llp Q\rrp\{\llp P\rrp/u\}$ \` by
definition\\ 
$\lb \llp Q\rrp\{\llp P\rrp/u\}\rb_z \logsim (\nub u)(\bang u(x).\lb\llp
P\rrp\rb_x 
\mid \lb\llp Q\rrp \rb_z\{v/u\})$ \` by Lemma~\ref{lem:compos}\\
$\obseq (\nub u)(\bang u(x).P \mid Q)$ \` by i.h. + congruence and $\equiv_\alpha$
\end{tabbing}
\end{description}
\end{proof}

 \begin{lemma}\label{thm:values}
 $M\Downarrow \TT$ iff $\lb M \rb_z \logsim \lb \TT\rb_z :: z{:}\mathbf{2}$
 \end{lemma}
 \begin{proof}
 By operational correspondence.
 \end{proof}

 \thmfaltp*
\begin{proof}
({\bf Soundness}, $\Leftarrow$)
Since $\cong$ is the largest consistent congruence compatible with
the booleans, let $M
\mathcal{R} N$ iff $\lb M \rb_z \obseq \lb N \rb_z$. We show that
$\mathcal{R}$ is one such relation.
\begin{enumerate}
\item 
(Congruence) 
Since $\obseq$ is a congruence, $\mathcal{R}$ is a
congruence. 
\item 
(Reduction-closed) 
Let $M \tra{} M'$ and $\lb M \rb_z \obseq \lb N
\rb_z$. Then we have by operational correspondence (Theorem~\ref{thm:ftopioc}) that $\lb M \rb_z
\tra{}^* P$ such that $M' \ll P$ and we have that $P \obseq \lb
M'\rb_z$ hence $\lb M'\rb_z \obseq \lb N\rb_z$, thus $\mathcal{R}$ is
reduction closed.

\item 
(Compatible with the booleans) Follows from
Lemma~\ref{thm:values}.
\end{enumerate}
({\bf Completeness}, $\Rightarrow$)
Assume to the contrary that $M \cong N : A$ and 
$\lb M \rb_z \not\obseq \lb N \rb_z :: z{:}A$.

This means we can find a distinguishing context 
$R$ such that $(\nub z,\tilde{x})(\lb M\rb_z
\mid R) \obseq \lb\TT\rb_y :: y{:}\lb\mathbf{2}\rb$ and
$(\nub z,\tilde{x})(\lb N\rb_z
\mid R) \obseq \lb\FF\rb_y :: y{:}\lb\mathbf{2}\rb$.
By Fullness (Theorem~\ref{cor:fullb}), we have that there exists some $L$ such that $\lb L\rb_y
\obseq R$, thus: $(\nub z,\tilde{x})(\lb M\rb_z
\mid \lb L\rb_y) \obseq \lb\TT\rb_y :: y{:}\lb\mathbf{2}\rb$ and
$(\nub z,\tilde{x})(\lb N\rb_z
\mid \lb L\rb_y) \obseq \lb\FF\rb_y :: y{:}\lb\mathbf{2}\rb$.
By Theorem~\ref{thm:fa_ltp1} (Soundness), we have that $L[M] \cong \TT$ and
$L[N] \cong \FF$ and thus $L[M] \not\cong L[N]$ which contradicts $M
\cong N : A$. 
\end{proof}

\thmfaptl*

\begin{proof}
({\bf Soundness}, $\Leftarrow$)
Let $M = \llp P \rrp$ and $N = \llp Q \rrp$. By
Theorem~\ref{thm:fa_ltp2} (Completeness) we have $\lb M\rb_z \obseq \lb
N\rb_z$. Thus by Theorem~\ref{thm:inv} we have: $\lb M\rb_z = \lb \llp P \rrp\rb_z
\obseq P$ and $\lb N\rb_z = \lb \llp Q \rrp\rb_z \obseq Q$. 
By compatibility with observational equivalence we have $P \obseq Q ::
z{:}A$.\\
%
({\bf Completeness}, $\Rightarrow$)
From $P \obseq Q :: z{:}A$, Theorem~\ref{thm:inv} and compatibility
with observational equivalence we have $\lb \llp P \rrp \rb_z \obseq \lb \llp
 Q \rrp \rb_z :: z{:}A$.
Let $\llp P \rrp = M$ and $\llp Q \rrp = N$. We have by
Theorem~\ref{thm:fa_ltp1} (Soundness) 
that $M\cong N : A$ and thus $\llp P \rrp
\obseq \llp Q \rrp : A$.
\end{proof}

\subsection{Proofs for \S~\ref{sec:hovals} -- Communicating Values}

\begin{lemma}[Type Soundness of Encoding]~
\begin{enumerate}
\item If $\Psi \vdash M : \tau$ then $\lb \Psi \rb ; \cdot \vdash \lb M
\rb_z :: z{:}\lb \tau\rb$
\item If $\Psi ; \Ga ; \D \vdash P :: z{:}A$ then $ \lb \Psi
  \rb , \lb\Ga\rb ; \lb \D \rb \vdash \lb P \rb :: z{:}\lb A \rb$
\end{enumerate}
\end{lemma}
\begin{proof}
Straightforward induction.
\end{proof}

\lemcompvals*

\begin{proof}
By induction on the typing for $M$. We make use of the fact that
$\logsim$ includes $\equiv_\bang$.
\begin{description}
\item[Case:] $M = y$ with $y=x$
\begin{tabbing}
$\lb M\{N/x\}\rb_z = \lb N\rb_z$ \\
$(\nub x)(\lb M \rb_z \mid \bang x(y).\lb N \rb_y) = 
 (\nub x)(\ov{x}\langle y \rangle.[y\leftrightarrow z] \mid \bang
 x(y).\lb N \rb_y)$ \` by definition\\
$\tra{}^+ (\nub x)(\lb N \rb_z \mid \bang x(y).\lb N\rb_y)$ \` by the
reduction semantics\\
$\logsim \lb N \rb_z$ \` by $\equiv_\bang$,  since $x\not\in
\mathsf{fn}(\lb N \rb_z)$
\end{tabbing}
\item[Case:] $M = y$ with $y\neq x$

\begin{tabbing}
$\lb M\{N/x\}\rb_z = \lb y \rb_z =  \ov{y}\langle w \rangle.[w\leftrightarrow z]$ \\
$(\nub x)(\lb M \rb \mid \bang x(y).\lb N \rb_y) = 
(\nub x)(\ov{y}\langle w \rangle.[w\leftrightarrow z] \mid \bang
 x(y).\lb N \rb_y)$ \` by definition\\
$\logsim \ov{y}\langle w \rangle.[w\leftrightarrow z]$ \` by $\equiv_\bang$
\end{tabbing} 

\item[Case:] $M = M_1\,M_2$
\begin{tabbing}
$\lb M_1\,M_2\{N/x\}\rb_z = \lb M_1\{N/x\} \, M_2\{N/x\}\rb_z =$\\
$(\nub y)(\lb M_1 \{N/x\}\rb_y \mid \ov{y}\langle u \rangle.(\bang u(w).\lb M_2\{N/x\}\rb_w \mid [y\leftrightarrow z]) $ \` by
definition\\
$(\nub x)(\lb M_1\,M_2\rb_z \mid \bang x(y).\lb N\rb_y) = 
(\nub x)( (\nub y)(\lb M_1 \rb_y \mid \ov{y}\langle u \rangle.(\bang
u(w).\lb M_2\rb_w \mid [y\leftrightarrow z]) \mid \bang x(y).\lb
N\rb_y))$\\
\` by definition\\
$\lb M_1 \{N/x\}\rb_y \logsim (\nub x)(\lb M_1\rb_y \mid \bang
x(a).\lb N \rb_a)$ \` by i.h.\\
$\lb M_2 \{N/x\}\rb_w \logsim (\nub x)(\lb M_2\rb_w \mid \bang
x(a).\lb N \rb_a)$ \` by i.h.\\
$\lb M_1\,M_2\{N/x\}\rb_z \logsim (\nub y)((\nub x)(\lb M_1\rb_y \mid \bang
x(a).\lb N \rb_a) \mid \ov{y}\langle u \rangle.(\bang u(w).\lb
M_2\{N/x\}\rb_w \mid [y\leftrightarrow z])) $\\ \` by congruence\\
$\logsim (\nub y)((\nub x)(\lb M_1\rb_y \mid \bang
x(a).\lb N \rb_a) \mid \ov{y}\langle u \rangle.(\bang u(w).(\nub x)(\lb M_2\rb_w \mid \bang
x(a).\lb N \rb_a) \mid [y\leftrightarrow z]))$\\ \` by congruence\\
$\logsim (\nub x)(\nub y)(\lb M_1 \rb_y \mid \ov{y}\langle u
\rangle.(\bang u(w).\lb M \rb_w \mid [y\leftrightarrow z] \mid \bang
x(a).\lb N\rb_a))$ \` by $\equiv_\bang$
\end{tabbing}

\item[Case:] $M = \lambda y{:}\tau_0 . M'$
\begin{tabbing}
$\lb\lambda y{:}\tau_0 . M' \{N/x\}\rb_z
= z(y).\lb M'\{N/x\}\rb_z$ \\
$(\nub x)(\lb M \rb_z \mid \bang x(y).\lb N \rb_y) = 
 (\nub x)(z(y).\lb M'\rb_z \mid \bang x(y).\lb N \rb_y)$ \` by
 definition\\
$\lb M'\{N/x\}\rb_z \logsim (\nub x)(\lb M \rb_z \mid \bang x(w).\lb
N\rb_w)$ \` by i.h.\\
$\lb\lambda y{:}\tau_0 . M' \{N/x\}\rb_z \logsim z(y).(\nub x)(\lb M' \rb_z \mid \bang x(w).\lb
N\rb_w)$ \` by congruence\\
$\logsim (\nub x)(z(y).\lb M' \rb_z \mid \bang x(w).\lb N\rb_w)$ \` by
commuting conversion
\end{tabbing}

\end{description}
\end{proof}

\begin{theorem}[Operational Completeness]
\label{lem:opccomp}~
\begin{enumerate}
\item If $\Psi \vdash M : \tau$ and $M \tra{} N$ then $\lb M \rb_z
  \wtra{} P$ such that $P \logsim \lb N\rb_z$
\item If $\Psi ; \Ga ; \D \vdash P :: z{:}A$ and $P \tra{} Q$ then
$\lb P \rb \tra{}^+ R$ with $R \logsim \lb Q \rb$
\end{enumerate}
\end{theorem}

\begin{proof}
We proceed by induction on the given derivation and case analysis on
the reduction.

\begin{description}
\item[Case:] $M = (\lambda x{:}\tau . M')\,N'$ with $M\tra{} M'\{N'/x\}$

\begin{tabbing}
$\lb M \rb_z = (\nub y)(\lb \lambda x{:}\tau . M'\rb_y \mid \ov
y\langle x \rangle.(\bang x(w).\lb N' \rb_w \mid [y\leftrightarrow z])
=$\\
$(\nub y)(y(x).\lb M' \rb_y \mid \ov
y\langle x \rangle.(\bang x(w).\lb N' \rb_w \mid [y\leftrightarrow z])$
\` by definition of $\lb {-}\rb$\\
$\tra{}^+ (\nub x)(\lb M'\rb_z \mid\, \bang x(w).\lb N' \rb_w)$ \` by
the reduction semantics\\
$\logsim \lb M'\{N'/x\}\rb_z$ \` by Lemma~\ref{lem:comp}
\end{tabbing}

\item[Case:] $M = M_1\, M_2$ with $M \tra{} M_1'\,M_2$ by $M_1 \tra{}
  M_1'$
\begin{tabbing}
$\lb M_1\, M_2\rb_z = (\nub y)(\lb M_1\rb_y \mid \ov{y}\langle
x\rangle.(\bang x(w).\lb M_2\rb_w \mid [y\leftrightarrow z])$ \` by
definition\\
$\lb M_1'\, M_2\rb_z = (\nub y)(\lb M_1'\rb_y \mid \ov{y}\langle
x\rangle.(\bang x(w).\lb M_2\rb_w \mid [y\leftrightarrow z])$ \` by
definition\\
$\lb M_1\rb_y \wtra{} P_1'$ such that $P_1' \logsim \lb M_1'\rb_y$ \`
by i.h.\\
$\lb M_1\, M_2\rb_z \wtra{} (\nub y)(P_1' \mid \ov{y}\langle
x\rangle.(\bang x(w).\lb M_2\rb_w \mid [y\leftrightarrow z])$ \` by
reduction semantics\\
$\logsim (\nub y)(\lb M_1'\rb_y \mid \ov{y}\langle
x\rangle.(\bang x(w).\lb M_2\rb_w \mid [y\leftrightarrow z])$ \` by congruence
\end{tabbing}

\item[Case:] $P = (\nub x)(x\langle M \rangle.P' \mid x(y).Q')$ with $P \tra{} 
 (\nub x)(P' \mid Q'\{M / y\})$

\begin{tabbing}
$\lb P \rb = (\nub x)(\ov{x}\langle y \rangle.(\bang y(w).\lb M \rb_w \mid \lb P'
\rb) \mid x(y).\lb Q' \rb)$ \` by definition\\
$\lb P \rb \tra{} (\nub x,y)(\bang y(w).\lb M \rb_w \mid \lb P' \rb \mid \lb
Q'\rb)$ \` by the reduction semantics\\
$\lb (\nub x)(P' \mid Q'\{M/y\}) \rb =
 (\nub x)(\lb P'\rb \mid \lb Q'\{ M /y\}\rb)$ \` by definition\\
$\logsim (\nub x,y)(\lb P'\rb \mid \lb Q'\rb \mid \bang y(w).\lb M
\rb_w)$ \` by Lemma~\ref{lem:comp} and congruence\\
\end{tabbing}

All remaining cases follow straightforwardly by induction.

\end{description}
\end{proof}

\begin{theorem}[Operational Soundness]
\label{lem:opcsound}~
\begin{enumerate}
\item If $\Psi \vdash M : \tau$ and $\lb M\rb_z \tra{} Q$ then
$M \tra{}^+ N$ such that $\lb N \rb_z \logsim Q$
\item If $\Psi ; \Ga ; \D \vdash P :: z{:}A$ and $\lb P \rb \tra{} Q$ then
$P \tra{}^+ P'$ such that $\lb P'\rb \logsim Q$
\end{enumerate}
\end{theorem}

\begin{proof}
By induction on the given derivation and case analysis on the
reduction step.

\begin{description}
\item[Case:] $M = M_1\, M_2$ with $\lb M_1\rb_y \tra{} R$
\begin{tabbing}
$\lb M_1\, M_2\rb_z = (\nub y)(\lb M_1\rb_y \mid \ov{y}\langle x
\rangle.(\bang x(w).\lb M_2\rb_w \mid [y\leftrightarrow z]))$ \` by
definition\\
$\tra{} (\nub y)(R \mid \ov{y}\langle x
\rangle.(\bang x(w).\lb M_2\rb_w \mid [y\leftrightarrow z]))$ \` by
reduction semantics\\
$M_1 \tra{}^+ M_1'$ with $\lb M_1' \rb_y \logsim R$ \` by i.h.\\
$M_1\, M_2 \tra{}^+ M_1' \, M_2$ \` by the operational semantics\\
$\lb M_1'\, M_2\rb_z = (\nub y)(\lb M_1'\rb_y \mid \ov{y}\langle x
\rangle.(\bang x(w).\lb M_2\rb_w \mid [y\leftrightarrow z]))$ \` by
definition\\
$\logsim  (\nub y)(R \mid \ov{y}\langle x
\rangle.(\bang x(w).\lb M_2\rb_w \mid [y\leftrightarrow z]))$ \` by congruence
\end{tabbing}

\item[Case:] $M = M_1 \, M_2$ with $(\nub y)(\lb M_1\rb_y \mid \ov{y}\langle x
\rangle.(\bang x(w).\lb M_2\rb_w \mid [y\leftrightarrow z])) \tra{}
(\nub y,x)(R \mid \bang x(w).\lb M_2\rb_w \mid [y\leftrightarrow z])$

\begin{tabbing}
$\lb M_1\rb_y \equiv (\nub \ov{a})(y(x).R_1 \mid R_2)$ \` by the
reduction semantics, for some $R_1, R_2$ and $\ov{a}$\\
$\Psi \vdash M_1 : \tau_0 \rightarrow \tau_1$ \` by inversion\\
{\bf Subcase:} $M_1 = y$, for some $y \in \Psi$\\
Impossible reduction.\\
{\bf Subcase:} $M_1 = \lambda x{:}\tau_0.M_1'$\\
$(\lambda x{:}\tau_0 . M_1')\, M_2 \tra{} M_1'\{M_2/x\}$ \` by
operational semantics\\
$\lb M_1'\{M_2/x\} \rb_z \logsim (\nub x)(\lb M_1'\rb_z \mid \bang
x(w).\lb M_2\rb_w)$ \` by Lemma~\ref{lem:comp}\\
$\lb (\lambda x{:}\tau_0 . M_1')\, M_2 \rb_z = (\nub y)(y(x).\lb
M_1'\rb_y \mid \ov{y}\langle x
\rangle.(\bang x(w).\lb M_2\rb_w \mid [y\leftrightarrow z]))$ \` by
definition\\
$R = \lb M_1'\rb_y$ \` by inversion\\
$(\nub y,x)(R \mid \bang x(w).\lb M_2\rb_w \mid [y\leftrightarrow z])
\logsim (\nub x)(\lb M_1'\rb_z \mid \bang x(w).\lb M_2\rb_w)$ \` by
reduction closure\\
{\bf Subcase:} $M_1 = N_1\,N_2$, for some $N_1$ and $N_2$\\
$\lb N_1\, N_2\rb_y = (\nub a)(\lb N_1\rb_a \mid \ov{a}\langle
b\rangle.(\bang b(d).\lb N_2\rb_d \mid [a\leftrightarrow y]))$ \` by definition\\
Impossible reduction.
\end{tabbing}

\item[Case:] $P = (\nub x)(x\langle M \rangle.P_1 \mid x(y).P_2)$

\begin{tabbing}
$\lb P\rb = (\nub x)(\ov{x}\langle y \rangle.(\bang y(w).\lb M\rb_w
\mid \lb P_1 \rb) \mid x(y).\lb P_2\rb)$ \` by definition\\
$\lb P \rb \tra{} (\nub x,y)(\bang y(w).\lb M \rb_w \mid \lb P_1\rb \mid
\lb P_2\rb)$ \` by reduction semantics\\
$P \tra{} (\nub x)(P_1 \mid P_2\{M/y\})$ \` by reduction semantics\\
$\lb (\nub x)(P_1 \mid P_2\{M/y\}) \rb \logsim (\nub x,y)(\lb P_1 \rb
\mid \lb P_2\rb \mid \bang y(w).\lb M \rb_w)$ \` by
Lemma~\ref{lem:comp} and congruence\\
\end{tabbing}

\item[Case:] $P = (\nub x)(x\langle M \rangle.P_1 \mid P_2)$

\begin{tabbing}
$\lb P \rb = (\nub x)(\ov{x}\langle y \rangle.(\bang y (w).\lb M\rb_w
\mid \lb P_1\rb) \mid \lb P_2\rb)$ \` by definition\\
$\lb P \rb \tra{} (\nub x)(\ov{x}\langle y \rangle.(\bang y (w).\lb M\rb_w
\mid \lb P_1\rb) \mid R)$ \` assumption, with $\lb P_2 \rb \tra{} R$\\
$P_2 \tra{}^+ P_2'$ with $\lb P_2'\rb \logsim R$ \` by i.h.\\
$P \tra{}^+ (\nub x)(x\langle M \rangle.P_1 \mid P_2')$ \` by
reduction semantics\\
$\lb (\nub x)(x\langle M \rangle.P_1 \mid P_2')\rb = (\nub x)(\ov{x}\langle y \rangle.(\bang y (w).\lb M\rb_w
\mid \lb P_1\rb) \mid \lb P_2'\rb)$ \` by definition\\
$\logsim (\nub x)(\ov{x}\langle y \rangle.(\bang y (w).\lb M\rb_w
\mid \lb P_1\rb) \mid R)$ \` by congruence
\end{tabbing}

All other process reductions follow straightforwardly from the
inductive hypothesis.
\end{description}
\end{proof}


\begin{lemma}[Type Soundness of Encoding]
~
\begin{enumerate}
\item If $\Psi ; \Ga ; \D \vdash P :: z{:}A$ then $\llp \Psi \rrp , \llp \Ga
\rrp ; \llp \D \rrp \vdash \llp P \rrp : \llp A \rrp$
\item If $\Psi \vdash M : \tau$ then $\llp \Psi \rrp ; \cdot \vdash \llp M
  \rrp : \llp \tau \rrp$
\end{enumerate}
\end{lemma}

\begin{proof}
Straightforward induction.
\end{proof}

\begin{lemma}[Compositionality]~\label{lem:compinv}
\label{lem:compinvlam}
\begin{enumerate}
\item If $\Psi , x{:}\tau ; \Ga ; \D\vdash P :: z{:}B$ and 
and $\Psi \vdash M : \tau$ then $\llp P \{M/x\}\rrp =_\alpha \llp P
\rrp \{ \llp M \rrp /x\}$
\item   If $\Psi , x{:}\tau \vdash M : \sigma$ and $\Psi \vdash N : \tau$ then
   $\llp M \{ N/x\}\rrp =_\alpha \llp M \rrp \{\llp N \rrp / x\}$
\end{enumerate}
\end{lemma}

\begin{proof}
Straightforward induction.
\end{proof}

\begin{theorem}[Operational Soundness]
\label{lem:opcrs}~
  \begin{enumerate}
  \item If $\Psi ; \Ga ; \D \vdash P :: z{:}A$ and $\llp P \rrp \tra{}
    M$ then $P \mapsto^* Q$ such that $M =_\alpha \llp Q \rrp$
  \item If $\Psi \vdash M : \tau$ and $\llp M \rrp \tra{} N$ then $M
    \rightarrow_\beta^+ M'$ such that $N =_\alpha \llp M' \rrp$
  \end{enumerate}
\end{theorem}

\begin{proof}
We proceed by induction on the given reduction and case analysis on
typing.

\begin{description}
\item[Case:] $\llp P_0\rrp \{ (x\,\bang \llp M_0 \rrp)/ x\} \tra{}
  M$
\begin{tabbing}
$\llp P_0\rrp \{ (x\,\bang \llp M_0 \rrp)/ x\} \tra{} M'\{(x\,\bang
\llp M_0 \rrp)/ x\}$ \` by operational semantics\\
$P_0 \mapsto P_0'$ with $P_0' =_\beta M'$ \` by i.h.\\
$x\langle M_0\rangle.P_0 \mapsto x\langle M_0 \rangle.P_0'$ \` by
extended reduction\\
$\llp x\langle M_0 \rangle.P_0'\rrp =  \llp P_0'\rrp\{(x\,\bang
\llp M_0 \rrp)/ x\}$ \` by definition\\
$=_\alpha  M'\{(x\,\bang
\llp M_0 \rrp)/ x\}$ \` by congruence
\end{tabbing}

The other cases are covered by our previous result for the reverse
encoding of processes.

\item[Case:] $\llp M_0\rrp \, \bang \llp M_1\rrp \tra{} M_0' \, \bang
  \llp M_1 \rrp$

  \begin{tabbing}
    $\llp M_0 \rrp \tra{} M_0'$ \` by inversion\\
    $M_0 \tra{}^+_\beta M_0''$ such that $M_0' =_\alpha \llp M_0''
    \rrp$ \` by i.h.\\
    $M_0 \, M_1 \tra{}^+_\beta M_0'' \, M_1$ \` by operational
    semantics\\
    $\llp M_0'' \, M_1 \rrp = \llp M_0'' \rrp \, \bang \llp M_1 \rrp
    =_\alpha M_0' \, \bang
  \llp M_1 \rrp$ \` by definition and by congruence\\ 
  \end{tabbing}

\item[Case:] $(\lambda x{:}\bang \llp \tau_0\rrp . \llet{\bang x =
    x}{\llp M_0\rrp})\,\bang \llp M_1\rrp \tra{} 
 \llet{\bang x = \bang \llp M_1 \rrp}{\llp M_0\rrp}$
 
\begin{tabbing}
$(\lambda x{:}\tau_0.M_0)\, M_1 \tra{} M_0\{M_1/x\}$ \` by inversion
and operational semantics\\
$ \llet{\bang x = \bang \llp M_1 \rrp}{\llp M_0\rrp} \tra{} \llp M_0
\rrp \{\llp M_1 \rrp / x\}$ \` by operational semantics\\
$=_\alpha \llp M_0 \{ M_1/x\}\rrp$ \` by
Lemma~\ref{lem:compinvlam}
\end{tabbing}
\end{description}
\end{proof}

\begin{theorem}[Operational Completeness]
\label{lem:opcrc}~
\begin{enumerate}
\item If $\Psi ; \Ga ; \D \vdash P :: z{:}A$ and $P \tra{} Q$ then
  $\llp P \rrp \tra{}_\beta ^* \llp Q \rrp$
\item If $\Psi \vdash M : \tau$ and $M \tra{} N$ then $\llp M \rrp
  \tra{}^+ \llp N \rrp$.
\end{enumerate}
\end{theorem}

\begin{proof}
We proceed by induction on the given reduction.
\begin{description}
\item[Case:] $(\nub x)(x\langle M \rangle.P_1 \mid x(y).P_2) \tra{} (\nub
  x)(P_1 \mid P_2\{M/x\})$ with $P$ typed via $\cut$ of $\rgt\wedge$
  and $\lft\wedge$
\begin{tabbing}
$\llp P \rrp = \llet{y\tensor x = \langle \bang \llp M \rrp \tensor
  \llp P_1 \rrp\rangle}{\llet{\bang y = y}{\llp P_2\rrp}}$ \` by definition\\
$\tra{} \llet{\bang y = \bang \llp M \rrp}{\llp P_2\rrp\{ \llp P_1\rrp
  / x\}}$ \` by operational semantics\\
$\tra{} \llp P_2\rrp\{\llp P_1 \rrp /x \}\{\llp M\rrp / x\}$ \` by
operational semantics\\
$\llp (\nub  x)(P_1 \mid P_2\{M/x\})\rrp =
\llp P_2\{M/x\} \rrp \{ \llp P_1 \rrp /x \}$ \` by definition\\
$=_\alpha \llp P_2\rrp\{\llp P_1 \rrp /x \}\{\llp M\rrp / x\}$ \` by Lemma~\ref{lem:compinv}
\end{tabbing}

\item[Case:]  $(\nub x)(x(y).P_1 \mid x\langle M \rangle.P_2) \tra{} (\nub
  x)(P_1\{M/x\} \mid P_2)$ with $P$ typed via $\cut$ of $\rgt\supset$
  and $\lft\supset$

\begin{tabbing}
$\llp P \rrp = \llp P_2\rrp\{ (\lambda x{:}\bang\llp \tau_0\rrp
. \llet{\bang x = x}{\llp
P_1\rrp})\, \bang\llp M \rrp/x\}$ \` by definition\\
$\tra{}_\beta^+ \llp P_2 \rrp \{ (\llp P_1 \rrp\{\llp M \rrp /x\})/
x\}$ \` by $\beta$ conversion\\
$\llp (\nub
  x)(P_1\{M/x\} \mid P_2)\rrp = \llp P_2 \rrp \{ \llp P_1 \{ M /x
  \}\rrp / x\}$ \` by definition\\
$=_\alpha\llp P_2 \rrp \{ (\llp P_1 \rrp\{\llp M \rrp /x\})/
x\}$ \` by Lemma~\ref{lem:compinv}
\end{tabbing}

The remaining process cases follow by induction.  

\item[Case:] $(\lambda x{:}\tau_0.M_0)\,M_1 \tra{} M_0\{M_1/x\}$
\begin{tabbing}
$\llp M \rrp = (\lambda x{:}\bang \llp \tau_0\rrp . \llet {\bang x =
  x}{\llp M_0 \rrp})\, \bang \llp M_1\rrp$ \` by definition\\
$\tra{}^+ \llp M_0 \rrp\{\llp M_1 \rrp / x\}$
$=_\alpha \llp M_0\{ M_1 /x \}\rrp$ \` by operational semantics and  Lemma~\ref{lem:compinv}
\end{tabbing}

\item[Case:] $M_0\,M_1 \tra{} M_0'\,M_1$ by $M_0 \tra{} M_0'$

\begin{tabbing}
$\llp M_0 \, M_1 \rrp = \llp M_0 \rrp \, \bang \llp M_1\rrp$ \` by
definition\\
$\llp M_0' \, M_1 \rrp = \llp M_0' \rrp \, \bang \llp M_1\rrp$ \` by
definition\\
$\llp M_0 \tra{}^+ \llp M_0'\rrp$ \` by i.h.\\
$\llp M_0 \rrp \, \bang \llp M_1\rrp \tra{}^+ \llp M_0' \rrp \, \bang
\llp M_1 \rrp$ \` by operational semantics
\end{tabbing}
\end{description}
\end{proof}


\thminvv*

We establish the proofs of the two statements separately:

\begin{theorem}[Inverse Encodings -- Processes]
If $\Psi ; \Ga ; \D \vdash P :: z{:}A$
then $\lb\llp P \rrp\rb_z \logsim \lb P\rb$
\end{theorem}

\begin{proof}
By induction on typing.

\begin{description}
\item[Case:] $\rgt{\wedge}$
\begin{tabbing}
$P = z\langle M \rangle.P_0$ \` by assumption\\
$\llp P \rrp = \langle \bang \llp M \rrp \tensor \llp P_0 \rrp\rangle$
\` by definition\\
$\lb \langle  \bang \llp M \rrp \tensor \llp P_0 \rrp\rangle \rb_z = $
 $\ov{z}\langle x\rangle.(\bang x(u).\lb \llp M\rrp \rb_u \mid \lb\llp
 P_0 \rrp\rb_z)$ \` by definition\\
$\lb z\langle M \rangle.P_0 \rb = \ov{z}\langle x \rangle.(\bang
x(u).\lb M\rb_u \mid \lb P_0 \rb)$ \` by definition\\
$\logsim \ov{z}\langle x\rangle.(\bang x(u).\lb \llp M\rrp \rb_u \mid \lb\llp
 P_0 \rrp\rb_z)$ \` by i.h. and congruence\\
\end{tabbing}
\item[Case:] $\lft\wedge$
\begin{tabbing}
$P = x(y).P_0$ \` by assumption\\
$\llp P \rrp = \llet{y\tensor x = x}{\llet{\bang y = y}{\llp
    P_0\rrp}}$ \` by definition\\
$\lb \llet{y\tensor x = x}{\llet{\bang y = y}{\llp
    P_0\rrp}} \rb_z = 
x(y). \lb\llp P_0\rrp\rb_z$ \` by definition\\
$\lb x(y).P_0\rb = x(y).\lb P_0 \rb$ \` by definition\\
$\logsim  x(y). \lb\llp P_0\rrp\rb_z$ \` by i.h. and congruence
\end{tabbing}

\item[Case:] $\rgt\supset$ 

\begin{tabbing}
$P = x(y).P_0$ \` by assumption\\
$\llp P \rrp = \lambda x{:}\bang \llp \tau \rrp.\llet{\bang x =
  x}{\llp P_0\rrp}$ \` by definition\\
$\lb \lambda x{:}\bang \llp \tau \rrp.\llet{\bang x =
  x}{\llp P_0\rrp}\rb_z = x(y).\lb \llp P_0 \rrp\rb_z$ \` by
definition\\
$\lb x(y).P_0 \rb = x(y).\lb P_0 \rb$ \` by definition\\
$\logsim x(y).\lb \llp P_0 \rrp\rb_z$ \` by i.h. and congruence
\end{tabbing}

\item[Case:] $\lft\supset$ 

\begin{tabbing}
$P = x\langle M \rangle.P_0$ \` by assumption\\
$\llp P \rrp = \llp P_0 \rrp\{(x\,\bang \llp M \rrp) / x\}$ \` by
definition\\
$\lb \llp P_0 \rrp\{(x\,\bang \llp M \rrp) / x\}\rb_z =
(\nub a)(\lb x\, \bang \llp M\rrp\rb_a \mid \lb \llp P_0 \rrp
\rb_z\{a/x\})$ \` by Lemma~\ref{lem:compos}\\
$= (\nub a)((\nub b)(\lb x\rb_b \mid \ov{b}\langle c\rangle.( \lb \bang \llp M\rrp\rb_c\mid
[b\leftrightarrow a]) \mid \lb \llp P_0 \rrp
\rb_z\{a/x\})$ \` by definition\\
$= (\nub a)((\nub b)([x\leftrightarrow b] \mid \ov{b}\langle c\rangle.( \bang
c(w). \lb\llp M\rrp\rb_w\mid
[b\leftrightarrow a]) \mid \lb \llp P_0 \rrp
\rb_z\{a/x\}))$ \` by definition\\
$\tra{} (\nub a)(\ov{x}\langle c \rangle.(\bang c(w).\lb\llp M\rrp \rb_w \mid [x\leftrightarrow a]) \mid 
\lb \llp P_0 \rrp
\rb_z\{a/x\})$ \` by reduction semantics\\
$\logsim \ov{x}\langle c \rangle.(\bang c(w).\lb\llp M\rrp \rb_w \mid \lb \llp P_0 \rrp
\rb_z)$ \` by commuting conversion and reduction \\
$\logsim \lb P \rb = \ov{x}\langle y \rangle.(\bang y(u).\lb M \rb_u \mid \lb
P_0 \rb)$ \` by i.h. and congruence\\
\end{tabbing}

\end{description}
\end{proof}

\begin{theorem}[Inverse Encodings -- $\lambda$-terms]
If $\Psi \vdash M : \tau$ then $\llp\lb M \rb_z\rrp =_{\beta} \llp M \rrp$
\end{theorem}

\begin{proof}
  By induction on typing.

  \begin{description}
  \item[Case:] Variable
    \begin{tabbing}
    $\lb M \rb_z = \ov{x}\langle y \rangle.[y\leftrightarrow z]$ \` by
    definition\\
    $\llp \ov{x}\langle y \rangle.[y\leftrightarrow z] \rrp = x$ \` by
    definition  
    \end{tabbing}
  \item[Case:] $\lambda$-abstraction

    \begin{tabbing}
   $\lb \lambda x{:}\tau_0.M_0 \rb_z = z(x).\lb M_0\rb_z$ \` by
   definition\\
   $\llp z(x).\lb M_0\rb_z \rrp = \lambda x{:}\bang\llp \tau_0\rrp.\llet{\bang x = x}{\llp \lb M_0 \rb_z\rrp}$ \` by
   definition\\
   $=_\beta \llp \lambda x{:}\tau_0 . M_0 \rrp =
    \lambda x{:}\bang\llp\tau_0 \rrp.\llet{\bang x = x}{\llp
     M_0\rrp}$ \` by i.h. and congruence
    \end{tabbing}

  \item[Case:] Application

    \begin{tabbing}
    $\lb M_0 \, M_1 \rb_z = (\nub y)(\lb M_0 \rb_y \mid \ov{y}\langle x
    \rangle.(\bang x(w).\lb M_1 \rb_w \mid [y\leftrightarrow z]) $ \`
    by definition\\
    $\llp (\nub y)(\lb M_0 \rb_y \mid \ov{y}\langle x
    \rangle.(\bang x(w).\lb M_1 \rb_w \mid [y\leftrightarrow z]) \rrp
    = \llp  \ov{y}\langle x
    \rangle.(\bang x(w).\lb M_1 \rb_w \mid [y\leftrightarrow z])
    \rrp\{\llp \lb M_0\rb_y \rrp / y\}$\\ \` by definition\\
    $=\llp \lb M_0 \rb_y\rrp\, \bang \llp \rb M_1\rb_w\rrp $ \` by definition\\
    $=_\beta \llp M_0 \, M_1 \rrp = \llp M_0 \rrp \, \bang\llp
    M_1\rrp$ \` by i.h. and congruence
    \end{tabbing}
    
  \end{description}
\end{proof}

\begin{lemma}\label{lem:vals}
Let $\cdot \vdash M : \tau$ and $\cdot \vdash V : \tau$ with $V
\not{\tra{}_{\beta\eta}}$.
 $\lb M \rb_z \logsim \lb V \rb_z$ iff $\llp M\rrp \tra{}_{\beta\eta}^* \llp V\rrp$
\end{lemma}
\begin{proof}
~
\begin{tabbing}
{\bf $(\Leftarrow)$}\\
$\llp M\rrp \tra{}_{\beta\eta}^* \llp V\rrp$ \` by assumption\\
If $\llp M\rrp = \llp V\rrp$ then $\lb V \rb_z \logsim \lb V\rb_z$ \` by
reflexivity\\
If $\llp M\rrp \tra{}_{\beta\eta}^+ \llp V\rrp$ then $\lb M \rb_z \wtra{} P \logsim \lb V\rb_z$ \` by
Lemma~\ref{lem:opccomp}\\
$\lb M \rb_z \logsim \lb V \rb_z$ \` by closure under reduction\\
{\bf $(\Rightarrow)$}\\
$V =_\alpha \lambda x{:}\tau_0.V_0$ \` by inversion\\
$\llp V \rrp = \lambda x{:}\bang \llp \tau_0\rrp.\llet{\bang x =
  x}{\llp V_0\rrp}$ \` by definition\\
$\lb V \rb_z = z(x).\lb V_0\rb_z$ \` by definition\\
$M : \tau_0 \rightarrow \tau_1$ \` by inversion\\
$\llp  M\rrp \tra{}_{\beta\eta}^* V' \not\tra{}$ \` by strong normalisation\\
We proceed by induction on the length $n$ of the (strong) reduction:\\
{\bf Subcase:} $n = 0$\\
$\llp M \rrp = \lambda x{:}\tau_0.M_0$ \` by inversion\\
$M_0 = V_0$ \` by uniqueness of normal forms\\
{\bf Subcase:} $n = n' + 1$\\
$\llp M\rrp \tra{}_{\beta\eta} M'$ \` by assumption\\
$\lb M \rb_z \wtra{} P \logsim \lb M' \rb_z$ \` by
Lemma~\ref{lem:opccomp}\\
$\lb M' \rb_z \logsim \lb V \rb_z$ \` by closure under reduction\\
$\llp M' \rrp \tra{}_{\beta\eta}^* \llp V \rrp$ \` by i.h.\\
$\llp M \rrp \tra{}_{\beta\eta}^* \llp V \rrp$ \` by transitive closure
\end{tabbing}

\end{proof}

\thmfav*

We establish the proof of the two statements separately.
\begin{theorem}
Let $\cdot \vdash M : \tau$ and $\cdot \vdash N : \tau$. We have that
$ \llp M\rrp =_{\beta\eta} \llp N\rrp$ iff $\lb M \rb_z \logsim \lb N \rb_z$
\end{theorem}

\begin{proof}~

\begin{tabbing}
{\bf Completeness ($\Rightarrow$)}\\
$\llp M\rrp =_{\beta\eta} \llp N\rrp$ iff $\exists S . \llp M\rrp \tra{}_{\beta\eta}^* S$ and $\llp N\rrp\tra{}_{\beta\eta}^* S$ \\
Assume $\tra{}^*$ is of length $0$, then: $\llp M\rrp  =_\alpha \llp N\rrp $, $\lb M \rb_z
\equiv \lb N \rb_z$ and thus $\lb M \rb \logsim \lb N \rb_z$\\
Assume $\tra{}^+$ is of some length $ > 0$:\\
$\llp M\rrp  \tra{}_{\beta\eta}^+ S$ and $\llp N\rrp  \tra{}_{\beta\eta}^+ S$, for some $S$ \` by assumption\\
$\lb M \rb_z \tra{}^+ P \logsim \lb S \rb_z$ and
$\lb N \rb_z \tra{}^+ Q \logsim \lb S \rb_z$ \` by
Lemma~\ref{lem:opccomp}\\
$\lb M \rb_z \logsim \lb S \rb_z$ and $\lb N \rb_z \logsim \lb S
\rb_z$ \` by closure under reduction\\
$\lb M \rb_z \logsim \lb N \rb_z$ \` by transitivity\\
{\bf Soundness ($\Leftarrow$)}\\
$\lb M \rb_z \logsim \lb N \rb_z$ \` by assumption\\
Suffices to show: $\exists S. \llp M \rrp \tra{}_{\beta\eta}^* S$ and
$\llp N \rrp \tra{}_{\beta\eta}^* S$\\
$\llp N \rrp \tra{}_{\beta\eta}^* S' \not\tra{}$ \` by strong
normalisation\\
We proceed by induction on the length $n$ of the reduction:\\
{\bf Subcase:}  $n = 0$\\
$\lb M \rb_z \logsim \lb S'\rb_z$ \` by assumption\\
$\llp M \rrp \tra{}_{\beta\eta}^* \llp N \rrp$ \` by Lemma~\ref{lem:vals}\\
{\bf Subcase:} $n = n' + 1$\\
$\llp N \rrp \tra{}_{\beta\eta} S'$ \` by assumption\\
$\lb N \rb_z \tra{} P \logsim \lb S'\rb_z$ \` by
Theorem~\ref{lem:opccomp}\\
$\lb M \rb_z \logsim \lb S' \rb_z$ \` by closure under reduction\\
$\llp M \rrp =_{\beta\eta} \llp S' \rrp$ \` by i.h.\\
$\llp M \rrp =_{\beta\eta} \llp N \rrp$ \` by transitivity\\
\end{tabbing}
\end{proof}

\begin{theorem}
Let $\cdot \vdash P :: z{:}A$ and $\cdot \vdash Q :: z{:}A$. We have
that $\lb P\rb \logsim \lb Q\rb$ iff $\llp P \rrp =_{\beta\eta} \llp Q \rrp$
\end{theorem}

\begin{proof}
~
\begin{tabbing}
{\bf $(\Leftarrow)$}\\
Let $M = \llp P \rrp$ and $N =  \llp Q \rrp $:\\
$\lb M \rb_z \logsim \lb N \rb_z$ \` by Theorem~\ref{thm:fav}
(Completeness)\\
$\lb M \rb_z = \lb \llp P \rrp \rb_z \logsim \lb P\rb$ and 
$\lb N \rb_z = \lb \llp Q \rrp \rb_z \logsim \lb Q\rb$ \` by Theorem~\ref{thm:inv2} \\
$\lb P \rb \logsim \lb Q \rb$ \` by compatibility of logical
equivalence\\
{\bf $(\Rightarrow)$}\\

$\lb\llp P \rrp\rb_z \logsim \lb\llp Q  \rrp\rb_z$ \` by
Theorem~\ref{thm:inv1} and compatibility of logical equivalence\\
$\llp P \rrp =_{\beta\eta} \llp Q \rrp$ \` by Theorem~\ref{thm:fav} (Soundness)
\end{tabbing}
\end{proof}

\subsection{Proofs for \S~\ref{sec:hopi} -- Higher-Order Session Processes}

\begin{theorem}[Operational Soundness]
~\label{lem:hoopcsound}
\begin{enumerate}
\item If $\Psi \vdash M : \tau$ and $\lb M\rb_z \tra{} Q$ then
$M \tra{}^+ N$ such that $\lb N \rb_z \logsim Q$
\item If $\Psi ; \Ga ; \D \vdash P :: z{:}A$ and $\lb P \rb \tra{} Q$ then
$P \tra{}^+ P'$ such that $\lb P'\rb \logsim Q$
\end{enumerate}
\end{theorem}

\begin{proof}
  By induction on the given reduction.

  \begin{description}
  \item[Case:]
    $ (\nub x)(P_0 \mid \ov{x}\langle a_0\rangle.([a_0\leftrightarrow
    y_0] \mid \dots \mid x\langle a_n\rangle.([a_n \leftrightarrow
    y_n] \mid P_1) \dots )) \tra{} (\nub x)(P_0' \mid \ov{x}\langle
    a_0\rangle.([a_0\leftrightarrow y_0] \mid \dots \mid x\langle
    a_n\rangle.([a_n \leftrightarrow y_n] \mid P_1) \dots )) $
   
\begin{tabbing}
  $P = x\leftarrow M_0 \leftarrow \ov{y_i};P_2$ with $\lb M_0 \rb_x = P_0$
  and $\lb P_1 \rb = P_2$ \` by inversion\\
  $M_0 \tra{}^+ M_0'$ with $\lb M_0'\rb_x \logsim P_0'$ \` by i.h.\\
  $(x\leftarrow M_0 \leftarrow \ov{y_i};P_2) \tra{}^+ (x\leftarrow M_0'
  \leftarrow \ov{y_i};P_2)$ \` by 
  reduction semantics\\
  $\lb x\leftarrow M_0'  \leftarrow \ov{y};P_2\rb = (\nub x)(\lb M_0\rb_x \mid \ov{x}\langle
  a_0\rangle.([a_0\leftrightarrow   y_0] \mid   \dots   \mid   x\langle 
  a_n\rangle.([a_n   \leftrightarrow   y_n] \mid   P_1) \dots ))$ \\\`
  by   definition\\
  $\logsim (\nub x)(P_0' \mid \ov{x}\langle   a_0\rangle.([a_0\leftrightarrow
  y_0] \mid   \dots   \mid   x\langle   a_n\rangle.([a_n   \leftrightarrow
  y_n] \mid   P_1)$ \`   by congruence
\end{tabbing}

\item[Case:]  $ (\nub x)(x(a_0).\dots .x(a_n).P_0 \mid \ov{x}\langle
                                                             a_0\rangle.([a_0\leftrightarrow
                                                             y_0] \mid
                                                             \dots
                                                             \mid
                                                             x\langle
                                                             a_n\rangle.([a_n
                                                             \leftrightarrow
                                                             y_n] \mid
                                                             P_1)
                                                             \tra{}$

                                                             $ (\nub
                                                             x,a_0)(x(a_1).\dots
                                                             . x(a_n).P_0
                                                             \mid
                                                             [a_0\leftrightarrow
                                                             y_0] \mid
                                                             x\langle
                                                             a_1
                                                             \rangle.
                                                             ([a_1\leftrightarrow
                                                             y_1] \mid
                                                             \dots
                                                             \mid
                                                             x\langle
                                                             a_n\rangle.([a_n
                                                             \leftrightarrow
                                                             y_n] \mid
                                                             P_1) = Q
                                                             $

  \begin{tabbing}
  $P = x\leftarrow \{x\leftarrow P_2 \leftarrow \ov{a_i}\} \leftarrow
  \ov{y_i} ; P_3$
  with $\lb P_3 \rb = P_1$ and $\lb P_2 \rb = P_0$ \` by inversion\\
  $x\leftarrow \{x\leftarrow P_2 \leftarrow \ov{a_i}\} \leftarrow
  \ov{y_i} ; P_3 \tra{} (\nub x)(P_2\{\ov{y_i}/\ov{a_i}\} \mid P_3)$ \` by
  reduction semantics\\
  $Q \tra{}^+ (\nub x)(P_0\{\ov{y_i}/\ov{a_i}\} \mid P_1) = (\nub x)(\lb
  P_2\rb\{\ov{y_i}/\ov{a_i}\} \mid \lb P_3\rb )$ \` by reduction
  semantics and definition\\
 \end{tabbing}
                                                             
  \end{description}
\end{proof}

\begin{theorem}[Operational Completeness]
~\label{lem:hoopccomp}
\begin{enumerate}
\item If $\Psi \vdash M : \tau$ and $M \tra{} N$ then $\lb M \rb_z
  \wtra{} P$ such that $P \logsim \lb N\rb_z$
\item If $\Psi ; \Ga ; \D \vdash P :: z{:}A$ and $P \tra{} Q$ then
$\lb P \rb \tra{}^+ R$ with $R \logsim \lb Q \rb$
\end{enumerate}
\end{theorem}

\begin{proof}
By induction on the reduction semantics.
\begin{description}
\item[Case:] $x\leftarrow M \leftarrow \ov{y_i} ; Q \tra{}
  x\leftarrow M' \leftarrow \ov{y_i} ; Q$ from $M \tra{} M'$
  \begin{tabbing}
$\lb x\leftarrow M \leftarrow \ov{y_i} ; Q\rb =(\nub x)(\lb M \rb_x \mid \ov{x}\langle
                                                             a_0\rangle.([a_0\leftrightarrow
                                                             y_0] \mid
                                                             \dots
\mid x\langle a_n\rangle.([a_n \leftrightarrow y_n] \mid \lb Q \rb)
\dots ))  $
\\\` by definition\\
$\lb M \rb_x \wtra{} R_0$ with $R_0 \logsim \lb M'\rb_x$ \` by
i.h.\\
$\lb x\leftarrow M \leftarrow \ov{y_i} ; Q\rb \wtra{} (\nub x)(R_0 \mid \ov{x}\langle
                                                             a_0\rangle.([a_0\leftrightarrow
                                                             y_0] \mid
                                                             \dots
\mid x\langle a_n\rangle.([a_n \leftrightarrow y_n] \mid \lb Q \rb)
\dots ))  $
\\\` by reduction semantics\\
$\logsim \lb x\leftarrow M \leftarrow \ov{y_i} ; Q\rb =(\nub x)(\lb M \rb_x \mid \ov{x}\langle
                                                             a_0\rangle.([a_0\leftrightarrow
                                                             y_0] \mid
                                                             \dots
\mid x\langle a_n\rangle.([a_n \leftrightarrow y_n] \mid \lb Q \rb)
\dots ))  $
\\\` by congruence

\end{tabbing}

\item[Case:] $ x\leftarrow \{x\leftarrow P_0 \leftarrow \ov{w_i}\}\leftarrow
\ov{y_i};Q \tra{} (\nub x)(P_0\{\ov{y_i}/\ov{w_i}\} \mid Q)$

  \begin{tabbing}
$\lb x\leftarrow \{x\leftarrow P_0 \leftarrow \ov{w_i}\}\leftarrow
\ov{y_i};Q\rb
=$\\
$\qquad(\nub x)(x(w_0).\dots.x(w_n).\lb P_0\rb \mid  \ov{x}\langle
                                                             a_0\rangle.([a_0\leftrightarrow
                                                             y_0] \mid
                                                             \dots
\mid x\langle a_n\rangle.([a_n \leftrightarrow y_n] \mid \lb Q \rb)
\dots ))  $
\\\` by definition\\
$\tra{}^+ (\nub x)(\lb P_0\rb\{\ov{y_i}/\ov{w_i}\} \mid \lb Q\rb)$ \` by
reduction semantics\\
$\logsim (\nub x)(\lb P_0\{\ov{y_i}/\ov{w_i}\}\rb \mid \lb Q \rb)$
\end{tabbing}

\end{description}

\end{proof}

\begin{theorem}[Operational Soundness]
~
 \begin{enumerate}
  \item If $\Psi ; \Ga ; \D \vdash P :: z{:}A$ and $\llp P \rrp \tra{}
    M$ then $P \mapsto^* Q$ such that $M =_\alpha \llp Q \rrp$
  \item If $\Psi \vdash M : \tau$ and $\llp M \rrp \tra{} N$ then $M
    \tra{}_\beta^+ M'$ such that $N =_\alpha \llp M' \rrp$
  \end{enumerate}
\end{theorem}

\begin{proof}
By induction on the given reduction.

\begin{description}
\item[Case:] $\llp P_0 \rrp\{(\llp M \rrp \,\ov{y_i})/x\} \tra{}
  N\{(\llp M \rrp\, \ov{y_i})/x\}$

\begin{tabbing}
$P = x\leftarrow M \leftarrow \ov{y_i} ; P_0$  \` by inversion\\
$P_0 \mapsto^* R$ with $N =_\alpha \llp R \rrp$ \` by i.h.\\
$P \mapsto^* x \leftarrow M \leftarrow \ov{y_i} ; R$ \` by definition
of $\mapsto$ \\
$\llp x \leftarrow M \leftarrow \ov{y_i} ; R \rrp = \llp R \rrp\{(\llp
M \rrp\, \ov{y_i})/x\}$ \` by definition\\
$=_\alpha N\{(\llp M \rrp\, \ov{y_i})/x\}$ \` by congruence
\end{tabbing}

\item[Case:] $\llp P_0 \rrp\{(\llp M \rrp\, \ov{y_i})/x\} \tra{}
  \llp P_0 \rrp\{M'/x\}$

\begin{tabbing}
$P = x\leftarrow M \leftarrow \ov{y_i} ; P_0$  \` by inversion\\
{\bf Subcase:} $\llp M \rrp\, \ov{y_i} \tra{} N\,\ov{y_i}$\\
$M  \tra{}_\beta^+ M''$ with $N =_\alpha \llp M''\rrp$ \` by i.h.\\
$P \mapsto^+ x\leftarrow M'' \leftarrow \ov{y_i} ; P_0$ \` by reduction
semantics\\
$\llp x\leftarrow M'' \leftarrow \ov{y_i} ; P_0 \rrp = \llp P_0
\rrp\{(\llp M'' \rrp\, \ov{y_i})/x\}$ \` by definition\\
$=_\alpha\llp P_0 \rrp\{M'/x\}$ \` by congruence\\
{\bf Subcase:} $\llp M \rrp\, \ov{y_i} \tra{} (\lambda
y_1. \dots . y_n.M_0)\,y_1\,\dots\,y_n$\\
$M = \{ x\leftarrow Q \leftarrow \ov{y_i}\}$ with $\llp Q \rrp = M_0$ \` by inversion\\
$P = x\leftarrow \{ x\leftarrow Q \leftarrow \ov{y_i}\} \leftarrow
\ov{y_i} ; P_0$ \` by inversion\\
$P \tra{} (\nub x)(Q \mid P_0)$ \` by reduction semantics\\
$\llp (\nub x)(Q \mid P_0) \rrp = \llp P_0\rrp \{\llp Q\rrp / x\}$ \` by definition\\
$(\lambda y_1. \dots . y_n.M_0)\,y_1\,\dots\,y_n \tra{}^+ M_0$ \` by
operational semantics\\
\end{tabbing}
\end{description}

\end{proof}

\begin{theorem}[Operational Completeness]
~\label{thm:hooplcomp}
\begin{enumerate}
\item If $\Psi ; \Ga ; \D \vdash P :: z{:}A$ and $P \tra{} Q$ then
  $\llp P \rrp \tra{}_\beta ^* \llp Q \rrp$
\item If $\Psi \vdash M : \tau$ and $M \tra{} N$ then $\llp M \rrp
  \tra{}^+ \llp N \rrp$
\end{enumerate}
\end{theorem}

\begin{proof}
By induction on the given reduction

\begin{description}
\item[Case:] $(x \leftarrow M \leftarrow \ov{y_i} ; P_0) \tra{} (x
  \leftarrow M' \leftarrow \ov{y_i} ; P_0)$ with $M \tra{} M'$

\begin{tabbing}
$\llp x \leftarrow M \leftarrow \ov{y_i} ; P_0 \rrp
 = \llp P_0 \rrp\{\llp M \rrp\,\ov{y_i} / x\}$ \` by definition\\
$\llp M \rrp \tra{}^* \llp M' \rrp$ \` by i.h.\\
$\llp x \leftarrow M' \leftarrow \ov{y_i} ; P_0 \rrp
 = \llp P_0 \rrp\{\llp M' \rrp\,\ov{y_i} / x\}$ \` by definition\\
$\llp P_0 \rrp\{\llp M \rrp\,\ov{y_i} / x\} \tra{}_\beta^* \llp P_0
\rrp\{\llp M' \rrp\,\ov{y_i} / x\}$ \` by congruence\\
\end{tabbing}

\item[Case:] $(x \leftarrow \{x \leftarrow Q \leftarrow \ov{y_i}\}
  \leftarrow \ov{y_i} ; P_0) \tra{} (\nub x)(Q \mid P_0)$ 

\begin{tabbing}
$\llp x \leftarrow \{x \leftarrow Q \leftarrow \ov{y_i}\}
  \leftarrow \ov{y_i} ; P_0 \rrp = \llp P_0 \rrp\{((\lambda
  y_0.\dots.\lambda y_n. \llp Q \rrp)\, y_0\,\dots y_n)/x\} $ \` by definition\\
$\tra{}^+_\beta \llp P_0 \rrp\{\llp Q\rrp / x\}$ \` by congruence and
transitivity\\
$\llp (\nub x)(Q \mid P_0) \rrp = \llp P_0 \rrp\{\llp Q\rrp / x\}$ \`
by definition
\end{tabbing}
\end{description}
\end{proof}

\thminvencsho*

We prove each case as a separate theorem.

\begin{theorem}[Inverse Encodings -- Processes]
If $\Psi ; \Ga ; \D \vdash P :: z{:}A$
then $\lb\llp P \rrp\rb_z \logsim \lb P\rb$
\end{theorem}

\begin{proof}
By induction on the given typing derivation. We show the new cases.
\begin{description}
\item[Case:] Rule $\{\}E$

\begin{tabbing}
$P = x\leftarrow M \leftarrow \ov{y};Q$ \` by inversion\\
$\llp P \rrp = \llp Q \rrp \{(\llp M \rrp\,\ov{y})/x\}$ \` by definition\\
$\lb \llp Q \rrp \{(\llp M \rrp\,\ov{y})/x\}\rb_z = (\nub a)(\lb \llp
M \rrp\,\ov{y} \rb_a \mid \lb \llp Q \rrp\rb_z\{a/x\})$ \` by
Lemma~\ref{lem:compos}\\
$= (\nub a,x)(\lb\llp M\rrp \rb_x \mid \ov{x}\langle a_0\rangle.([a_0\leftrightarrow
                                                             y_0] \mid
                                                             \dots
\mid x\langle a_n\rangle.([a_n \leftrightarrow y_n] \mid \lb \llp Q\rrp
\rb\{a/x\}) \dots ))$ \` by definition\\
$\equiv  (\nub x)(\lb\llp M\rrp \rb_x \mid \ov{x}\langle a_0\rangle.([a_0\leftrightarrow
                                                             y_0] \mid
                                                             \dots
\mid x\langle a_n\rangle.([a_n \leftrightarrow y_n] \mid \lb \llp Q\rrp
\rb) \dots ))$ \\
$\lb P \rb =  (\nub x)(\lb M \rb_x \mid \ov{x}\langle a_0\rangle.([a_0\leftrightarrow
                                                             y_0] \mid
                                                             \dots
\mid x\langle a_n\rangle.([a_n \leftrightarrow y_n] \mid \lb Q \rb)
\dots )) $
\` by definition\\
$\logsim (\nub x)(\lb\llp M\rrp \rb_x \mid \ov{x}\langle a_0\rangle.([a_0\leftrightarrow
                                                             y_0] \mid
                                                             \dots
\mid x\langle a_n\rangle.([a_n \leftrightarrow y_n] \mid \lb \llp Q\rrp
\rb) \dots ))$ \` by i.h.
\end{tabbing}
\end{description}
\end{proof}

\begin{theorem}[Inverse Encodings -- $\lambda$-terms]
\label{app:inv4}
If $\Psi \vdash M : \tau$ then $\llp\lb M \rb_z\rrp =_\beta \llp M \rrp$
\end{theorem}

\begin{proof}
By induction on the given typing derivation. We show the new cases.

\begin{description}
\item[Case:] Rule $\{\}I$

\begin{tabbing}
$M = \{x \leftarrow P \leftarrow \ov{y_i}\}$ \` by inversion\\
$\lb M\rb_z =  z(y_0).\dots.z(y_n).\lb P\{z/x\} \rb$ \` by definition\\
$\llp z(y_0).\dots.z(y_n).\lb P\{z/x\} \rb \rrp = \lambda
y_0.\dots.\lambda y_n.\llp \lb P\{z/x\} \rb \rrp $ \` by definition\\
$\lb M \rb = \lambda y_0.\dots.\lambda y_n.\llp P \rrp$ \` by
definition\\
$=_\beta \lambda y_0.\dots.\lambda y_n.\llp \lb P\{z/x\} \rb \rrp$ \`
by i.h.
\end{tabbing}
\end{description}
\end{proof}

\subsection{Strong Normalisation for Higher-Order Sessions}\label{app:snho}

We modify the encoding from processes to $\lambda$-terms by
considering the encoding of derivations ending with the $\cpy$ rule as
follows (we write $\llp {-} \rrp^+$ for this revised encoding):
\[
\llp (\nub x)u\langle x\rangle.P \rrp^+ \triangleq \llet{\one = \munit}{\llp P\rrp^+\{u/x\}}
\]
All other cases are as before. It is immediate that the revised
encoding preserves typing. We now revisit operational completeness as:

\begin{lemma}[Operational Completeness]\label{app:opcompsn}
If $\Psi ; \Ga ; \D \vdash P :: z{:}A$ and $P \tra{} Q$ then $\llp P
\rrp^+ \tra{}^+_\beta \llp Q \rrp^+$
\end{lemma}
\begin{proof}
~

\begin{description}
\item[Case:] $(\nub u)(\bang u(x).P_0 \mid \ov{u}\langle x
  \rangle.P_1) \tra{} (\nub u)(\bang u(x).P_0  \mid (\nub x)(P_0 \mid
  P_1))$
\begin{tabbing}
$\llp (\nub u)(\bang u(x).P_0 \mid \ov{u}\langle x
  \rangle.P_1) \rrp^+ = \llet{\one = \munit}{\llp P_1\rrp^+\{u/x\}\{\llp
  P_0\rrp^+ / u\}}$\\
$ = \llet{\one = \munit}{\llp P_1\rrp^+\{\llp
  P_0\rrp^+ / x\}\{\llp
  P_0\rrp^+ / u\}}$ \` by definition\\
$\tra{} \llp P_1\rrp^+\{\llp   P_0\rrp^+ / x\}\{\llp
  P_0\rrp^+/u\}$ \` by operational
semantics\\
$\llp (\nub u)(\bang u(x).P_0  \mid (\nub x)(P_0 \mid
  P_1)) \rrp^+ = \llp P_1 \rrp^+\{\llp P_0\rrp^+/x\}\{\llp
  P_0\rrp^+/u\}$ \` by definition
\end{tabbing}
Other cases are unchanged.
\end{description}
\end{proof}

We remark that with this revised encoding, operational soundness
becomes:
\begin{lemma}
If $\Psi ; \Ga ; \D \vdash P :: z{:}A$ and $\llp P \rrp^+ \tra{} M$
then $P \mapsto^* Q$ such that $\llp Q\rrp \tra{}^* M$.
\end{lemma}
\begin{proof}
~
\begin{description}
\item[Case:] $\llp P \rrp^+ = \llet{\one = \munit}{\llp
    P_0\rrp^+\{u/x\}}$ with $\llp P \rrp^+ \tra{}\llp
    P_0\rrp^+\{u/x\}$
\begin{tabbing}
$\llp P \rrp^+ =  \llet{\one = \munit}{\llp
    P_0\rrp^+\{u/x\}} \tra{} \llp P_0\rrp^+\{u/x\}$ \` by operational
  semantics, as needed.
\end{tabbing}
Remaining cases are fundamentally unchanged.
\end{description}
\end{proof}

The revised encoding remains mutually inverse with the $\lb{-}\rb_z$
encoding. We show only the relevant new cases.

\begin{lemma}
If $\Psi ; \Ga ; \D \vdash P :: z{:}A$ then $\lb \llp P \rrp^+\rb_z
\logsim \lb P \rb$
\end{lemma}
\begin{proof}
~

\begin{description}
\item[Case:] $\cpy$ rule

\begin{tabbing}
$\llp P \rrp^+ = \llet{\one = \munit}{\llp P_0\rrp^+\{u/x\}}$ \` by
definition\\
$\lb \llet{\one = \munit}{\llp P_0\rrp^+\{u/x\}} \rb_z = (\nub y)(\zero
\mid \lb \llp P_0\rrp^+\{u/x\} \rb_z)$ \` by definition\\
$\equiv \lb \llp P_0\rrp^+\{u/x\} \rb_z$ \` by structural congruence\\
$\logsim (\nub x)(\ov{u}\langle w \rangle.[w\leftrightarrow x] \mid
\lb \llp P_0 \rrp^+\rb_z)$ \` by compositionality\\
$\logsim \lb P\rb $ \` by i.h. + congruence + definition of $\logsim$ for open processes
\end{tabbing}
\end{description}
\end{proof}

\begin{lemma}
If $\Psi \vdash M : \tau$ then $\llp \lb M \rb_z \rrp^+ =_\beta \llp M \rrp^+$
\end{lemma}
\begin{proof}
~

\begin{description}
\item[Case:] $\m{uvar}$ rule

\begin{tabbing}
$\lb u \rb_z = (\nub x)u\langle x \rangle.[x\leftrightarrow z]$ \` by
definition\\
$\llp (\nub x)u\langle x \rangle.[x\leftrightarrow z]\rrp^+ = 
\llet{\one = \munit}{u} =_\beta u$\\

\end{tabbing}

\end{description}

\end{proof}

\thmsn*

\begin{proof}
The result follows from the operational completeness result above
(Lemma \ref{app:opcompsn}), which requires every process reduction to
be matched with one or more reductions in the $\lambda$-calculus.
We can thus prove our result via strong normalisation of
$\tra{}_\beta$: Assume an infinite reduction
  sequence $P \tra{} P' \tra{} P'' \tra{} \dots$, by completeness this
  implies that there must exist an infinite sequence $\llp P \rrp
  \tra{}^+_\beta \llp P' \rrp \tra{}_\beta^+ \llp P'' \rrp
  \tra{}^+_\beta \dots$, deriving a contradiction.
\end{proof}


\end{document}